\documentclass[a4paper,11pt]{article}
\pdfoutput=1 

\usepackage{jhep_like} 

\usepackage{caption}
\usepackage{subcaption}
\usepackage{graphicx}

\newcommand{\HL}[1]{\textcolor{magenta}{#1}}
\newcommand{\ZS}[1]{\textcolor{blue}{\it [ZS: #1]}}

\newcommand{\newtext}[1]{\textcolor{blue}{#1}}

\newcommand{\SV}[1]{\textcolor{red}{#1}}

\usepackage{tensor}
\usepackage{braket}
\usepackage{comment}

\usepackage{dsfont}
\usepackage{bbm}
\usepackage{tikz}
\usepackage{upgreek}

\usepackage{mathtools}

\usepackage[normalem]{ulem}
\usepackage{amsmath}
\usepackage{enumerate}
\usepackage{epsfig}
\usepackage{mathbbol}

\def\Im{\mathop{\rm Im} }

\newcommand\half{{\ensuremath{\frac{1}{2}}}}

\newcommand\vev[1]{{\ensuremath{\left\langle{#1}\right\rangle}}}

\newcommand\myeq{\stackrel{\mathclap{\normalfont\tiny\mbox{uncorrelated}}}{=}}

\newcommand{\be}{\begin{equation}}
\newcommand{\ee}{\end{equation}}
\newcommand{\bea}{\begin{eqnarray}}
\newcommand{\eea}{\end{eqnarray}}
\newcommand{\bega}{\begin{gather}}
\newcommand{\eega}{\end{gather}}
\newcommand{\nn}{\nonumber\\}
\newcommand{\bi}{\begin{itemize}}
\newcommand{\ei}{\end{itemize}}
\newcommand{\ben}{\begin{enumerate}}
\newcommand{\een}{\end{enumerate}}
\newcommand{\bca}{\begin{cases}}
\newcommand{\eca}{\end{cases}}
\newcommand{\bln}{\begin{align}}
\newcommand{\eln}{\end{align}}
\newcommand{\bst}{\begin{split}}
\newcommand{\est}{\end{split}}
\def\ie{\begin{equation}\begin{aligned}}
\def\fe{\end{aligned}\end{equation}}
\newcommand{\bma}{\le(\begin{matrix}}
\newcommand{\ema}{\end{matrix}\ri)}
\newcommand{\bwt}{\begin{widetext}}
\newcommand{\ewt}{\end{widetext}}

\def\b{{\beta}}

\newcommand\sig{\sigma}

\newcommand\om{\omega}

\newcommand\Ga{{\ensuremath{{\Gamma}}}}
\newcommand\de{{\ensuremath{{\delta}}}}
\newcommand\De{{\ensuremath{{\Delta}}}}

\def\th{{\theta}}

\newcommand\ov{\over}
\newcommand\ha{{\half}}

\def\le{\left}
\def\ri{\right}

\newcommand\sA{{\ensuremath{{\mathcal A}}}}

\newcommand\sH{{\ensuremath{{\mathcal H}}}}

\newcommand\sL{{\ensuremath{{\mathcal L}}}}

\newcommand\sN{{\ensuremath{{\mathcal N}}}}
\newcommand\sO{{\ensuremath{{\mathcal O}}}}

\newcommand\sS{{\mathcal S}}

\newcommand\sZ{{\mathcal Z}}

\newcommand{\Tr}{\text{Tr}}

\newcommand{\eij}{\overline{E}_{ij}}
\usepackage{amsthm}

\newtheorem{definition}{Definition}[section]

\newtheorem{claim}[definition]{Claim}

\begin{document}

\title{Local dynamics and the structure of chaotic eigenstates}

\author[a,*]{ Zhengyan Darius Shi,} 
\author[a,b,*]{ Shreya Vardhan,}
\author[a]{and Hong Liu}
\affiliation[a]{Center for Theoretical Physics, 
Massachusetts Institute of Technology, Cambridge, MA 02139}
\affiliation[b]{Stanford Institute for Theoretical Physics, Stanford University, Stanford, CA 94305}
 \note[*]{These authors contributed equally to the work. }

\emailAdd{zdshi@mit.edu}
\emailAdd{vardhan@stanford.edu}
\emailAdd{hong\_liu@mit.edu}

\preprint{MIT-CTP/5567}

\abstract{We identify new universal properties of the energy eigenstates of chaotic systems with local interactions, which distinguish them both from integrable systems and from non-local chaotic systems. 
We
study the relation between the energy eigenstates of the full system and products of energy eigenstates of two extensive subsystems, using  a family of spin chains in (1+1) dimensions as an illustration. The magnitudes of the coefficients relating the two bases have a simple universal form as a function of $\omega$, the energy difference between the full system eigenstate and the product of eigenstates. This form explains the exponential decay with time of the probability for a product of eigenstates to return to itself during thermalization. 
We also find certain new  statistical properties of the coefficients.
 While it is generally expected that the coefficients are uncorrelated random variables, we point out that correlations  implied by unitarity are important for understanding the transition probability between two products of  eigenstates, and the evolution of operator expectation values during thermalization. Moreover, we find  
that there are additional correlations resulting from locality, 
which lead to a slower growth of the second Renyi entropy than the one predicted by an uncorrelated random variable approximation. 
}

\maketitle


\section{Introduction}\label{sec:intro}


Characterizing universal properties of chaotic quantum many-body systems is central to many areas of physics, ranging from condensed matter to high-energy physics. A number of universal features are known for the spectrum and energy eigenstates. 
Early works concentrated on the distribution of energy eigenvalues, whose local statistics have widely been observed to obey the eigenvalue statistics of Gaussian random matrices \cite{wigner1951,mehta2004random,Dyson1962es, bohigas, berry_tabor, macdonald, berry,rabson2004_wignerdyson,kudo2005_wignerdyson,santos2010_wignerdyson}. 
Physical properties of individual energy eigenstates are elucidated by 
the eigenstate thermalization hypothesis (ETH)~\cite{deutsch_1991_ETH,srednicki_1994_ETH} and subsystem ETH~\cite{dymarsky_2018_subsystemETH, grover_eig}, according to which a system in a generic energy eigenstate behaves macroscopically like a thermal system. 

The spectral statistics and ETH properties 
 do not  distinguish whether a system is local or not. However, the locality of interactions in realistic physical systems  has important consequences.  For example, there are bounds on the speed at which local few-body operators grow or the rate at which  quantum information spreads between spatial regions~\cite{lieb_1972_LRbound, bravyi_2007_boundent_bipartite, acoleyen_2013_entrates_arealaw,marien_2016_entbound,vershynina_2019_renyibound, shi2022_renyibounds}.  Since the dynamics of a system are fully determined by the energy eigenstates and eigenvalues, locality must be encoded in the spectral properties. On the other hand, since the projector onto an  energy eigenstate 
is a highly non-local operator, any imprints of locality on it must be subtle. We identify examples of such imprints in this paper.
 


One can probe constraints from local 
interactions on the structure of energy eigenstates by examining the behavior of several dynamical quantities. 
Consider, for example, the evolution of the Renyi entropies of a subsystem $A$ during the equilibration process of some non-equilibrium state $\ket{\psi}$
expressed in terms of the energy eigenbasis $\{\ket{a}\}$,\footnote{The $n=1$ case is understood as a limit and gives the von Neumann entropy.} 
\bega \label{ren1}
S_{n,A}(t) = - \frac{1}{n-1} \log \text{Tr}[\rho_A(t)]^n, \quad n = 1, 2, ... \\
 \label{ren2}
 \rho_A(t) = \text{Tr}_B\le[e^{-iHt}\ket{\psi} \bra{\psi} e^{iHt}\ri] = \sum_{a,b} e^{- i (E_a -E_b) t} c^a_\psi c^{b*}_\psi  \text{Tr}_B (\ket{a} \bra{b}) , \\ 
 \ket{\psi} = \sum_a c^a_\psi \ket{a}  \ . \label{cp}
\end{gather} 
 For local systems, $S_{n,A}(t)$ is expected to exhibit linear growth in $t$ \cite{calabresecardy_2005_evolent_1d, kimhuse_2013_entballistic,liusuh_2014_entTsunami}, which can be used to constrain  the collective behavior of coefficients $c^a_\psi$. 

A particularly useful setup to probe imprints of local dynamics is to divide the system into two extensive subsystems $A$ and $B$, and take $\ket{\psi} = \ket{i}_A \ket{j}_B$ to be a product of energy eigenstates $\ket{i}_A$ and $\ket{j}_B$ of $A$ and $B$. Below we abbreviate $\ket{ij} \equiv \ket{i}_A\ket{j}_B$. In this case, the  equilibration of $\ket{ij}$ originates solely from local interactions at the interface of $A$ and $B$, and it may be expected that sharper statements can be made about the collective behavior of the corresponding coefficients $c^a_{ij}$. 
For definiteness, we will consider a one-dimensional spin chain whose Hamiltonian $H$ can be decomposed as 
$H = H_A + H_B + H_{AB}$, with $H_{AB}$ the local interaction between $A$ and $B$. See Fig. \ref{fig:hab_local}. 

\begin{figure}[!h]
    \centering
    \includegraphics[width=8cm]{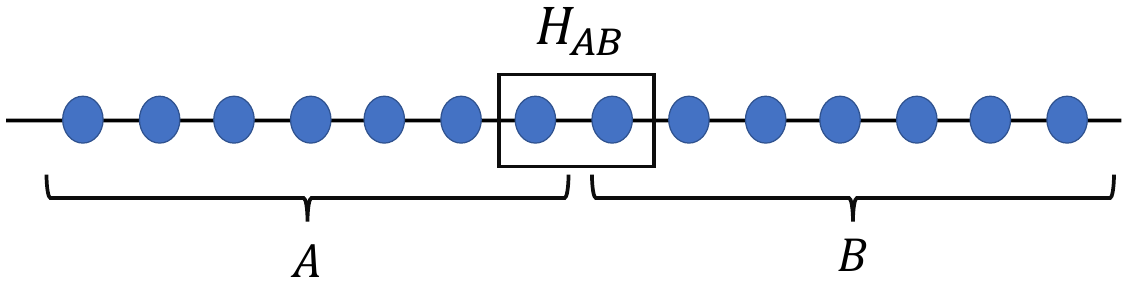}
    \caption{In (1+1) dimensions, the interaction term $H_{AB}$ is supported on an $\sO(1)$ number of sites across the boundary between $A$ and $B$.}
    \label{fig:hab_local}
\end{figure}


We can consider a few interesting dynamical quantities during the approach to equilibrium in this context, in addition to \eqref{ren1}. One such quantity is 
the return probability as a function of time, 
 \be 
P(t) = |\bra{ij} e^{-i H t} \ket{i j}|^2 =  \sum_{a, b} |c^a_{ij}|^2 |c^b_{ij}|^2 e^{- i (E_a-E_b) t}
 \, . \label{ptdef}
\ee 
For a generic state $\ket{\psi}$, the return probability $\vev{\psi|e^{- i H t} |\psi}$ is exponentially small in the system size $L$ at any time $t$ of order $\sO(L^0)$. For example, for a product state between all sites of the system, the return probability shows a Gaussian decay of the form \cite{torres}
\be \label{prod_return}
\vev{\psi|e^{- i H t} |\psi} = c \,  e^{- \Gamma_{\rm prod}^2t^2}
\ee
where $c$ is an $\sO(1)$ constant, and 
\be 
\Gamma_{\rm prod}^2 = \braket{\psi | H^2 |\psi} = \sO(L) \, . 
\ee
If we ignore exponentially small values in the thermodynamic limit, this quantity therefore does not show a non-trivial evolution at any $\sO(1)$ or larger time scale.~\footnote{If we do consider values of $\braket{\psi|e^{-iHt}|\psi}$ that are exponentially small in $L$, then then there are some interesting features of the time-evolution at long time scales which are studied in \cite{torres}.}
However, for states of the form $\ket{i}_A \ket{j}_B$, since $H$ differs from $H_A + H_B$ only by a local term, we may expect  \eqref{ptdef}  to have non-trivial dynamics for $\sO(1)$ times. By studying this quantity in local chaotic spin chains in this paper, we find that this is indeed the case. As we discuss in Sec.~\ref{sec:eig}, $P(t)$ exhibits exponential decay with an $\sO(1)$ decay time scale $1/\Gamma$. Due to the one-dimensional nature and translation-invariance of the system, $\Gamma$ can be interpreted as an emergent energy scale characterizing the local dynamics of the system. 


Another quantity we consider is the time evolution of transition probability from one product of eigenstates to another, 
\bega \label{tra}
P_{ij, xy}(t)= \le|\vev{xy|e^{- i Ht} |ij} \ri|^2 =   \sum_{a,b} e^{- i (E_a - E_b) t} c_{ij}^a c_{xy}^{a*} c_{ij}^{b*} c_{xy}^b \ . 
\end{gather}
Again, since the time-evolution is governed by  $H_{AB}$, this quantity turns out to have a non-trivial evolution in the thermodynamic limit, as we discuss in Sec.~\ref{sec:phco}. 
From~\eqref{ptdef} to~\eqref{tra} to~\eqref{ren2}, the dynamical  quantities defined above probe increasingly more subtle correlations along $c^a_{ij}$: 
$P(t)$ probes only the collective behavior of $|c^a_{ij}|$, while the transition probabilities~\eqref{tra} are also sensitive to the phases of $c^a_{ij}$ and probe correlations among four of them. The Renyi entropies  $S_{n,A}(t)$  probe correlations among at least eight such coefficients (for $n=2$). 


The structure of $c^a_\psi$ has been discussed in earlier works in different 
contexts for different types of $\ket{\psi}$~\cite{goldstein2006_EBansatz,popescu2006_EBansatz,GroverLu_2019, murthy_2019_structure, reimann_fast, centrallim, israelev, deutsch_1991_ETH, deutsch_3, deutsch_therment}. 
A common ansatz is that $c^a_\psi$ can be treated as {\it independent and identically  distributed (iid) complex Gaussian random variables}.  
In particular, for $c^a_{ij}$, the ``ergodic bipartition (EB)'' ansatz of~\cite{GroverLu_2019} states that 
 \be 
\overline{c^a_{ij}}=0, \quad  \overline{c^a_{ij}  c^{b\ast}_{kl}} = \bca \frac{1}{\sN} \delta_{ab} \delta_{ik} \delta_{jl}  &
|E_a - E_{A,i} - E_{B,j}|\leq \Ga \cr
0 & |E_a - E_{A,i} - E_{B,j}| > \Ga
\eca
\label{erg}
 \ee
 where $\sN$ is a normalization constant and $\Ga$ is an energy scale which does not scale with the volume of the system.  
The EB ansatz leads to simple analytic expressions for the $n$-th Renyi entropies in an energy eigenstate $\ket{a}$ of the full system, which have been confirmed numerically in chaotic spin chains~\cite{GroverLu_2019}. 


However, on applying the EB ansatz to $P(t)$, $P_{ij, xy} (t)$ and $S_{n,A}(t)$ (for $\ket{\psi} = \ket{i}_A \ket{j}_B$), we find unphysical predictions for all three quantities. See Fig.~\ref{fig:S2_growth} for a cartoon comparison for $S_2 (t)$. 
The evolution from the EB ansatz shows large oscillations, and first reaches the equilibrium value at a time scale which is independent of the volume of the system. This is incompatible with locality, as it has been shown~\cite{Bravyi_2006_LRbounds_toporder,bravyi_2007_boundent_bipartite,acoleyen_2013_entrates_arealaw,marien_2016_entbound,vershynina_2019_renyibound,shi2022_renyibounds} that in a system with local interactions, 
the Renyi entropies can grow at most linearly with time. 
Similarly, applying the EB ansatz to the calculation of $P(t)$ leads to large oscillations and is not compatible with exponential decay observed in local systems. For $P_{ij,xy}(t)$, \eqref{erg} predicts that this quantity reaches its saturation for any $\sO(1)$ time scale, in contrast to the non-trivial evolution at $\sO(1)$ times that we observe in Sec.~\ref{sec:phco} below.


\begin{figure}[!h]
    \centering
    \includegraphics[width = 0.7\textwidth]{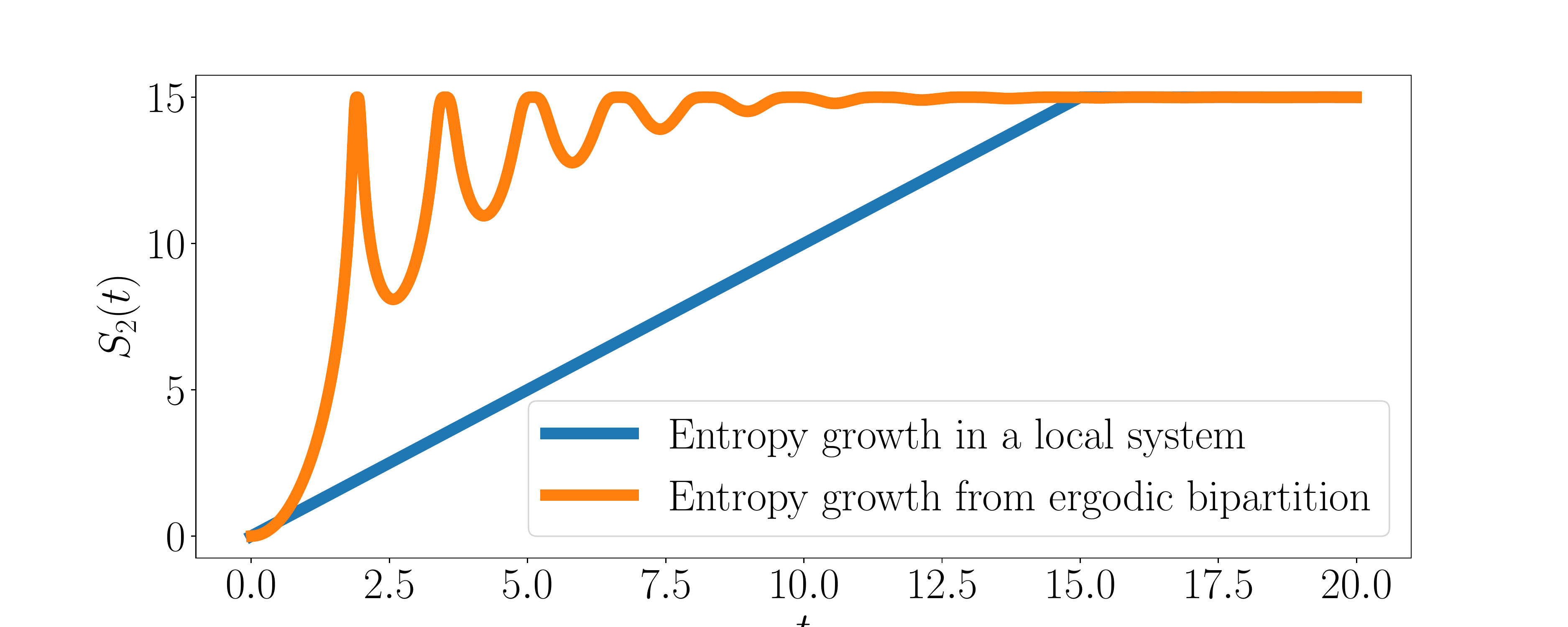}
    \caption{A cartoon comparison between the behavior of $S_{2,A}(t)$ for an extensive subsystem $A$ in a chaotic spin chain and the behavior implied by the EB ansatz. From  the EB ansatz, the $\sO(L)$ saturation value is first reached at an $\sO(1)$ time, which is inconsistent with locality.}
    \label{fig:S2_growth}
\end{figure}

In this paper, we make a number of observations about the universal collective properties of the coefficients $c^a_{ij}$ from studies of a class of chaotic spin chain systems. These properties help explain the evolution of the above dynamical quantities during thermalization. We summarize our main results below: 

\ben 

\item  We find that $|c^a_{ij}|^2$ can be approximated in a simple universal form as a function of $\omega \equiv E_a - \bar E_{ij}$, where $ \bar E_{ij}=\vev{ij|H|ij} $. 
The function is given by 
\be \label{fsum}
f(\omega) = e^{S\le(\frac{E_a + \bar{E}_{ij}}{2}\ri)} |c^a_{ij}|^2= 
\begin{cases} 
\frac{C_0}{\omega^2 + \Gamma^2} & |\omega| \lesssim \sigma \\
C_1 e^{-\lambda |\omega|} & |\omega| \gtrsim \sigma   
\end{cases}  \, \, \, . 
\ee 
where $S(E)$ is the thermodynamic entropy, $\sigma> \Gamma$, and  in 1+1 dimensions, all constants appearing in the above expression are $\sO(1)$. We provide a general analytic argument for this functional form with inputs from random matrix theory, which also applies in higher dimensions. In higher dimensions, $\Gamma$ and $\sigma$ both scale as some power of the area of the boundary of $A$. 

The time scale for the  exponential decay of $P(t)$ at intermediate times is  $1/\Gamma$. 

\item   To explain the evolution of the transition probability $P_{ij,xy}(t)$  using 
\eqref{tra}, we propose the following minimal model for the average of four $c^a_{ij}$. Here $m_1$ should be seen as shorthand for $i_1, j_1$, and the overline denotes averaging over an $\sO(1)$ number of energy levels close to $a_1$, $i_1$, $i_2$, and other indices.  
\begin{align} \label{c2_intro} 
\overline{c^{a_1}_{m_1} {c^{a_2}_{m_2}}^{\ast} c^{b_1}_{n_1} {c^{b_2}_{n_2}}^{\ast}} &= \overline{ |c^{a_1}_{m_1}|^2} \, \overline{ |c^{b_1}_{n_1}|^2 } (\delta_{a_1 a_2} \delta_{b_1b_2} \delta_{m_1m_2} \delta_{n_1 n_2} + \delta_{a_1 b_2} \delta_{b_1a_2} \delta_{m_1n_2} \delta_{n_1 m_2} )\cr
& - k_{a_1, b_1, m_1, n_1} (\delta_{a_1 a_2} \delta_{b_1 b_2} \delta_{m_1n_2}\delta_{n_1m_2} + \delta_{a_1 b_2} \delta_{b_1 a_2} \delta_{m_1m_2}\delta_{n_1n_2} ) \, \nn
& + ...
\end{align}
Here $k_{a, b, m, n}$ is $\sO(e^{-3S})$, and can be parameterized by a smooth function $g(\omega_1, \omega_2, \omega_3)$ of the  energy differences $\omega_1=E_a - E_b$, $\omega_2=\bar E_m - \bar E_n$, $\omega_3=(E_a+E_b)- (\bar E_m +\bar E_n)$.

If we assumed that the $c^a_{ij}$ were iid Gaussian random variables, we would only have the first line in the above expression. This assumption would lead to the unphysical prediction that $P_{ij,xy}(t)$ is equal to its saturation value for any $\sO(1)$ time.

The correlations in the second line are needed to ensure that $\overline{|\braket{m|n}|^2} = \overline{|\braket{a|b}|^2}=0$ for $a \neq b$ and $m \neq n$. They are also sufficient to describe the evolution of the transition probability at leading order in $e^{-S}$, which depends on both $f(\omega)$ and $g(\omega_1, \omega_2, \omega_3)$.


\item
The average in the first two lines of \eqref{c2_intro} is non-zero only for a particular set of index structures, which is motivated by averages over  random unitary matrices. We find for a local chaotic spin chain that the average is also non-zero at $\sO(e^{-3S})$ for certain  other index structures not explicitly included in \eqref{c2_intro}. While these additional correlations do not play a role in the evolution of $P(t)$ and $P_{ij,xy}(t)$, they explain the evolution of 
\be C^{(2)}(x_1y_1, ij, ij, x_2y_1; t) = \overline{\braket{x_1 y_1| e^{-iHt}|ij} \braket{ij| e^{iHt}|x_2 y_1} }, \quad x_1 \neq x_2 \, .  
\ee
If there were no further contributions to \eqref{c2_intro} at $\sO(e^{-3S})$, we would predict that this quantity is zero at $\sO(e^{-S})$. 
We find numerically that it is non-zero at $\sO(e^{-S})$ and has a non-trivial time-evolution at $\sO(1)$ times.  

\item For the average of eight factors of $c^a_{ij}$, we again have correlations beyond the minimal ones required by the unitarity of the change of basis from $\{\ket{a}\}$ to $\{\ket{ij}\}$. Such correlations play a dominant role in the evolution of the second Renyi entropy $S_{2,A}(t)$. As a result of these correlations, the rate of entanglement growth is slower than it would be on assuming the generalization of the first two lines of  \eqref{c2_intro} for  eight factors of $c^a_{ij}$.   Since the rate of entanglement growth is constrained by the locality of interactions, these unexpected  correlations are  important imprints of locality on the eigenstates. 


\item The correlations implied by unitarity also play a role in the evolution of expectation values of local operators in the state $\ket{ij}$. We find a simple expression for $\braket{ij|W(t)|ij}$ for a local operator $W$ in terms of the functions $f(\omega)$ and $g(\omega_1, \omega_2, \omega_3)$. We also discuss the interplay of the properties of $c^a_{ij}$ with the eigenstate thermalization hypothesis. 

 \een

The plan of the paper is as follows. We introduce the setup and model in Sec.~\ref{sec:setup}. In Sec.~\ref{sec:eig}, we discuss the structure of the absolute values $|c^{a}_{ij}|$ and how it explains the  evolution of the return probability. In Sec.~\ref{sec:phco}, we discuss the phase correlations among the $c^a_{ij}$ needed to explain the time-evolution of the transition probability. We study the evolution of the second Renyi entropy and the further phase correlations implied by it in Sec.~\ref{sec:ent}. In Section \ref{sec:corr}, we discuss the evolution of correlation functions and the interplay with the eigenstate thermalization hypothesis.  We discuss some future directions  in Sec.~\ref{sec:disc}. A detailed comparison to an earlier ansatz of \cite{murthy_2019_structure} is discussed in Appendix \ref{app:MSlimitation}. Appendices \ref{app:therm} and \ref{app:EDF_derivation_details} justify some of the approximations used in the main text.

\section{Setup}
\label{sec:setup}

In this paper, we will use as an illustration the following family of spin chain models in (1+1) dimensions: 
\begin{equation}
    H = \sum_{i=1}^{L-1} S^z_i S^z_{i+1} + h \sum_{i=1}^L S^z_i + g \sum_{i=1}^L S^x_i \, .\label{ham} 
\end{equation}
With a choice of the parameter $g$ (we will use $g = -1.05$), the system is integrable for $h=0$, and chaotic for sufficiently large $h$. 
We are interested in the thermodynamic limit with the number of lattice sites $L$ going to infinity. 
The energy eigenstates of $H$ are denoted as $\ket{a}$ with eigenvalues $E_a$. 
We use $dE e^{S(E)}$ to represent the number of energy eigenstates of $H$ within the energy interval $[E, E + dE]$. 
In the thermodynamic limit, we have $E_a , S(E) \propto L$, with $E_a$ going from $-\infty$ to $+\infty$, and 
the maximal value of $S(E)$ occurring at $E=0$. We can associate an inverse temperature $\beta (E)$ to energy $E$ by 
\be 
\beta (E) = \frac{d S (E)}{d E} ,\quad \beta (E+ \delta E) = \beta (E) +O(1/L), \quad \text{for} \quad \delta E \sim O(L^0) \ .
\ee
While our explicit numerical results will be restricted to~\eqref{ham}, we expect 
our conclusions may be applicable to general chaotic systems including those in higher dimensions.

Now divide the system into two extensive subsystems, $A$ and $B$. The total Hamiltonian can then be split as 
\be 
H = H_A + H_B + H_{AB}
\ee
where $H_A$ and $H_B$ are supported only on $A$ and $B$ respectively, and  $H_{AB}$ is a local interaction term supported on the boundary between $A$ and $B$. Let $\ket{i}_A$, $\ket{j}_B$ denote the eigenstates of $H_A$, and $H_B$ respectively, with 
eigenvalues $E_{A,i}$ and $E_{B,j}$. 
The expectation value of $H$ in $\ket{ij}\equiv\ket{i}_A \ket{j}_B$ is given by 
\be \label{mean}
\overline E_{ij} \equiv \bra{ij} H \ket{ij} = E_{A, i} + E_{B, j}+\Delta_{ij}, \quad \Delta_{ij} \equiv \bra{ij} H_{AB} \ket{ij}
 \ee 
and the variance is given by 
 \be 
\sigma_{E,ij}^2 \equiv \bra{ij} (H - \overline E_{ij})^2  \ket{ij}
= \bra{ij} H_{AB}^2 \ket{ij} - {\Delta_{ij}}^2 \, .
\label{var}
\ee
Given that $H_{AB}$ is $L$-independent, $\Delta_{ij}$ and $\sigma_{E,ij}$ are both of order $\sO(L^0)$.\footnote{In higher dimensions, they should scale with the area of the boundary of subsystem $A$.}

The evolution of a product of subsystem eigenstates $\ket{ij}$ by the full Hamiltonian $H$ is an example of equilibration to finite temperature, where the effective temperature $T$ is associated with the average energy density of $\ket{ij}$.
Here the  equilibration process is driven by the presence of local interactions $H_{AB}$, so that this setup provides an ideal laboratory for exploring consequences of locality. From~\eqref{mean}--\eqref{var}, $\ket{ij}$ is expected to equilibrate to the microcanonical ensemble at energy $\overline{E}_{ij}$, in the sense that at late times it macroscopically resembles a universal equilibrium density matrix $\rho(\overline{E}_{ij})$.


At infinite temperature,  aspects of the equilibration process can deduced from universal properties of operator growth in chaotic systems~\cite{wenwei_abanin_2017_entdynanmics, mezei_2017_entspreadchaotic, lensky_2019_chaoshighT}. This approach has the advantage that locality is built in at the outset, but it has not been clear how to generalize it to finite temperature. Moreover, this approach does not provide insight into the effects of locality on the structure of energy eigenstates.  

Instead, we will consider the decomposition 
\be
\ket{i}_A \ket{j}_B = \sum_a c^a_{ij} \ket{a} 
\ee
and study the evolution of $\ket{i}_A \ket{j}_B$ in terms of properties of the coefficients $c^a_{ij}$. 
As discussed in Sec.~\ref{sec:intro}, the behavior of quantities such as the return probability~\eqref{ptdef} and Renyi entropies~\eqref{ren1}--\eqref{ren2} can be used to infer the imprints of locality on the energy eigenstates. 
We explore the extent to which the coefficients $c^a_{ij}$ have a universal structure, and which aspects of this structure are related to various dynamical properties.\footnote{In the presence of global symmetries, it is more appropriate to consider a decomposition of $\ket{a}$ into subsystem eigenstates in the same charge sector as $\ket{a}$. However, since the conceptual lessons we learn will not be sensitive to this subtlety, we will restrict our attention to systems with no global symmetry.}


\section{Amplitude structure and return probability}  \label{sec:eig}

In this section we will first describe our main results for the structure of the amplitude $|c^a_{ij}|$ in a local chaotic system, and then provide numerical evidence and analytic arguments for these properties.

\subsection{Universal structure of the amplitude \texorpdfstring{$|c^a_{ij}|$}{}}

Our main proposal is that
 in a chaotic system with local interactions, 
\be
e^{S\le(\frac{E_a + \bar{E}_{ij}}{2}\ri)} |c^a_{ij}|^2 \approx f\le(\frac{E_a + \bar E_{ij}}{2}, E_a -  \bar E_{ij}\ri)\, \label{eq:f_universal}
\ee
where $f$ is a smooth function of the average energy $\frac{E_a + \bar{E}_{ij}}{2}$ and the energy difference $E_a-\bar E_{ij}$.
We refer to $f$ as the {\it eigenstate distribution function} (EDF), as it characterizes the spread of the full system eigenstates in a reference basis of products of subsystem eigenstates. In more general contexts, this function is sometimes referred to as the local density of states (LDOS) in the literature.

In the model~\eqref{ham},  we find that the approximation of the LHS of \eqref{eq:f_universal} by a smooth function improves rapidly on increasing the value of $h$ (with $g=-1.05$ fixed), consistent with the expectation that the system becomes more chaotic. As we discuss below, a quantitative measure of the smoothness of the function can thus provide a diagnostic for the onset of chaos in this family of Hamiltonians. 

We will provide evidence for the following properties of $f$: 
\ben

\item  $f$ has a finite limit in the thermodynamic limit $L \to \infty$.

\item $f$ has support only for $\om \equiv E_a-\bar E_{ij} \sim \sO(L^0)$ in the $L \to \infty$ limit. 

\item $f$ depends weakly on $\frac{E_a + \bar E_{ij}}{2}$. More explicitly, we expect that 
$f (E + \de E , \om) \approx f (E, \om)  + \sO(1/L)$ for $\de E \sim \sO(L^0)$. So below for notational simplicity we will often write $f(\om)$, only keeping   the second argument explicit. 

\item In the chaotic regime of the Hamiltonian \eqref{ham}, we find that $f$ fits well to the following simple functional form: 

\be \label{eq:f_form}
f(\omega) = \begin{cases} 
\frac{C_0}{\omega^2 + \Gamma^2} & |\omega| \lesssim \sigma \\
C_1 e^{-\lambda |\omega|} & |\omega| \gtrsim \sigma   
\end{cases}  \, \, , 
\ee 
where the constants $\Gamma$, $C_i$, $\sigma$ and  $\lambda$ are all $\sO(L^0)$, and in particular $\sigma$ is an  energy scale a few times larger than $\Gamma$.  
We expect that $f(\omega)$ is a smooth function that interpolates between the two regimes in \eqref{eq:f_form}. 

A  Lorentzian regime has previously been seen for  $|\braket{a|\psi}|^2$ for other choices of $\ket{\psi}$, in particular for product states, using ideas from random matrix theory \cite{Flambaum2000, deutsch_1991_ETH, deutsch_3, kota_2014_embeddedRMT_quantum}. We provide an analytic argument for the Lorentzian regime for the case of $\ket{\psi}= \ket{ij}$ in Sec.~\ref{subsec:derivation_f_form}, which makes use of similar ideas from random matrix theory. We argue that the Lorentzian also applies in higher-dimensional systems, where $\Gamma$ and $\sigma$ both grow with the area of the boundary of $A$. We therefore conjecture that  is universal for sufficiently small $\omega$ in 
any local chaotic system, although the value of the constants depend on microscopic details.

\item Using general arguments for any system with local interactions along the lines of the Lieb-Robinson bounds, we show that the EDF must obey the following upper bound in any spatial dimension 
\be \label{ub_1}
f\le( \om \ri) \leq C e^{- \le(\frac{|\beta|}{2}+\gamma \ri) |\om|} \, ,  \quad \beta = S'\le(\frac{E_a + \bar E_{ij}}{2}\ri)  \approx S' (\bar E_{ij})
\ee
where $\gamma$ is an $\sO(1)$ constant determined by details of the interactions in the system, and $C$ is a constant which does not scale with volume $V$ but can scale with the area of the boundary of $A$ in higher dimensions. Eq.~\eqref{ub_1} constrains the asymptotic behaviour for large $\omega$, and is consistent with the numerical result \eqref{eq:f_form}. 

\een
We discuss the implications of~\eqref{eq:f_form} for $P(t)$ in Sec.~\ref{sec:pt}, where we will show that it leads to exponential decay, which is consistent with explicit numerical computation of $P(t)$. The direct numerical supports for~\eqref{eq:f_universal} and~\eqref{eq:f_form} are given in Sec.~\ref{subsec:numerical_f_form}. Heuristic analytic arguments for the two regimes of~\eqref{eq:f_form} are given in Sec.~\ref{subsec:derivation_f_form}, and the proof for~\eqref{ub_1} is given in Sec.~\ref{subsec:loc}.

A proposal similar to~\eqref{eq:f_universal} was previously made by Murthy and Sredniki in~\cite{murthy_2019_structure}. Here we point out some differences with their approach and results.  \cite{murthy_2019_structure} used a statement similar to~\eqref{eq:f_universal} as a starting point, 
 and then argued  for certain properties of the EDF using the eigenstate thermalization hypothesis. 
In this work, we numerically test the validity of using a smooth function on the RHS of \eqref{eq:f_universal}, and find it to be a non-trivial property of chaotic systems which distinguishes them from integrable systems.  
\cite{murthy_2019_structure} argued that $f(\omega)$ is a Gaussian in two or higher spatial dimensions. Here, we find explicitly for one spatial dimension that $f(\omega)$ is given by a Lorentzian at small $\omega$ and an exponential at large $\omega$, and argue that we should also see  the Lorentzian regime in higher dimensions. While there is no discrepancy in one spatial dimension, where the arguments of \cite{murthy_2019_structure} do not apply, there does seem to be a discrepancy in higher dimensions. We examine the argument of \cite{murthy_2019_structure} in more detail in Appendix \ref{app:MSlimitation}, and discuss some potential  limitations of it.  


We close this subsection by mentioning some kinematic constraints on $f$. From  
 the normalization of $\ket{i j}$,  
\be 
1 = \sum_a |c^a_{ij}|^2 = \int d \omega \, e^{{\frac{\omega}{2}} \beta (\bar E_{ij})} \,  f\le(\omega\ri)  ,
\label{norm1}
\ee
where we have approximated the sum over $a$ by $\int dE_a e^{S (E_a)}$, and used that $f$ is supported 
for $\omega \equiv E_a-\bar E_{ij} \sim O(L^0)$ 
and $S(E_a)- S\left(\frac{E_a + \bar E_{ij}}{2}\right) \approx {\frac{\omega}{2}} \beta (\bar E_{ij}) $, with 
$\beta (\bar E_{ij})  = \frac{dS}{dE} |_{\bar E_{ij}} \approx \beta (\frac{E_a + \bar E_{ij}}{2}) + O(1/L)$. 
Assuming that the product of the density of states of $H_A$ and $H_B$ is approximately equal to the density of states of $H$ (we justify this approximation in Appendix \ref{app:therm} below), the normalization of $\ket{a}$ leads to a similar constraint, 
\be 
1 = \sum_{ij}|c^a_{ij}|^2 = \int d\omega \, e^{-\frac{\omega}{2} \beta(E_a)}f(\omega) \label{norm2} \, . 
\ee
From~\eqref{mean}--\eqref{var}
\bega
0 = 
 \int d E_a \,  e^{S(E_a)} (E_a - \bar E_{ij})  |c^a_{ij}|^2 = \int d \omega \, e^{{\omega \ov 2} \beta (\bar E_{ij}) } \,  f\le( \omega \ri) \omega ,\label{mean2}
\\
\sigma_{E,ij}^2 = \int d E_a \,  e^{S(E_a)} (E_a - \bar E_{ij})^2 |c^a_{ij}|^2 = \int d \omega \, e^{{\omega \ov 2} \beta (\bar E_{ij}) } \,  f\le( \omega \ri) \omega^2 \ .
 \label{var2}
\end{gather}
Equation~\eqref{ub_1} ensures that the integrals~\eqref{norm1}--\eqref{var2} are convergent. All equations can be satisfied for arbitrary $i, j$ due to the weak dependence of $\beta(\bar E_{ij})$ and $f$ on $\bar E_{ij}$.

\subsection{Evolution of \texorpdfstring{$P(t)$}{}}
\label{sec:pt}

In this subsection, we discuss the implications of~\eqref{eq:f_universal}--\eqref{eq:f_form} for the return probability~\eqref{ptdef} and compare them with explicit numerical computations. 

From~\eqref{ptdef} and~\eqref{eq:f_universal} we have
\be 
P(t) \approx |\theta(t)|^2, \quad \quad \theta(t) 
= \int_{-\infty}^{\infty} d \omega \, e^{\beta\omega/2}  f( \omega ) \, e^{-i \omega t}  \, , \label{ptf}
\ee
where $\beta=S'(\bar E_{ij})$. At early times, expanding 
$e^{-i\omega t}$ in the integrand of $\th (t)$ in Taylor series, we find 
\begin{align}
        \theta(t) &=  \int\, d\omega \, e^{\beta/2 \omega} f (\omega) \sum_n \frac{(-i \omega t)^n}{n!} = 1 - \ha \sigma^2_{E,ij} \,  t^2 + O(t^3)
   \label{theta} \end{align}
where we have used~\eqref{norm1}--\eqref{var2}. That the small $t$-expansion~\eqref{theta} is well defined is warranted by~\eqref{ub_1}. Hence to quadratic order in $t$,
\begin{equation}
    P(t) \approx 1 - \sigma^2_{E,ij} \,  t^2  + O(t^3) \, .  
\end{equation}
 We verify this quadratic decay of $P(t)$ 
 at early times in Fig. \ref{subfig:EarlyPt}. 
 The numerical coefficient of the quadratic fit is very close to $\sigma_{E,ij}^2 = 1$. To understand this value, we recall that $\sigma_{E,ij}^2 = \bra{ij} H_{AB}^2 \ket{ij} - \bra{ij}H_{AB}\ket{ij}^2$. Since $H_{AB}$ is a product of two Pauli-$Z$ operators across the cut, the first term is $\bra{ij} H_{AB}^2 \ket{ij} = 1$. As for the second term, $\lim_{L \rightarrow \infty} \bra{ij} H_{AB} \ket{ij} = 0$ whenever $\ket{i}, \ket{j}$ are at effective temperatures above the critical temperature for the $Z_2$ symmetry-breaking phase transition. These facts together imply that for states sufficiently close to the middle of the spectrum, $\sigma_{E,ij}^2 \approx \bra{ij} H_{AB}^2 \ket{ij} = 1$ at finite but large $L$.

\begin{figure}[!h] 
\centering
    \begin{subfigure}{\textwidth}
        \centering
        \includegraphics[width = \textwidth]{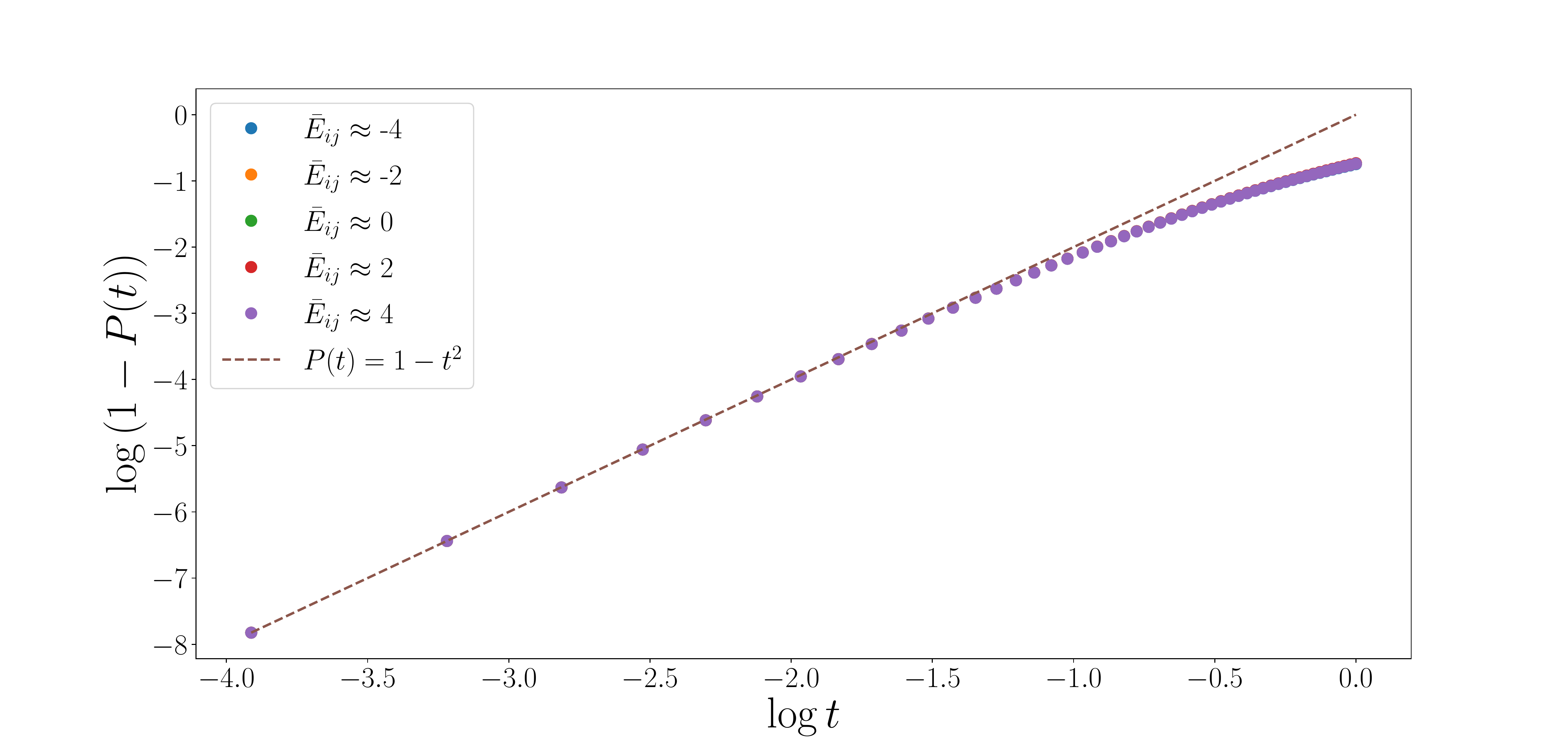}
        \caption{}
        \label{subfig:EarlyPt}
    \end{subfigure}
    \begin{subfigure}{0.49\textwidth}
        \centering
        \includegraphics[width = \textwidth]{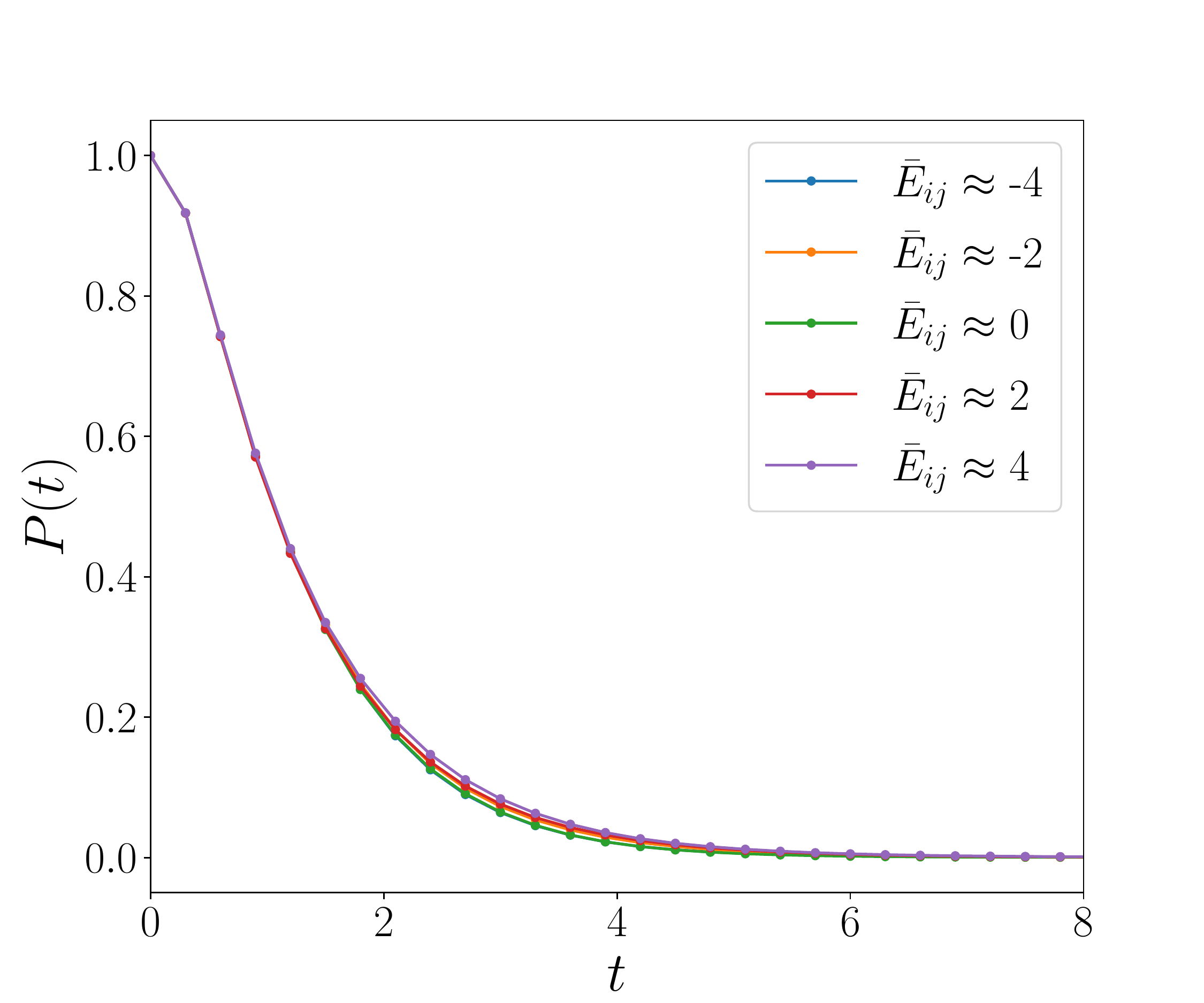}
        \caption{}
        \label{subfig:Pt_L20}
    \end{subfigure}
    \begin{subfigure}{0.49\textwidth}
        \centering
        \includegraphics[width = \textwidth]{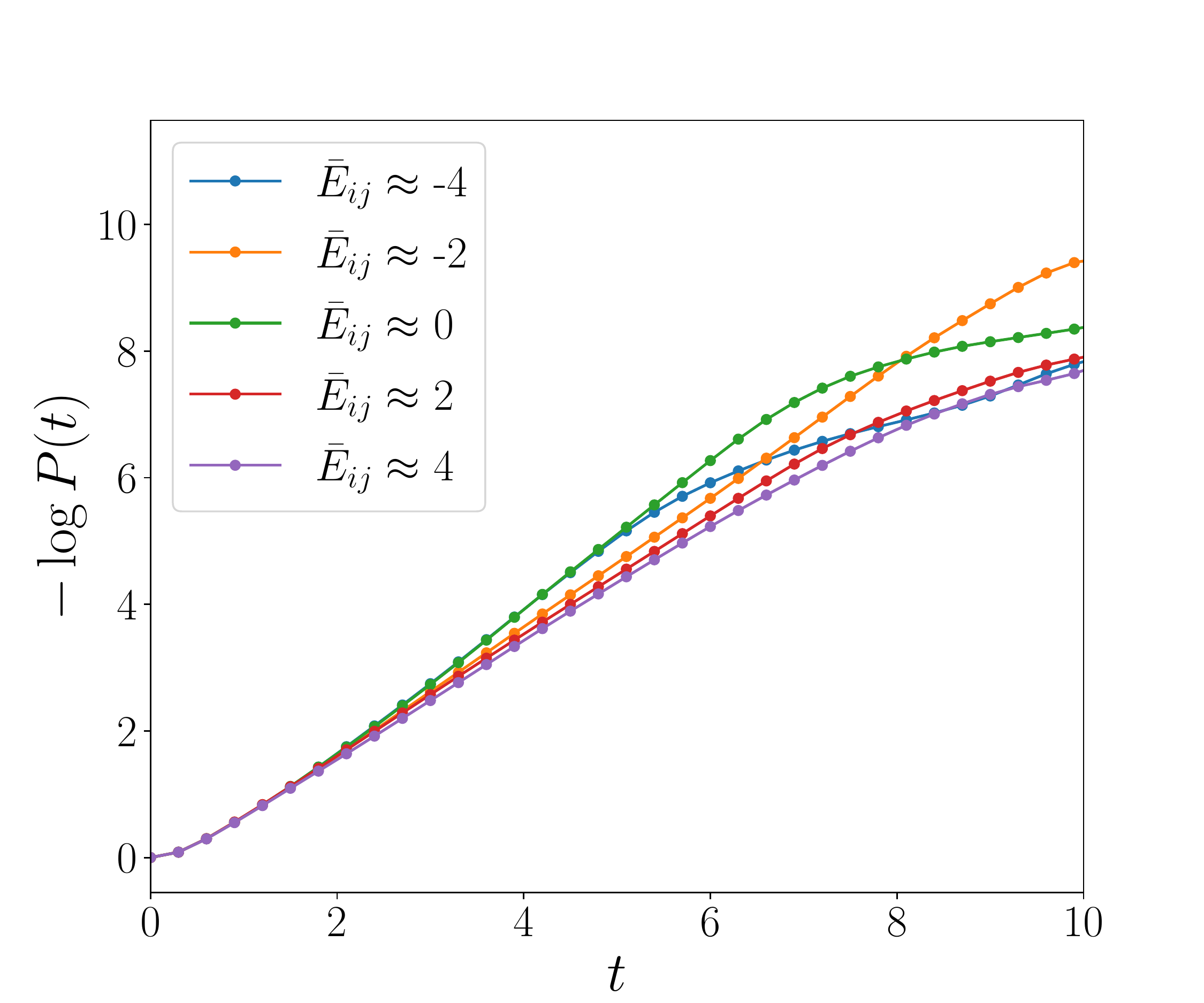}
        \caption{}
        \label{subfig:LogPt_L20}
    \end{subfigure}
    \caption{Here we plot the return probability $P(t)$ for various factorized states with different values of $\bar E_{ij}$. The top/bottom panel shows the early/intermediate time behavior of $P(t)$ at $L = 22$/$L=20$. At early times, curves at different energies perfectly align and fit very well to the expected quadratic form. The coefficient of the quadratic term $\Gamma_{ij}^2 \approx 1$ for $\bar E_{ij} \in [-4,4]$, consistent with the analytic prediction. Following this quadratic regime, $P(t)$ decays exponentially for an extended range of time, until finite size effects begin to dominate at late times. One can numerically estimate the decay rate by fitting $-\log P(t)$ to a linear function for intermediate times (i.e. after the quadratic regime and before saturation). Using data from $\bar E_{ij} \in [-6,6]$, we estimate the decay rate to be $0.98 \pm 0.13$, which is consistent with the Lorentzian prediction $2 \Gamma \approx 0.94$. 
    }
    \label{fig:Loschmidt}
\end{figure}

For $t \gg \frac{1}{\sigma}$, the quadratic scaling breaks down and transitions to an exponential decay. To understand this behavior, we invoke the functional form of the EDF in \eqref{eq:f_form}. For simplicity we consider initial states with $\beta=0$ and make the following decomposition
\begin{equation}
    f(\omega) = f_L(\omega) + f_{\rm res}(\omega) \,, \quad f_L(\omega) = \frac{C_0}{\omega^2+\Gamma^2} \,, \quad |f_{\rm res}(\omega)| \leq |f_L(\omega)|  \,.
\end{equation} 
Upon taking a Fourier transform, 
\begin{align}
    P(t) 
    &= \frac{C_0^2 \pi^2}{\Gamma^2} \, e^{-2\Gamma t} \, \,   + 2 \frac{C_0 \pi}{\Gamma} e^{-\Gamma t}f_{\rm res}(t) + \, \,  f_{\rm res}(t)^2 \,,
\end{align}
where $f_{\rm res}(t) = \int d \omega e^{-i\omega t} f_{\rm res}(\omega)$. Generally, there are two possible scenarios for the large $t$ behavior of $P(t)$: (1) If the singularities of $f_{\rm res}(\omega)$ satisfy $|\Im \omega_{\rm sing}| > \Gamma$, then regardless of the precise functional form of $f_{\rm res}(\omega)$, the single exponential decay $e^{-2\Gamma t}$ of $P(t)$ dominates at late times. (2) If $f_{\rm res}(\omega)$ has singularities closer to the real axis, or if $f_{\rm res}(\omega)$ is not analytic at all, the asymptotic late-time decay will be controlled by $P(t) \approx f_{\rm res}(t)^2$. In Fig.~\ref{subfig:LogPt_L20}, we see that the numerical exponential decay rate of $P(t)$ is close to $2\Gamma$, in favor of scenario (1).

In a finite-size system, the decay of $P(t)$ must eventually be cut off at $t = \mathcal{O}(L)$ as $P(t)$ saturates to an $\sO(e^{-S})$ value: 
\begin{align}  
\lim_{t \rightarrow \infty}P(t) \approx 
& \lim_{T \rightarrow \infty}\frac{1}{T}\int_0^T dt\, \sum_{a,b} |c^a_{ij}|^2 |c^b_{ij}|^2 e^{i t (E_a - E_b) } \nonumber \\ &  \approx \sum_a |c^a_{ij}|^4 \approx e^{-S(\bar E_{ij})} f^2 (0) \, . 
\end{align} 
Thus we expect the exponential decay to hold in the range 
\be 
{1 \ov \sigma} \ll t \ll {S (\bar E_{ij}) \ov 2 \Gamma}  \ .
\ee

Note that $P(t)$ is a special case of a quantity called the Loschmidt echo 
\be 
\sL
(t) = |\braket{\psi| e^{-i H't} e^{i H t}| \psi}|^2, \quad H' 
= H + H_1 
\label{le}
\ee
which characterizes the sensitivity of the system to small perturbations of the Hamiltonian. This quantity has been extensively studied in the context of  quantum chaotic few-body systems \cite{Peres_1984_stability_quantummotion,Prosenreview_2006_Loschmidtecho}, where $||H_1||$ is taken to be sufficiently small that we can apply perturbation theory, and \eqref{le} is found to decay exponentially with time. It has also been studied for quantum many-body systems, where on taking $H_1$ to act on an $\sO(1)$ number of sites, and $\ket{\psi}$ to be a product state, \eqref{le} decays exponentially or faster with time~\cite{vardhan_2017_echo_disorder}. $P(t)$ is special case of \eqref{le} for $\ket{\psi}=\ket{ij}$ and  $H_1= H_{AB}$.

\subsection{Numerical results in spin chain models}\label{subsec:numerical_f_form}

In this section, we numerically test the approximation \eqref{eq:f_universal} for the family of Hamiltonians in \eqref{ham}, with $g=-1.05$, and $h$ varied from 0 to higher values. By the standard measures of chaos through spectral statistics, the $h=0$ case is integrable, while the $h=0.5$ case, and most generic values of $h$, are chaotic \cite{kimhuse_2013_entballistic,atas2013_ratiostats}.~\footnote{We use open boundary conditions in all cases.} 

Fixing some eigenstate $\ket{a}$ of $H$, we evaluate the LHS of \eqref{eq:f_universal} by averaging the quantity $e^{S\le(\frac{E_a + \bar{E}_{ij}}{2}\ri)} |c^a_{ij}|^2$ over $\ket{ij}$ in  a small energy interval $\bar E_{ij}' \in [\overline{E}_{ij} - \sigma/2, \overline{E}_{ij} + \sigma/2]$ with $\sigma = 0.2$. 
We then plot this quantity as a function of the energy difference  $\overline{E}_{ij}- E_a$. The results for the cases $h=0.5$ and $h=0$  are shown in Fig. \ref{fig:self_av}, for a few different choices of $E_a$ within a much smaller energy range
. In the chaotic system, the resulting plot turns out to be approximately the same for the different nearby values of $E_a$. In the integrable system, there is much greater fluctuation between plots for nearby energies.

\begin{figure}[!h]
    \centering
    \includegraphics[width=0.45\textwidth]{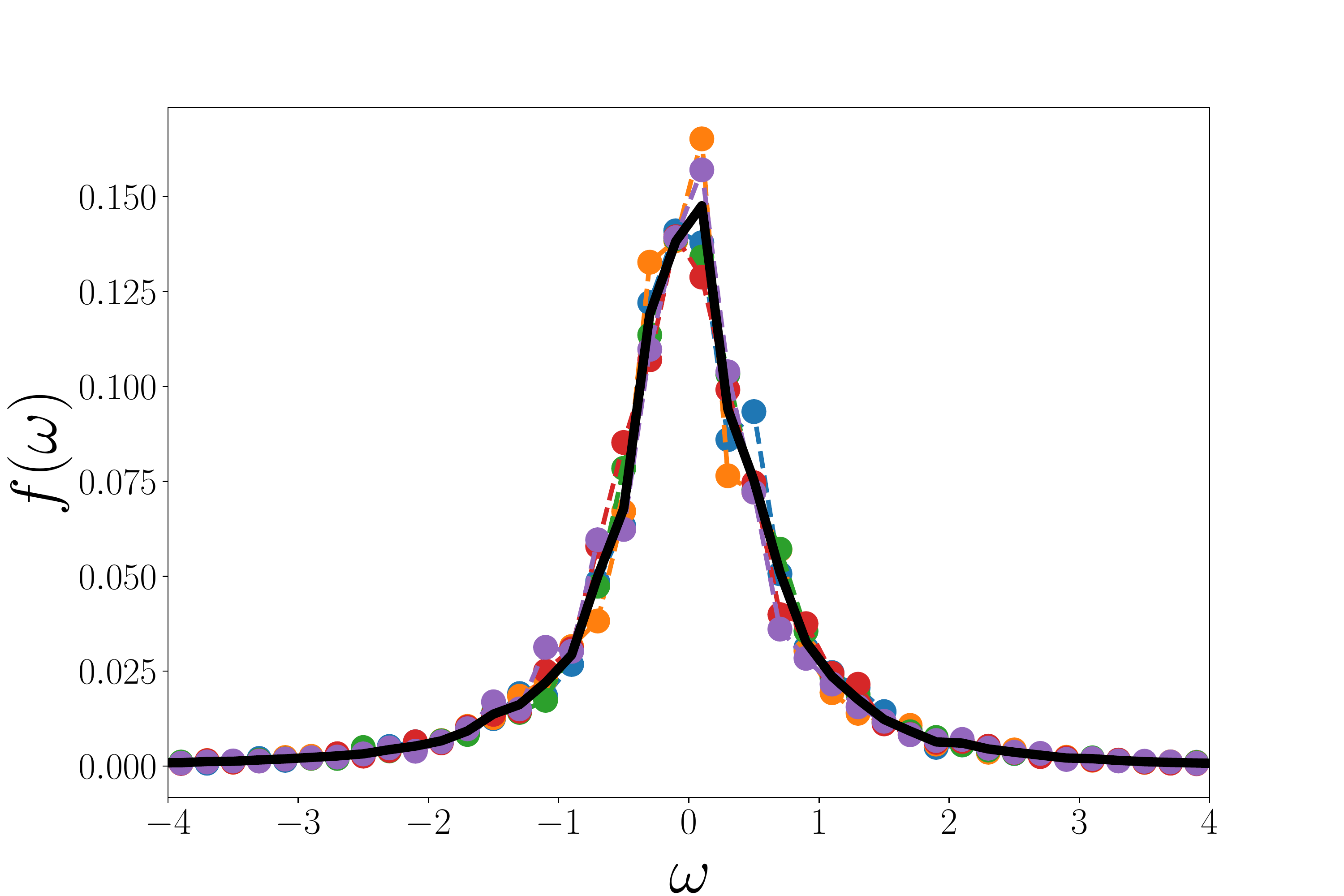}
\includegraphics[width=0.45\textwidth]{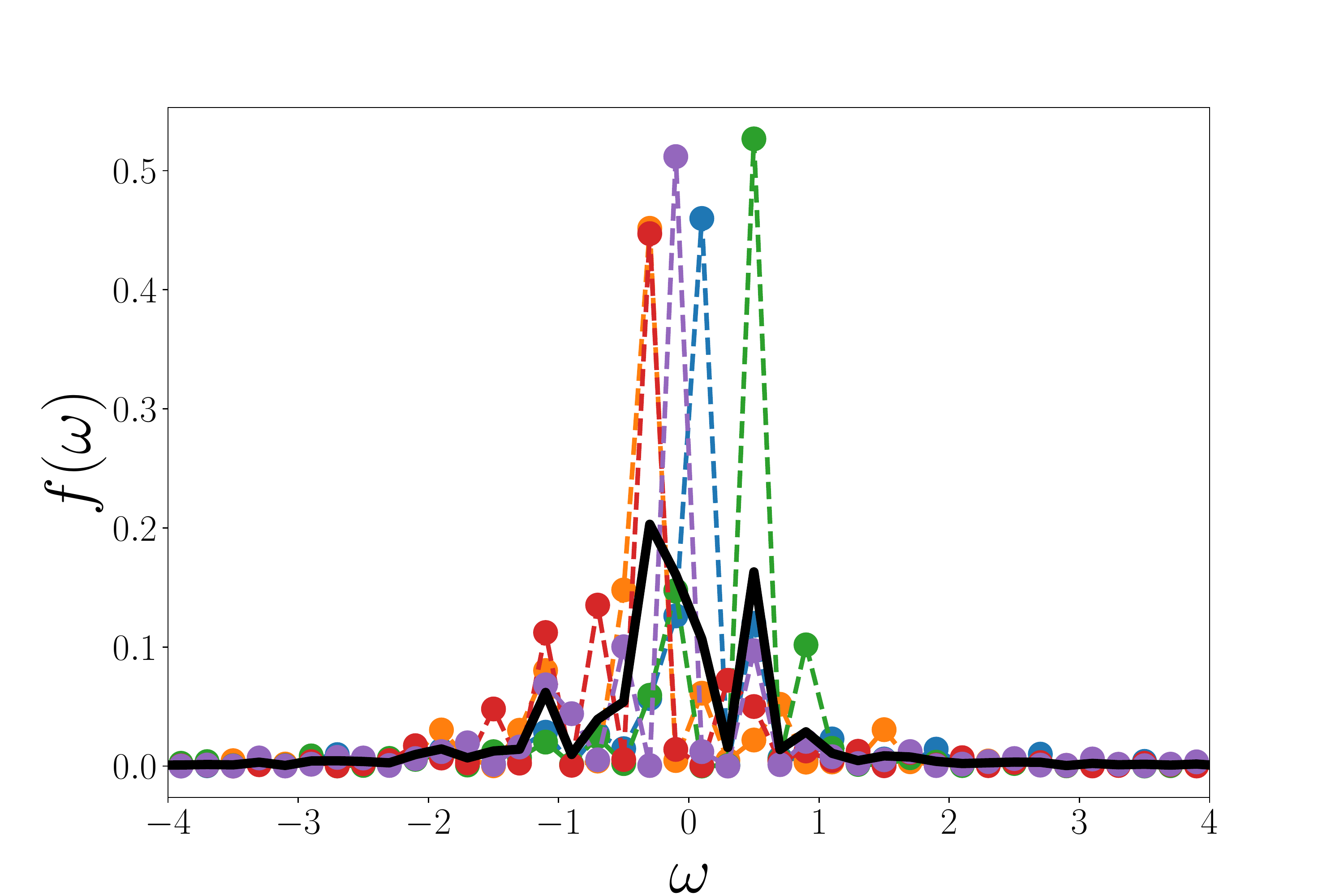} 
    \caption{We scatter plot the LHS of \eqref{eq:f_universal} evaluated with 5 different choices of $E_a$ near $E_a = 0$ for a chaotic (left) and integrable (right) spin chain at system size $L = 14$. The black curve is the average over 5 states for each model. The sample-to-sample fluctuations for different choices of $E_a$ 
    are much larger in the integrable case. 
    }
    \label{fig:self_av}
\end{figure}

\begin{figure}[!h]
    \centering
    \includegraphics[width = \textwidth]{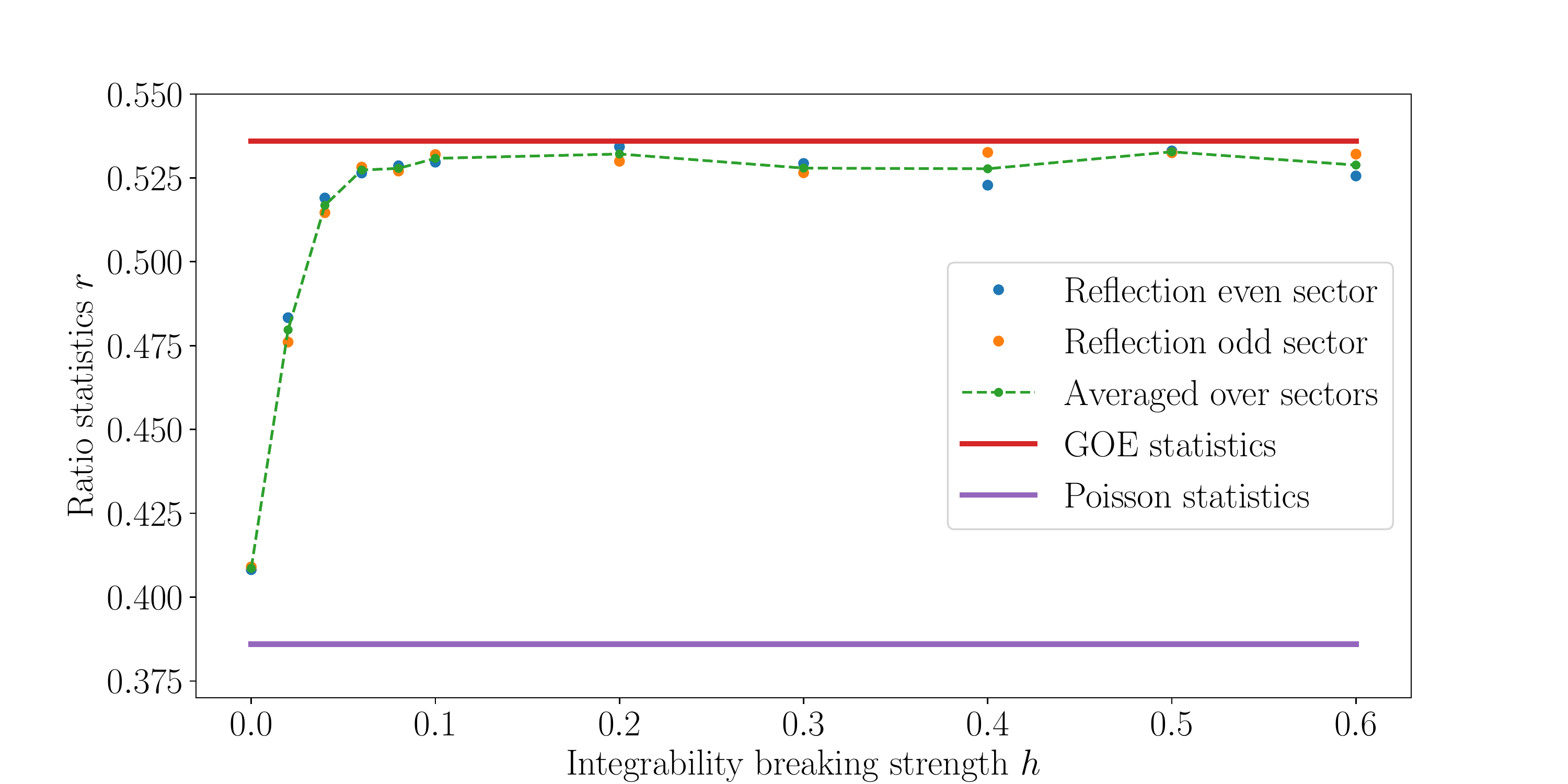}
    \caption{We plot the ratio statistics $r$ of the energy eigenvalues, introduced in  \cite{huse_ogan_fermions},  for the model \eqref{ham} with fixed $g = -1.05$ and $h$ going from 0 to 0.6 at system size $L = 14$. The quantity $r$ is the average over all energy levels $n$ of $r_n =\text{min}(\delta_n, \delta_{n+1})/\text{max}(\delta_n, \delta_{n+1})$, where $\delta_n= E_{n+1}-E_{n}$. 
    The red/purple lines show the characteristic values of this average in the GOE/Poisson distributions, which are expected to match the values in chaotic/integrable systems.}
    \label{fig:ratio_stats_varyh}
\end{figure}

\begin{figure}[!h]
    \centering
    \includegraphics[width=0.45\textwidth]{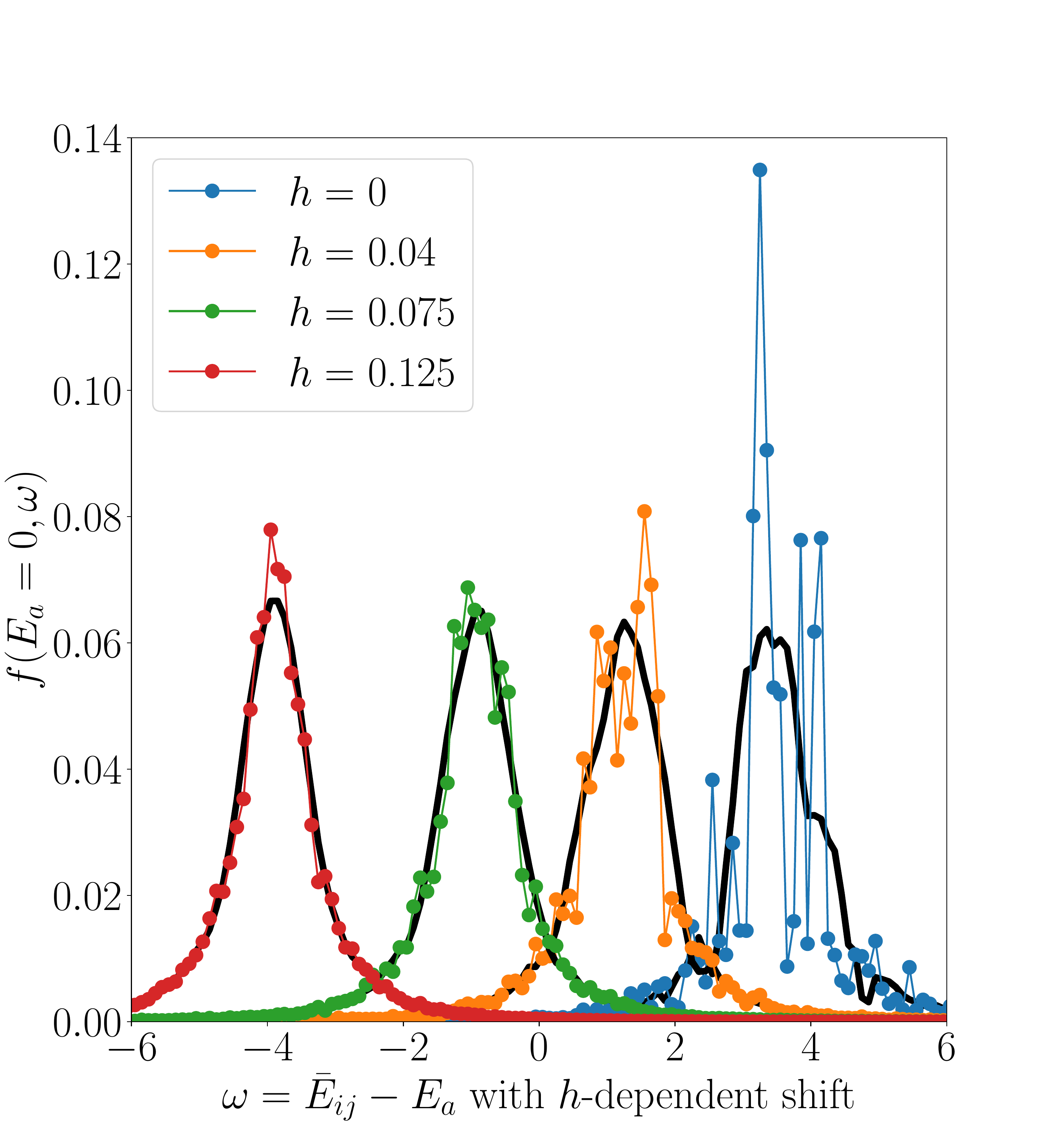}
    \includegraphics[width=0.45\textwidth]{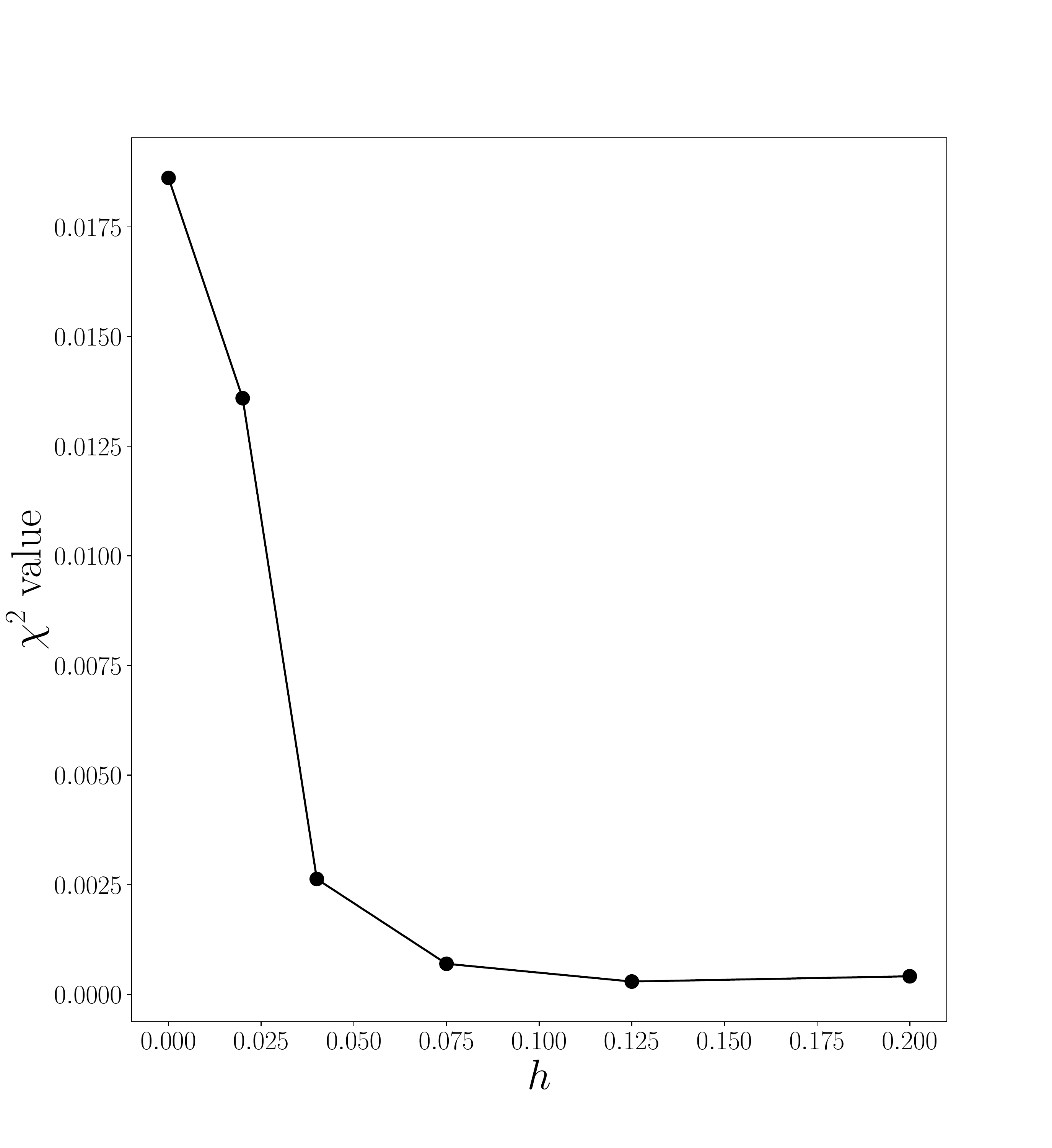}
    \caption{Left: Comparison between the quantity on the LHS of \eqref{eq:f_universal} and the Savitsky-Golay smooth approximation to it (black curves) for a few different values of $h$ close to the integrable point $h = 0$. Right: The quantity $\chi^2(h)$ in \eqref{chi} as a function of $h$ which measures the proximity between the smooth approximation and the raw data. Both plots are for system size $L=14$.}
    \label{fig:smooth}
\end{figure}

A related immediate observation from Fig. \ref{fig:self_av} is that after averaging over a few states with nearby values of $E_a$, the functional form of \eqref{eq:f_universal} is much smoother for chaotic systems than for integrable systems. It is expected that for finite system size, the model \eqref{ham} with $g=-1.05$ has a crossover from integrable to chaotic behaviour at some small finite value of $h$.~\footnote{In the thermodynamic limit, this transition point is expected to approach $h=0$.}  This suggests that there should be a rapid increase in how well the quantity on the LHS of \eqref{eq:f_universal} is approximated by a smooth function on increasing $h$ from 0. We confirm this expectation in Fig. \ref{fig:smooth}. To quantify the smoothness, we consider the averaged curves $f_h(\omega)$ for $E_a=0$ (as in the black curves of Fig. \ref{fig:self_av}) for different values of $h$, and smooth the averaged curves using a Savitzky-Golay filter\footnote{The Savitzky-Golay performs local least-squares fitting to degree-$k$ polynomials in sliding windows of size $n$. We chose $k = 3, n = 5$ for the fitting but the qualitative conclusions we draw are independent of the precise choice.}. 
We then measure the difference between the actual curve $f_h(\omega)$ and its smoothed version $f_{h, \, \rm smooth}(\omega)$ using the $\chi^2$ statistic,
\be \label{chi}
\chi^2(h) = \sum_{i=1}^{n} \frac{\le(f_h(\omega_i)- f_{h, {\rm smooth}}(\omega_i)\ri)^2}{f_{h, {\rm smooth}}(\omega_i)}
\ee
where $n$ is the number of choices of $\omega_i$ over which we sample the function. As shown in Fig. \ref{fig:smooth}, $\chi^2$ decreases rapidly with $h$, indicating the improvement of the smooth approximation. 
From this measure, for $L = 14$, the onset of chaos occurs around $h = 0.1$. This value roughly coincides with the onset of random matrix energy level statistics in the same family of models, as seen in Fig.~\ref{fig:ratio_stats_varyh}
. The onset of chaos for finite system size in terms of properties of eigenstates has previously been studied using other measures in for instance \cite{pandey_2020_adiabatic,sugiura_2021_adiabatic}.

Let us now discuss the universal properties of the smooth eigenstate distribution  function $f(\omega)$ in the chaotic regime, using the case with $h=0.5$ as the representative example. We first consider the dependence of $f(\omega)$ on the system size, and verify that its functional form is independent of $L$, as shown in Fig. \ref{fig:Lindependence_f}. 

\begin{figure}[!h]
    \centering
    \includegraphics[width=0.45\textwidth]{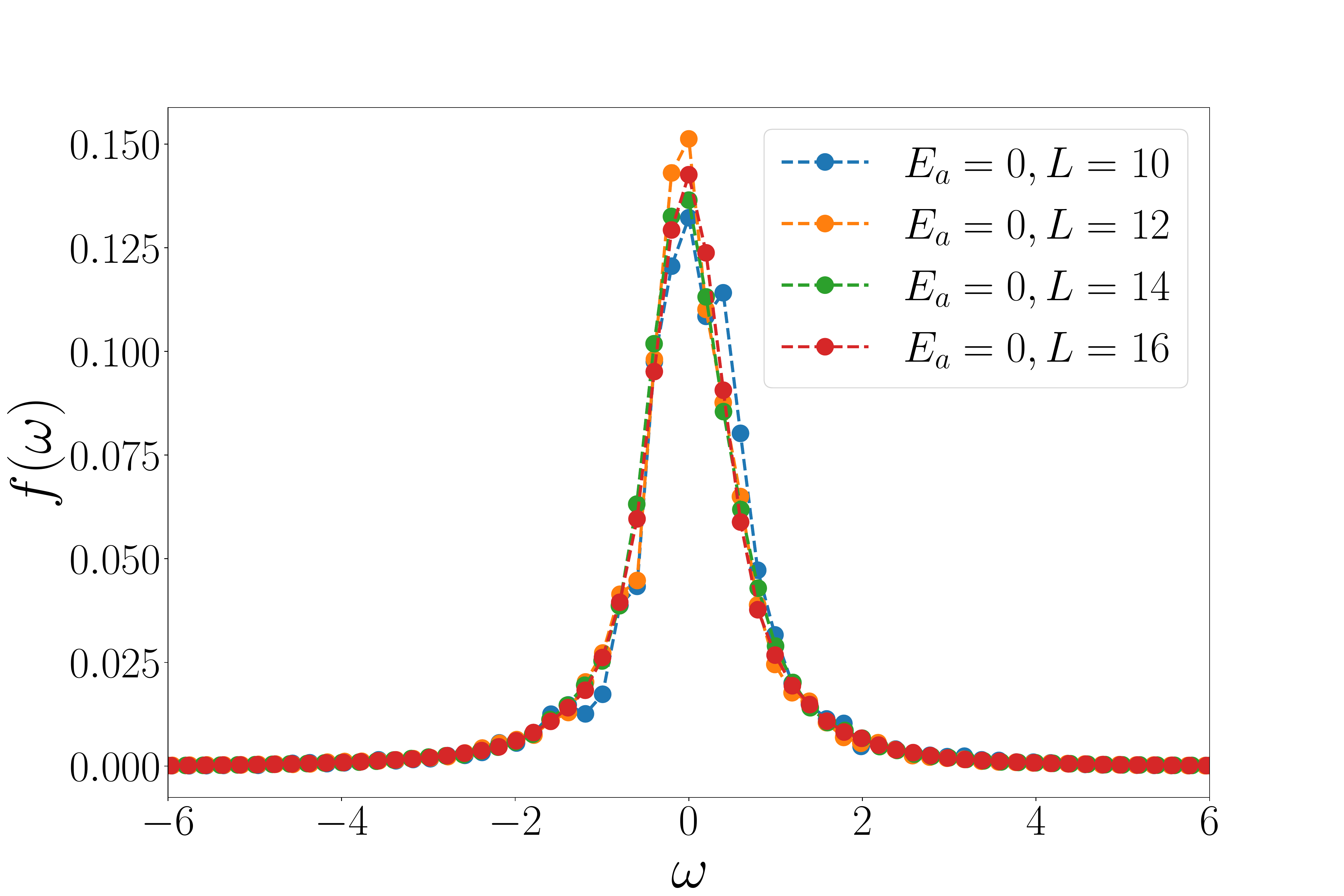}\includegraphics[width=0.45\textwidth]{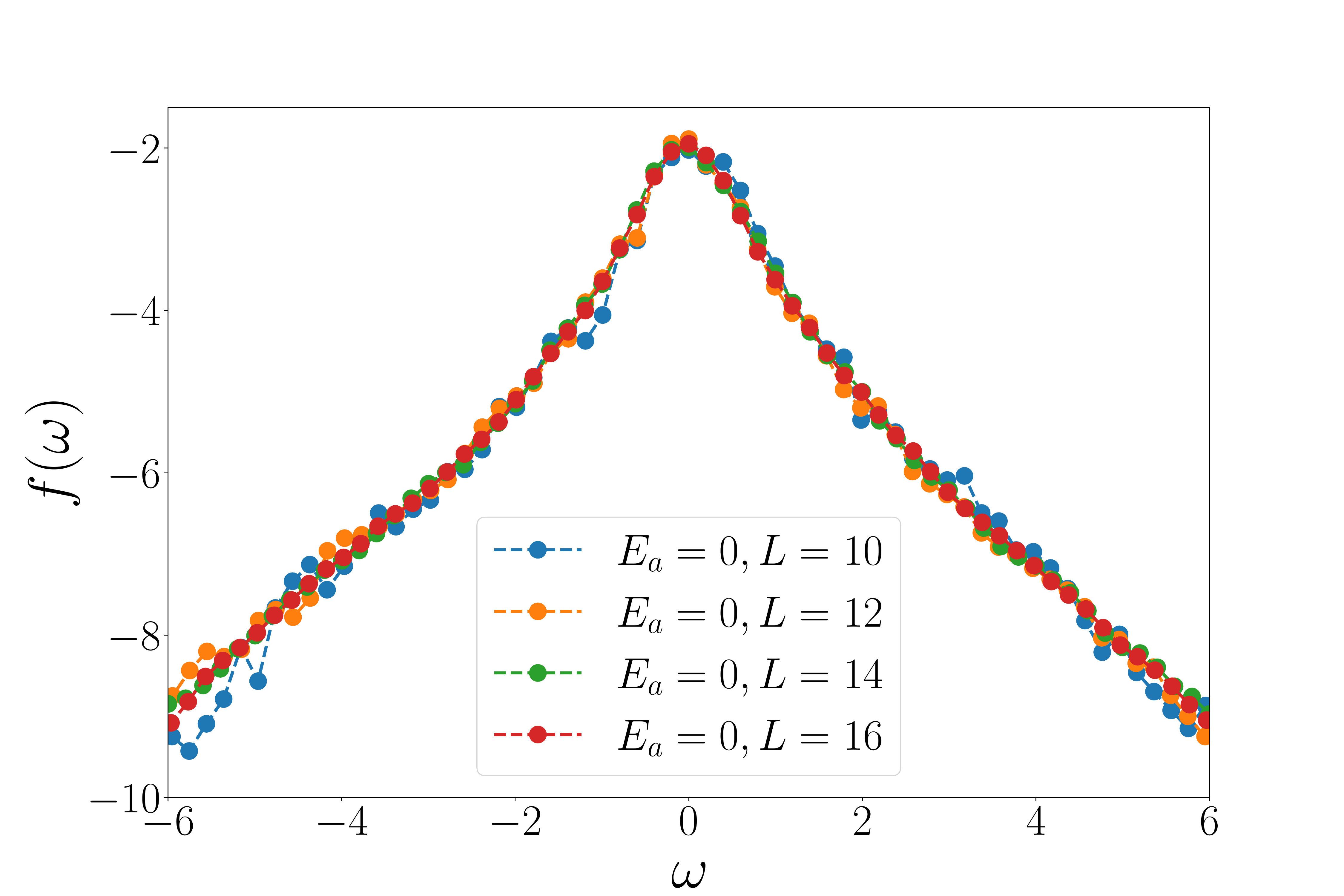}
    \caption{We plot the EDF evaluated at increasing system sizes $L$, with fixed energy $E_a \approx 0$. On the left panel (where the $y$ axis is linear-scale), the difference between curves at different values of $L$ cannot be resolved when $\omega$ is large. To emphasize that even the large $\omega$ regime is independent of $L$, we also show the log-linear plot in the right panel.}
    \label{fig:Lindependence_f}
\end{figure}

Next, we verify that EDF has a weak dependence on the average energy $(E_a + \bar E_{ij})/2$. In the left panel of Fig.~\ref{fig:energy_diff}, we carry out the same averaging procedure used to evaluate the EDF in the black curve in Fig. \ref{fig:self_av}, now considering different values of $E_a$, and making a small shift on the $\omega$ axis so that the peaks of the different cases coincide. We find that the  curves approximately coincide, indicating 
that $f(E_a -\frac{\omega}{2} , \omega)$ are close for different choices of $E_a$. In the right panel of Fig. \ref{fig:energy_diff}, we contrast this with the behaviour of the quantity $e^{S\le(\bar E_{ij} \ri)} |c^a_{ij}|^2$ considered in \cite{murthy_2019_structure}, 
which has a stronger dependence on the average energy. 

\begin{figure}[!h] 
\includegraphics[width=0.5\textwidth]{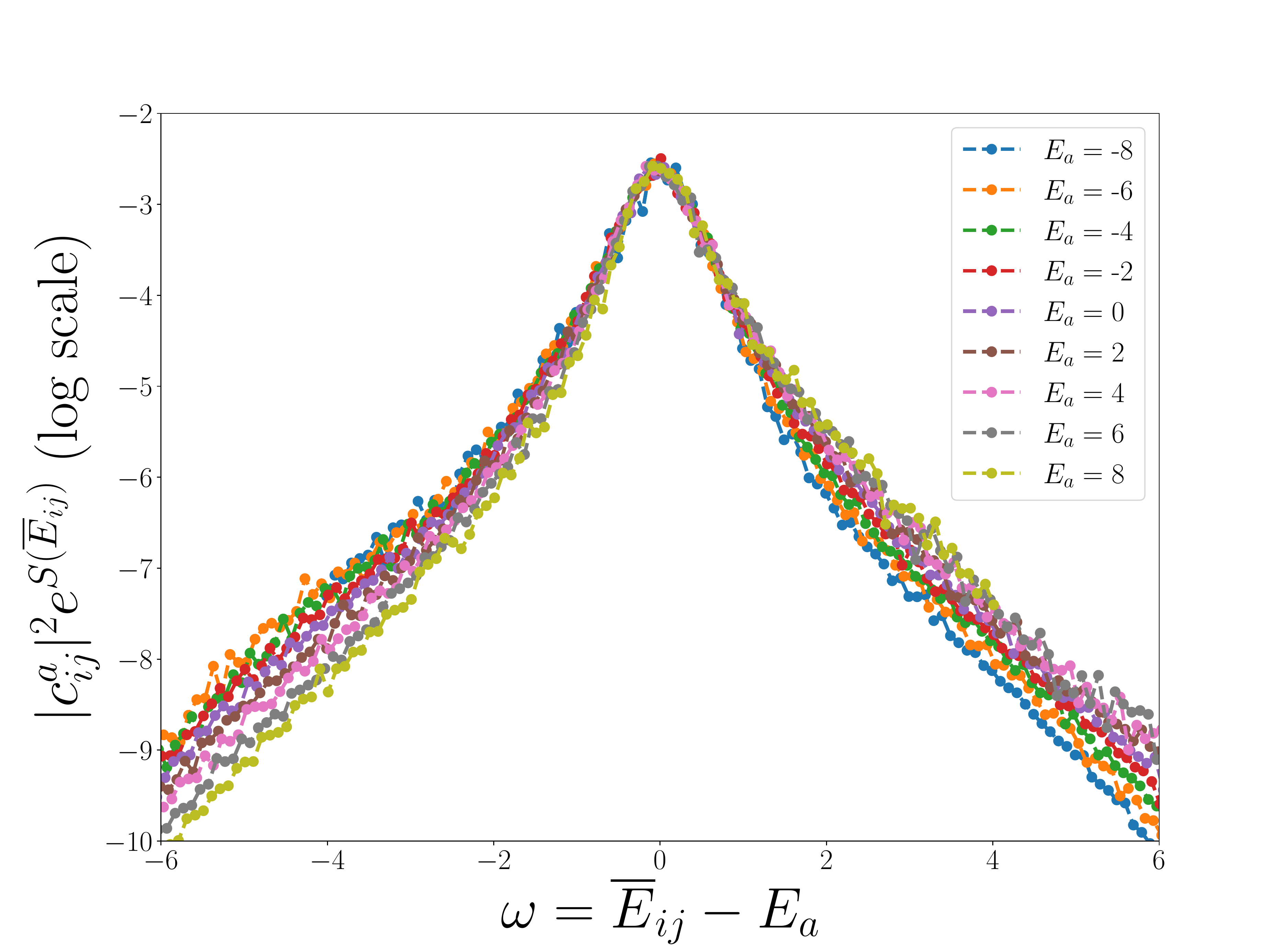} \includegraphics[width=0.5\textwidth]{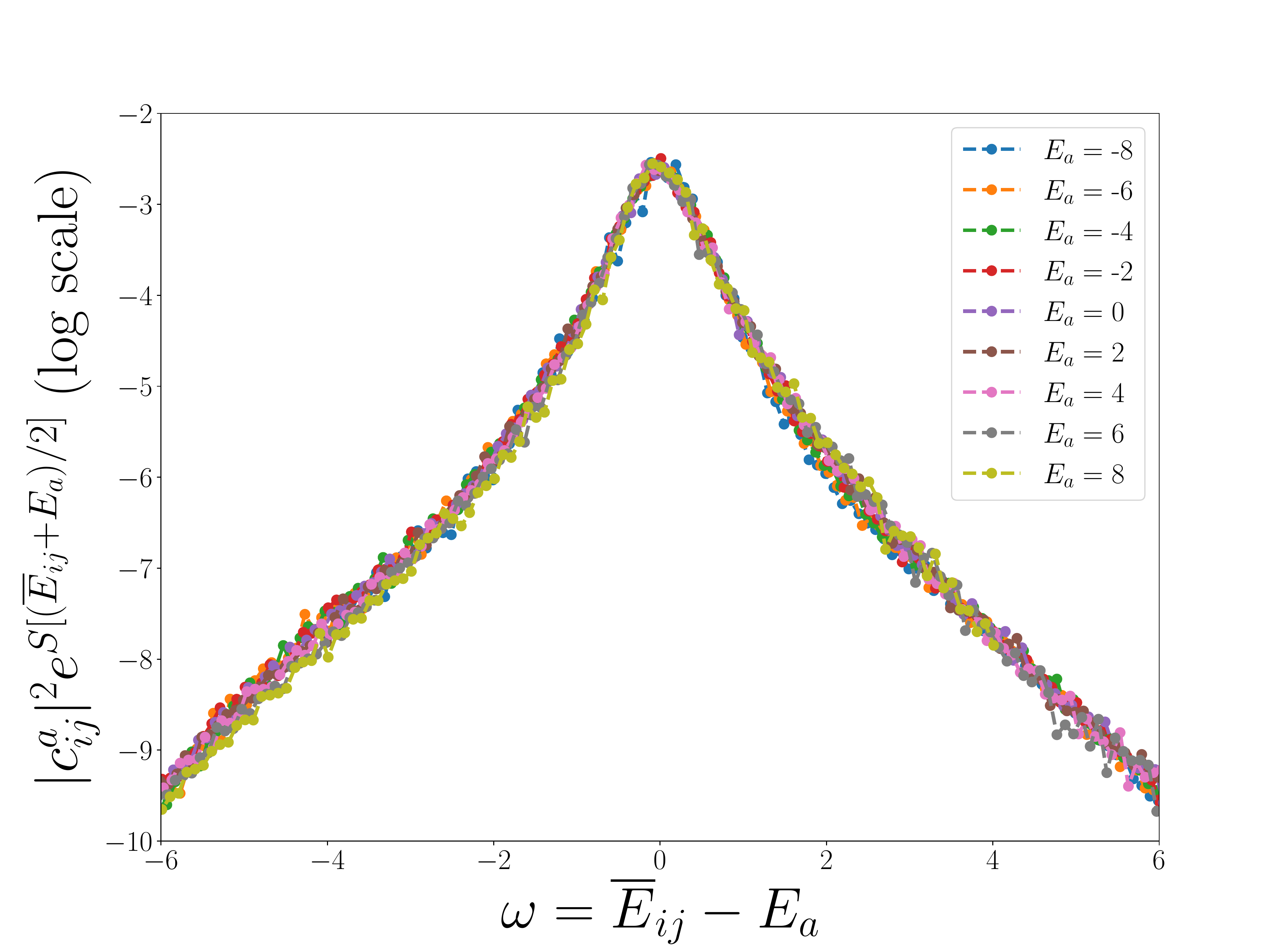}
\caption{On the left/right panel, we plot the quantity $e^{S(\bar E_{ij})} |c^a_{ij}|^2$ and the quantity $e^{S(\frac{\bar E_{ij} + E_a}{2})} |c^a_{ij}|^2$ as a function of $\omega = \bar E_{ij} - E_a$ for 9 different values of $E_a$ at system size $L = 16$. In both panels, for different values of $E_a$, we applied a small horizontal shift so that the peaks are aligned. It is clear that at large $\omega$, the adjustment of entropy factor in the right panel corrects for the systematic dependence on $E_a$ observed in the left panel. This provides further justification for the definition of the eigenstate distribution function in \eqref{eq:f_universal}. 
}
\label{fig:energy_diff}
\end{figure} 

\begin{figure}[!h]
    \centering
    \includegraphics[width = 0.8\textwidth]{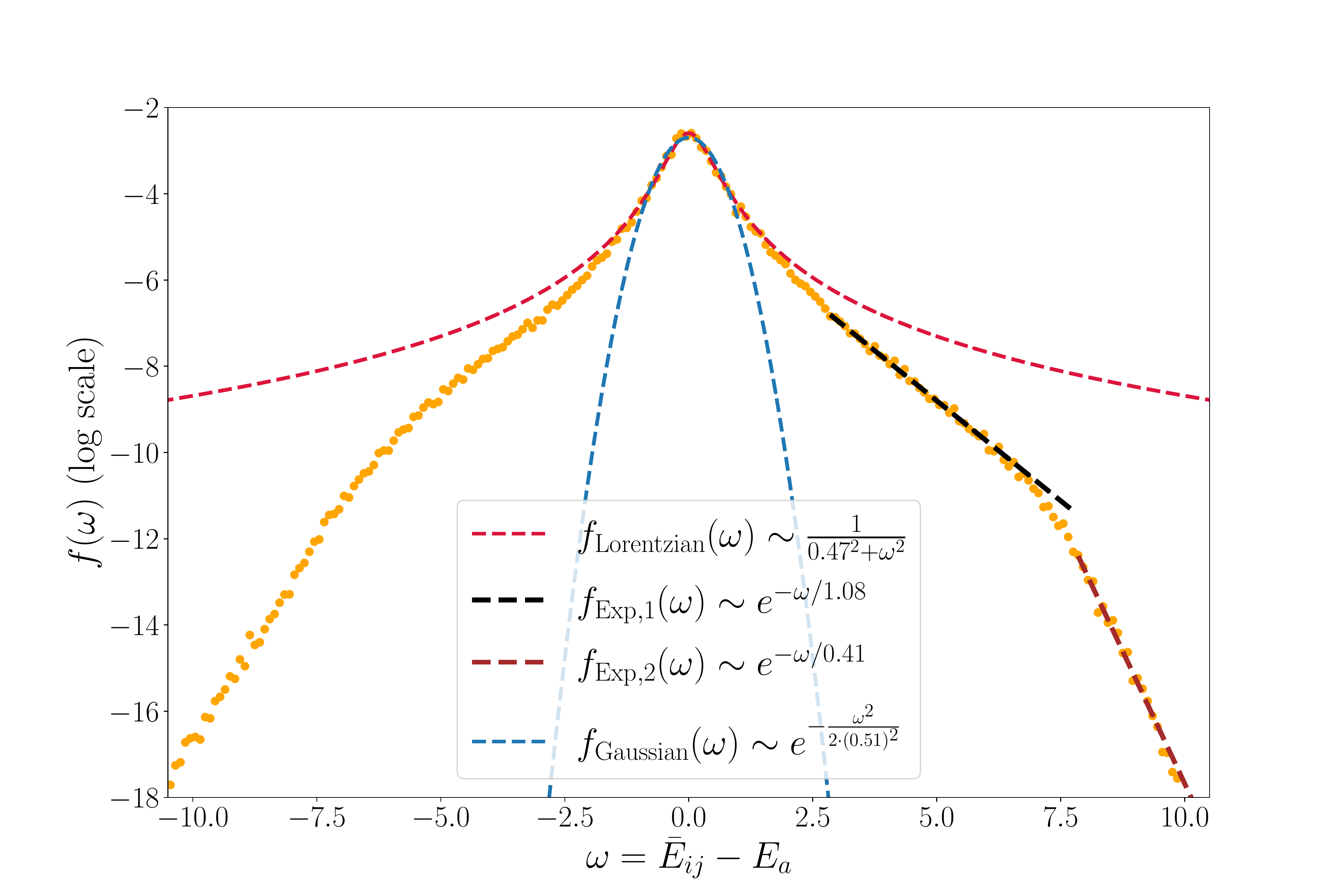}
    \caption{
    We plot the functional form of the EDF for $L = 16$, averaged over 5 states with energy closest to $E_a = 0$. At small $\omega$, a Lorentzian fit works well, and in particular is more accurate than a Gaussian fit. At large $\omega$, we observe two exponential regimes with different decay rates.  
    }
    \label{fig:f_form}
\end{figure}



Finally, we study the universal  functional form of $f$. By fitting the function in different regimes as shown in Fig.~\ref{fig:f_form}, we obtain the piecewise form \eqref{eq:f_form}. 
Note that we expect the $\sO(1)$ numbers $\Gamma, \sigma$ and $\lambda$ in \eqref{eq:f_form} to depend on details of the system. We compare the best fit to a Gaussian form, which works well for a much smaller range of $|\omega|$. As mentioned in the previous subsection, this is expected as the arguments of \cite{murthy_2019_structure} do not apply in one spatial dimension.

In the remaining discussion, we  explain aspects of this universal functional form and consider its consequences for dynamical quantities.


\subsection{Physical argument for the form of the eigenstate distribution function}\label{subsec:derivation_f_form}



The accurate fitting of $f(\omega)$ to a Lorentzian at sufficiently  small $\omega$ and an exponential at large $\omega$ suggests the existence of some simple underlying principles. Here, we motivate the Lorentzian regime using random matrix universality and spatial locality as inputs. 
This argument also allows us to relate the width of the Lorentzian to certain microscopic parameters of the Hamiltonian.


Let us decompose the Hilbert space of the full system into $\ket{i} \ket{j}$ and its orthogonal complement, $\mathcal{H} = \mathcal{H}_c \oplus \ket{i}\ket{j}$. After diagonalizing the Hamiltonian $H$ within the subspace $\mathcal{H}_c$, we can put $H$ into a block form
\begin{equation}
    H = \begin{pmatrix} \bar E_{ij} & V_{2,ij} & \ldots & V_{D,ij} \\ V_{2,ij}^* & \epsilon_2 & 0 & \ldots \\ \ldots & 0 & \ldots & 0 \\ V_{D,ij}^* & 0 & \ldots & \epsilon_D \end{pmatrix} 
\end{equation}
where $D = \dim \mathcal{H}$, and $V_{m, ij} = \braket{ij| H_{AB}| m}$. In this block form, we can easily write down exact characteristic equations for the eigenvalues $E_a$ and the eigenvectors $\ket{a}$:
\begin{equation}\label{eq:3.4heuristic_1}
    E_a - \bar E_{ij} = \sum_{m > 1} \frac{|V_{m,ij}|^2}{E_a - \epsilon_m} \,,
\end{equation}
\begin{equation}\label{eq:3.4heuristic_2}
    |c^a_{ij}|^2 = |\braket{a|ij}|^2 = \left( 1 + \sum_{m>1} \frac{|V_{m,ij}|^2}{(E_a - \epsilon_m)^2} \right)^{-1} \approx \left(\sum_{m>1} \frac{|V_{m,ij}|^2}{(E_a - \epsilon_m)^2} \right)^{-1} \,,
\end{equation}
where in the last identity, we used the fact that $|c^a_{ij}|^2$ is exponentially small in the system size to drop the factor of $1$ inside the bracket. 

 Let us define 
\be  \Gamma_{ij}(\epsilon_m) = \pi \rho(\epsilon_m) |V_{m,ij}|^2
\ee
and rewrite \eqref{eq:3.4heuristic_1} and \eqref{eq:3.4heuristic_2} in terms of $\Gamma_{ij}(\epsilon_m)$,  
\begin{equation}\label{eq:3.4heuristic_3}
    E_a - \bar E_{ij} = \frac{1}{\pi} \sum_m \frac{\Gamma_{ij}(\epsilon_m) \rho(\epsilon_m)^{-1}}{E_a - \epsilon_m}\,, \quad |c^a_{ij}|^{-2} \rho(E_a)^{-1} \approx \frac{1}{\pi \rho(E_a)}\sum_{m>1} \frac{\Gamma_{ij}(\epsilon_m) \rho(\epsilon_m)^{-1}}{(E_a - \epsilon_m)^2}  \,.
\end{equation}
In order to express the LHS of the second equation \eqref{eq:3.4heuristic_3} in terms of the LHS of the first, we need to understand the structure of $\Gamma_{ij}(\epsilon_m)$. Note that  
\begin{align} \label{v_conv}
|V_{m, ij}|^2 = |\sum_{kl\neq ij} \braket{m|kl}   \braket{kl|H_{AB}|ij}|^2 \,.
\end{align}
Since $\{ \ket{m} \}$ are the eigenstates of $H$ projected to $\sH_c$, $|\braket{m|ij}|^2$ as a function of $\epsilon_m- \bar E_{ij}$ should have a similar characteristic width to $f(\omega)$. From \eqref{v_conv}, 
$\Gamma_{ij}$ should then have a larger characteristic width than $f(\omega)$, as it involves a convolution of $\braket{m|kl}$  with the matrix element $\braket{kl|H_{AB}|ij}$, which has its own characteristic width as a function of $\bar E_{kl}- \bar E_{ij}$.  Let us therefore introduce an energy scale $\sigma$ such that $|\Gamma'_{ij}(\epsilon)|\ll1$ for $\epsilon - \bar E_{ij} \lesssim \sigma$.  Our approximations below will be for the range $|E_a - \bar E_{ij}| \lesssim \sigma$, which from the above argument is larger than the characteristic width
of $f(E_a - \bar E_{ij})$.


Now note that even though  $\Gamma_{ij}(\epsilon)$ and $\rho(\epsilon)$ are smooth,  the discrete sums in \eqref{eq:3.4heuristic_3} cannot be replaced with integrals over the smooth integrands due to the singular contributions near $\epsilon_m \approx E_a$. Instead, we use simplifying assumptions to directly evaluate the discrete sums.  
 Since $\Gamma_{ij}(E_a)$ varies slowly with $E_a$ for $E_a -\bar E_{ij} \lesssim \sigma$, let us approximate $\Gamma_{ij}(\epsilon_m) \approx \Gamma_{ij}(E_a)$ for the dominant terms in \eqref{eq:3.4heuristic_3}. Furthermore, let us assume that the levels $\epsilon_m$ are uniformly spaced with spacing $\rho(E_a)^{-1}$. This assumption is motivated by the repulsion between nearby energy levels in a generic chaotic system and is a standard assumption in the literature for related problems (see e.g.~\cite{Flambaum2000,kota_2014_embeddedRMT_quantum}). Under this assumption, the discrete sums in \eqref{eq:3.4heuristic_3} can be evaluated explicitly and we find
\begin{equation}
    \begin{aligned}
    E_a - \bar E_{ij} &
    = \Gamma_{ij}(E_a)\cot \left[\pi (E_a-\epsilon_a) \rho(E_a)\right] \,, \\
    |c^a_{ij}|^2 \rho(E_a) &\approx 
    \frac{1}{\pi \Gamma_{ij}(E_a) \left(1 + \cot^2 \left[\pi (E_a-\epsilon_a) \rho(E_a)\right] \right)} \,,
    \end{aligned}
\end{equation}
where $\epsilon_a$ is the level closest to $E_a$. Combining the two equations above, and noting that $\rho(E_a) \approx \rho\left(\frac{E_a + \bar E_{ij}}{2}\right)$ for $E_a - \bar E_{ij} = \mathcal{O}(1)$, we find 
\begin{equation}\label{eq:3.4_f_form_final}
    \boxed{f_{\rm approx}(E_a-\bar E_{ij}) = \frac{1}{\pi} \frac{\Gamma_{ij}(E_a)}{(E_a - \bar E_{ij})^2 + \Gamma_{ij}(E_a)^2} \,.} 
\end{equation}
For $E_a - \bar E_{ij} \lesssim \sigma$, since $|\Gamma_{ij}'(E_a)| \ll 1$,  \eqref{eq:3.4_f_form_final} implies a Lorentzian form of $f(\omega)$ with width $\Gamma_{ij}(\bar E_{ij})$. From the definition of $\Gamma_{ij}$, we can see that this width is $\sO(1)$ in a local system in one spatial dimension. For $E_a - \bar E_{ij} \gtrsim \sigma$, $|\Gamma_{ij}'(E_a)| \ll 1$ no longer holds and we expect deviations from \eqref{eq:3.4_f_form_final}.

  The uniform level  spacing assumption used above, though standard in the literature, is uncontrolled given that small deviations from uniform spacing lead to large fluctuations in $1/(E_a - \epsilon_m)$ when $\epsilon_m \approx E_a$. In Appendix \ref{app:EDF_derivation_details}, we discuss alternative approximations which can also explain the Lorentzian regime of $f(\omega)$. In these approximations, we evaluate the sums in \eqref{eq:3.4heuristic_3} by making certain assumptions about the functional form of $\Gamma_{ij}(\epsilon)$.  We show that the Lorentzian form of $f(\omega)$ is robust to different choices of the functional form of $\Gamma_{ij}(\epsilon)$.

We confirm various assumptions going into the above argument numerically  for our spin chain model. In Fig. \ref{fig:BreitWigner_test2}, we show that the width of $\Gamma_{ij}(\epsilon_m)$ is indeed larger than that of $f(\omega)$. Moreover, we explicitly compare the LHS and RHS of \eqref{eq:3.4_f_form_final} using the numerically derived forms of $f(\omega)$ and $\Gamma_{ij}(\epsilon)$, and find excellent agreement with \eqref{eq:3.4_f_form_final} up to $E_a - \bar E_{ij}$ a few times larger than $\sigma$. This range is even larger than warranted by the analytic arguments.

\begin{figure}[!h]
    \centering
    \includegraphics[width = 0.7\textwidth]{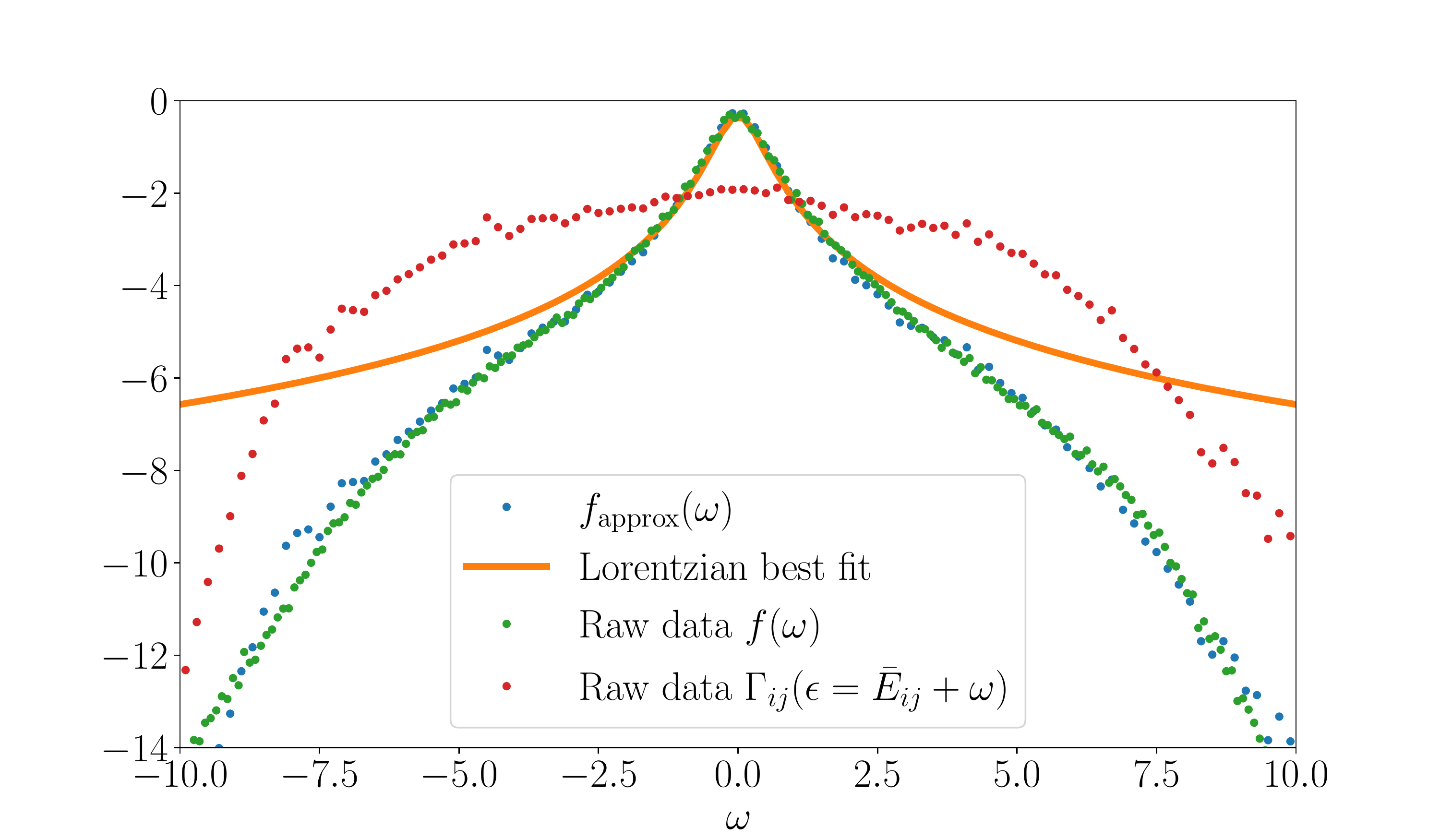}
    \caption{We numerically extract the form of $\Gamma_{ij}(\epsilon) = \pi \rho(\epsilon_m) |V_{m,ij}|^2$ as well as $f(\omega)$. $\Gamma_{ij}(\epsilon_m)$ fits well to a Lorentzian with width $\sigma \approx 3.4$, confirming the fundamental assumption we make in Sec.~\ref{subsec:derivation_f_form} that the characteristic width of $\Gamma_{ij}(\epsilon)$ is much larger than that of $f(\omega)$. Plugging $\Gamma_{ij}(\epsilon)$ into $f_{\rm approx}(\omega)$ gives a good numerical fit to the EDF for $|\omega| \lesssim 6$, which is larger than the range of validity $|\omega| \lesssim \sigma$ of the analytic arguments.}
    \label{fig:BreitWigner_test2}
\end{figure}

Notice that the above argument for the Lorentzian should hold quite generally for chaotic systems, and the input about a local system in 1+1 dimensions is used only in the final step, to deduce that the width of the Lorentzian is $\sO(1)$. The same argument applies in local systems in higher dimensions, where the width of the Lorentzian grows with the area of the boundary of $A$. The energy range $\sigma$ for which the Lorentzian is valid also grows with the area, and is still larger than the width of the Lorentzian.~\footnote{Note that this does not agree with the prediction of a Gaussian form in higher dimensions in \cite{murthy_2019_structure}, which we discuss in Appendix \ref{app:MSlimitation}.} Similarly, systems with long-range interactions in 1+1 dimensions are distinguished from local systems by the fact that in the former, the width of the Lorentzian scales with the system size.

We are not able to analytically determine the asymptotics of $f(E_a - \bar E_{ij})$ at $|E_a - \bar E_{ij}| \gtrsim \sigma$, but the exponential decay observed in numerics is consistent with the rigorous upper bounds that we derive in Sec.~\ref{subsec:loc}. Moreover, the rate of this exponential decay can be estimated via the variance of the EDF
\begin{equation}
    \sigma_{E,ij}^2 = \bra{ij} H_{AB}^2 \ket{ij} - \bra{ij}H_{AB}\ket{ij}^2 = \int d \omega f(\omega)\, \omega^2 \,. 
\end{equation}
If the Lorentzian form for $f(\omega)$ is used, the $\omega$-integral is divergent. This means that the true variance is determined by the exponentially decaying regime of $f(\omega)$. Assuming $f(\omega) \sim e^{- \frac{|\omega|}{\omega_*}}$ at large $\omega$, the constraint above implies that $\omega_* \leq \sigma_{E,ij}/\sqrt{2}$, which is consistent with numerical simulations in Fig.~\ref{fig:f_form}. 

\subsection{Upper bound on \texorpdfstring{$f(\omega)$}{} from locality}
\label{subsec:loc}

In this subsection, we show the exponentially decaying lower bound \eqref{ub_1} on $f(\omega)$, using the results of \cite{arad_2016_connectinglocalglobal}.



In the discussion below we will make use of $||O||$, the operator norm of $O$,
defined as 
\begin{align}  
||O|| & = \text{sup}_{||\ket{\psi}|| =1 } \, || O \ket{\psi} ||   = \max_{||\ket{\psi}||=1, ||\ket{\phi}||=1 }|\braket{\psi| O |\phi}|, \quad \quad ||\ket{\psi}|| = \sqrt{\braket{\psi| \psi}} \,.
\end{align} 
We define the following projectors onto the eigenstates of $H$ and $H_{A}+ H_B$ respectively in some energy interval $I$:  
\be 
\Pi_I = \sum_{E_a \,  \in \, I} \ket{a}\bra{a}, \quad \quad Q_I = \sum_{E_{A, i}+ E_{B,j} \, \in \, I} \ket{i}\ket{j}\bra{i}\bra{j} \, . 
\ee

By a slight modification of the analysis in \cite{arad_2016_connectinglocalglobal}, we can show that for any $E_{AB}$, $E$ such that $E_{AB} > E$, and for any $\delta>0$, we have
\be 
|| \, Q_{[E_{AB}, \,  E_{AB} + \delta]} \,  \Pi_{[E-\delta, \,  E]} \, || \leq C e^{- \gamma(E_{AB} - E)}, \quad E_{AB} > E \label{a1}
\ee
This statement applies to lattice models with local interactions in  any number of dimensions, for a somewhat more general notion of locality defined in \cite{arad_2016_connectinglocalglobal}, and to both chaotic and integrable systems. In particular, it applies to Hamiltonians with spatially local interactions such as \eqref{ham}.  Here 
$\gamma$
is an $\sO(1)$ constant associated with the parameters of the lattice model, which is defined explicitly in \cite{arad_2016_connectinglocalglobal}, and referred to as $\lambda$ in their notation. The constant $C$ depends on the quantity 
\be 
\partial A = \sum_{h_x \text{ in } H_{AB}} || h_x || 
\ee 
which is proportional to the area of the boundary of $A$. In the discussion below, we will absorb various constants that are $\sO(1)$ or that scale with $\partial A$ in $C$. 
Similarly, for any $E_{AB}$, $E$ such that $E_{AB} < E$, and $\delta>0$, we have 
\be 
|| \, Q_{[E_{AB} -\delta, \,  E_{AB}]} \,  \Pi_{[E, \,  E+\delta]} \, || \leq C e^{- \gamma(E- E_{AB})}, \quad E_{AB} < E \, . \label{a2}
\ee

Let us now assume that we have a chaotic spin chain, where the approximation \eqref{eq:f_universal} is valid, and see what constraints we get on the eigenstate distribution function $f(\omega)$ from \eqref{a1} and \eqref{a2}. For $E_A + E_B > E$, 
consider the operator appearing on the LHS of \eqref{a1}, 
\be
O = Q_{[E_{AB}, \,  E_{AB} + \delta]} \,  \Pi_{[E-\delta, \,  E]} = \sum_{\substack{E_{AB} \leq E_i + E_j \leq E_{AB} + \delta \, , \\  \\ E-\delta \leq E_a \leq E }} \, {c^a_{ij}}^{\ast} \ket{i} \ket{j} \bra{a} 
\ee
Note that 
\begin{align} 
\Tr[O] & = \sum_{E- \delta \leq E_{a'} \leq E} \braket{a'|O|a'} = \sum_{E_{AB}\leq E_{i'} + E_{j'} \leq E_{AB}+\delta} \braket{i'j'|O |i'j'} \label{l1} \\ 
& = \sum_{\substack{E_{AB} \leq E_{i} + E_{j} \leq E_{AB} + \delta \, , \\  \\ E-\delta \leq E_a \leq E }}  |c^a_{ij}|^2 \label{l2}
\end{align} 

Using the first expression in \eqref{l1} together with \eqref{a1}, we find 
\be 
\Tr[O] \leq N_{[E-\delta, E]} \,  C  e^{-\gamma (E_{AB}-E)}
\ee
where $N_{[E-\delta, E]}$ is the number of eigenstates of $H$ in the interval $[E-\delta, E]$. Taking $\delta$ to be small enough such that the density of states can be approximated as a constant in the interval, we have 
\be 
N_{[E-\delta, E]} = \int_{E-\delta}^E  \, dE' \, e^{S(E')} \approx e^{S(E)} \delta \, . 
\ee
By similarly using the second expression in \eqref{l1}, we end up with the following upper bound for $\Tr[O]$ (absorbing the $\sO(1)$ constant $\delta$ in $C$): 
\be \label{omin}
\Tr[O] \leq \min \le(e^{S(E)}, e^{S_{\rm fac} (E_{AB})} \ri) \,  C e^{-\gamma(E_{AB}- E)} 
\ee
where $S_{\rm fac}(E)$ is the thermodynamic entropy of $H_A + H_B$. Further, assume that the window $\delta$ is small enough such that the expectation value $\bra{i}\bra{j} H_{AB} \ket{i} \ket{j}$ for all states $\ket{i}\ket{j}$ with $E_{AB}\leq E_i + E_j \leq E_{AB}+\delta$ is approximately equal, and call this value $\tilde E_{AB}$. Then by absorbing a factor $e^{\gamma ||H_{AB}||}$ in $C$, we find 
\be 
\Tr[O] \leq \min (e^{S(E)}, e^{S_{\rm fac}(E_{AB})}) C e^{-\gamma(\bar E_{AB}-E)}
\ee
where $\bar E_{AB} = E_{AB} + \tilde{E}_{AB}$. 

To complete the argument, we need to relate $\Tr [O]$ to the EDF $f(\omega)$. We first observe that $e^{S_{\rm fac}(E)} \approx e^{S(E)}$ up to $\sO(1)$ multiplicative factors (see Appendix~\ref{app:therm} for a detailed argument). This approximation allows us to express \eqref{l2} in terms of the density of states, using $E_{ij}$ as shorthand for $E_i+E_j$: 
\be 
\Tr[O] \approx \int_{E-\delta}^{E} dE_a \, \int_{E_{AB}}^{E_{AB}+ \delta}  d E_{ij} \, e^{S_{\rm fac}(E_{ij}) + S(E_a)} |c^{a}_{ij}|^2 \approx \int_{E_{AB}}^{E_{AB}+ \delta}  d E_{ij} \, e^{S(E_{ij}) + S(E_a)} |c^{a}_{ij}|^2 \, . \label{troc} 
\ee
Let us now expand both $S(E_{ij})$ and $S(E_a)$ around $S(\bar E)$, where $\bar E = \frac{\bar E_{ij}+E_a}{2}$ and $\bar E_{ij} = E_{ij} + \Delta_{ij} = E_{ij} +\bra{i}\bra{j} H_{AB} \ket{i}\ket{j}  $, and assume that $\omega = E_a - \bar E_{ij}$ is less than extensive. Define $\beta = S'(\bar E)$. Then we find, on changing variables to $\bar E$ and $\omega$, 
\begin{align} 
\Tr[O] &\approx \int_{\frac{\bar E_{AB} + E}{2} - \delta/2}^{\frac{\bar E_{AB} + E}{2} + \delta/2} d\bar{E} \, e^{S(\bar E)} e \int_{E-\bar E_{AB}-2\delta}^{E- \bar E_{AB}} d\omega \, e^{-\beta \tilde E_{ij}} f(\bar E, \, \omega) \, \\
&\approx 2\delta^2 \,  e^{S\le( \frac{\bar E_{AB} + E}{2} \ri)} f\le( \frac{E + \bar E_{AB}}{2}, E- \bar E_{AB} \ri) e^{-\beta \tilde E_{AB}} \label{trof}
\end{align} 
where in the final form we assume that $\delta$ is sufficiently small that $S(\bar E)$, $f(\bar E, \omega)$, and $\tilde E_{ij}$ can all be treated as constants in the above windows. 

Combining \eqref{trof} with \eqref{omin}, we obtain the following upper bound on the eigenstate distribution function: 
\be \label{ub}
f\le( \frac{E + \bar E_{AB}}{2}, E- \bar E_{AB} \ri) \leq C e^{- \le(\frac{|\beta|}{2}+\gamma \ri) |E - \bar E_{AB}|} \, ,  \quad \beta = S'\le(\frac{E + \bar E_{AB}}{2}\ri) \, . 
\ee
Here we have absorbed $1/2\delta^2$ and factors of $e^{\beta \tilde E_{AB}}$ in $C$. 
Starting with \eqref{a2} in the case $E_{AB}<E$, we find the same bound.

\section{Transition probability and correlations from unitarity} \label{sec:phco}

To understand the behaviour of quantities like the the transition probability $P_{ij,xy}(t)$ in terms of the coefficients $c^a_{ij}$, it is useful to see the $c^a_{ij}$ as being drawn from a statistical ensemble. Averages over this statistical ensemble are realized by averages over an $\sO(1)$ number of $a$, $i$, $j$ close to some energies $E_a$, $E_i$, $E_j$. 
In this section, we show that statistical correlations among the coefficients are important for understanding the evolution of the transition probability. We also show numerically that further correlations are present beyond the minimal ones needed to explain the transition probability.   

\subsection{Transition probability} 

Let us consider the time evolution of the transition  probability 
\bega 
P_{ij, xy} (t) = \le|\vev{xy|e^{- i Ht} |ij} \ri|^2 =  \sum_{a,b} e^{- i (E_a - E_b) t} c_{ij}^a c_{xy}^{a*} c_{ij}^{b*} c_{xy}^b  \ , \quad x y \neq ij \, . 
\label{euh}
\end{gather} 
The initial value of this quantity is zero. Assuming that the spectrum is non-degenerate, the late-time saturation value comes from the $a=b$ terms in the above sum: 
\be 
P_{ij,xy}(t\rightarrow \infty) = \sum_a |c^a_{ij}|^2  |c^a_{xy}|^2 \,.
\end{equation}
Approximating this sum as an integral, we see that the saturation value is  the convolution of the eigenstate distribution function with itself:
\begin{equation}
    P_{ij,xy}(t\rightarrow \infty) \approx e^{-S\left(\frac{E_{ij} + E_{xy}}{2}\right)} \int d E_a f\left(\frac{E_a + \bar E_{ij}}{2}, E_a - \bar E_{ij}\right) f\left(\frac{E_a + \bar E_{xy}}{2}, \bar E_{xy} - E_a\right) \,.
\end{equation}
Neglecting the weak dependence of the eigenstate distribution function on the first argument and approximating the dependence on the second argument as a Lorentzian, we can do the integral explicitly and find
\begin{equation}\label{eq:Pijxy}
    \begin{aligned}
    P_{ij,xy}(t\rightarrow \infty) &\approx e^{-S\left(\frac{E_{ij} + E_{xy}}{2}\right)} \int d E \frac{C_0}{\Gamma^2 + (E_{ij} - E)^2} \frac{C_0}{\Gamma^2 + (E_{xy} - E)^2} \\
    &\approx e^{-S\left(\frac{E_{ij} + E_{xy}}{2}\right)} \frac{2\pi C_0^2}{\Gamma} \frac{1}{(2\Gamma)^2 + (E_{ij} - E_{xy})^2} \,,
    \end{aligned}
\end{equation}
which is another Lorentzian with width $2\Gamma$. We verify this behaviour  numerically in Fig.~\ref{subfig:P_ijxy_latetime}. The best fit width $0.89$ is indeed close to twice the best fit width $0.47$ of the eigenstate distribution function found in Fig.~\ref{fig:f_form}.

\begin{figure}[!h]
    \centering
    \begin{subfigure}{0.49\textwidth}
        \centering
        \includegraphics[width = \textwidth]{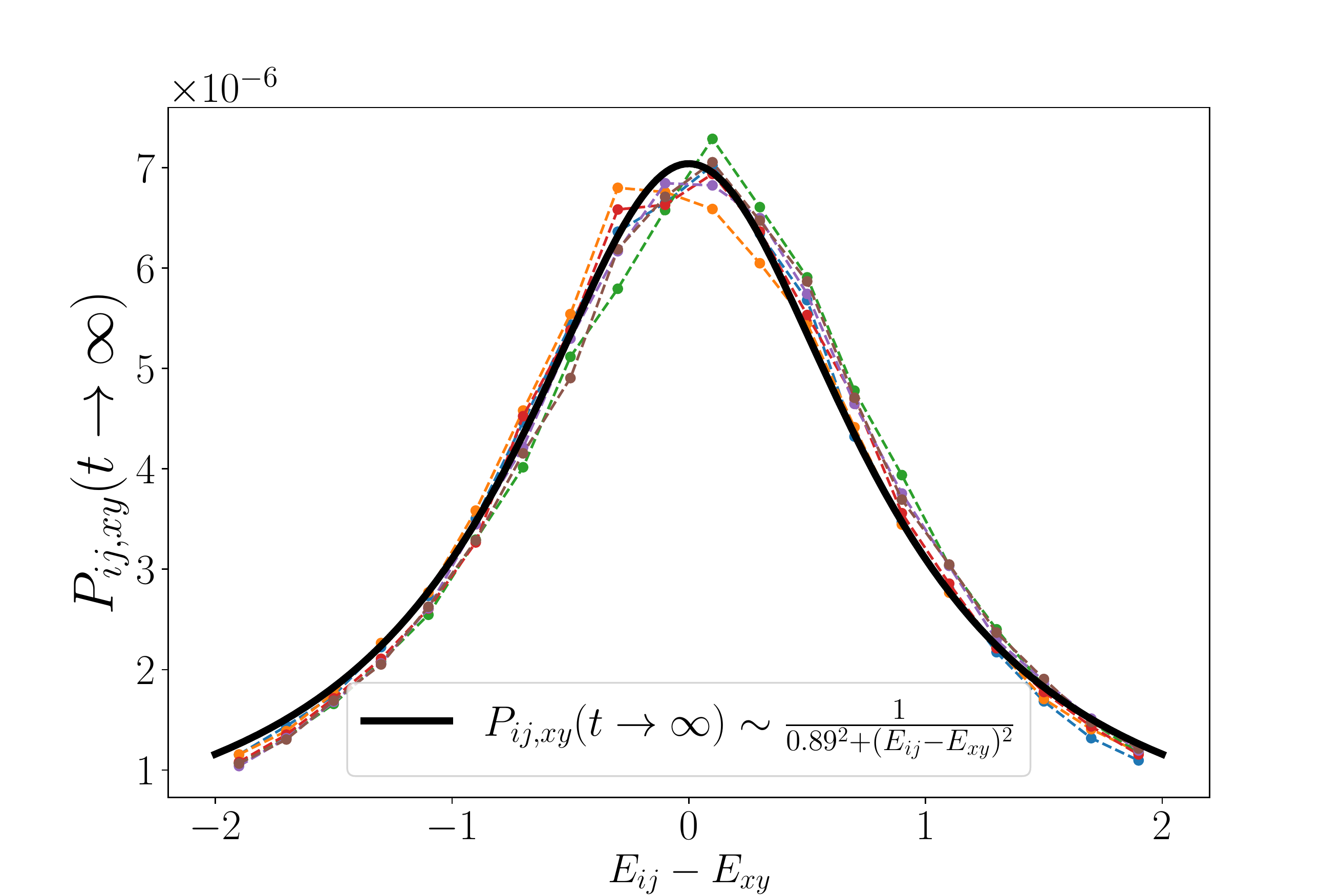}
        \caption{Late time saturation}
        \label{subfig:P_ijxy_latetime}
    \end{subfigure}
    \begin{subfigure}{0.49\textwidth}
        \centering
        \includegraphics[width = \textwidth]{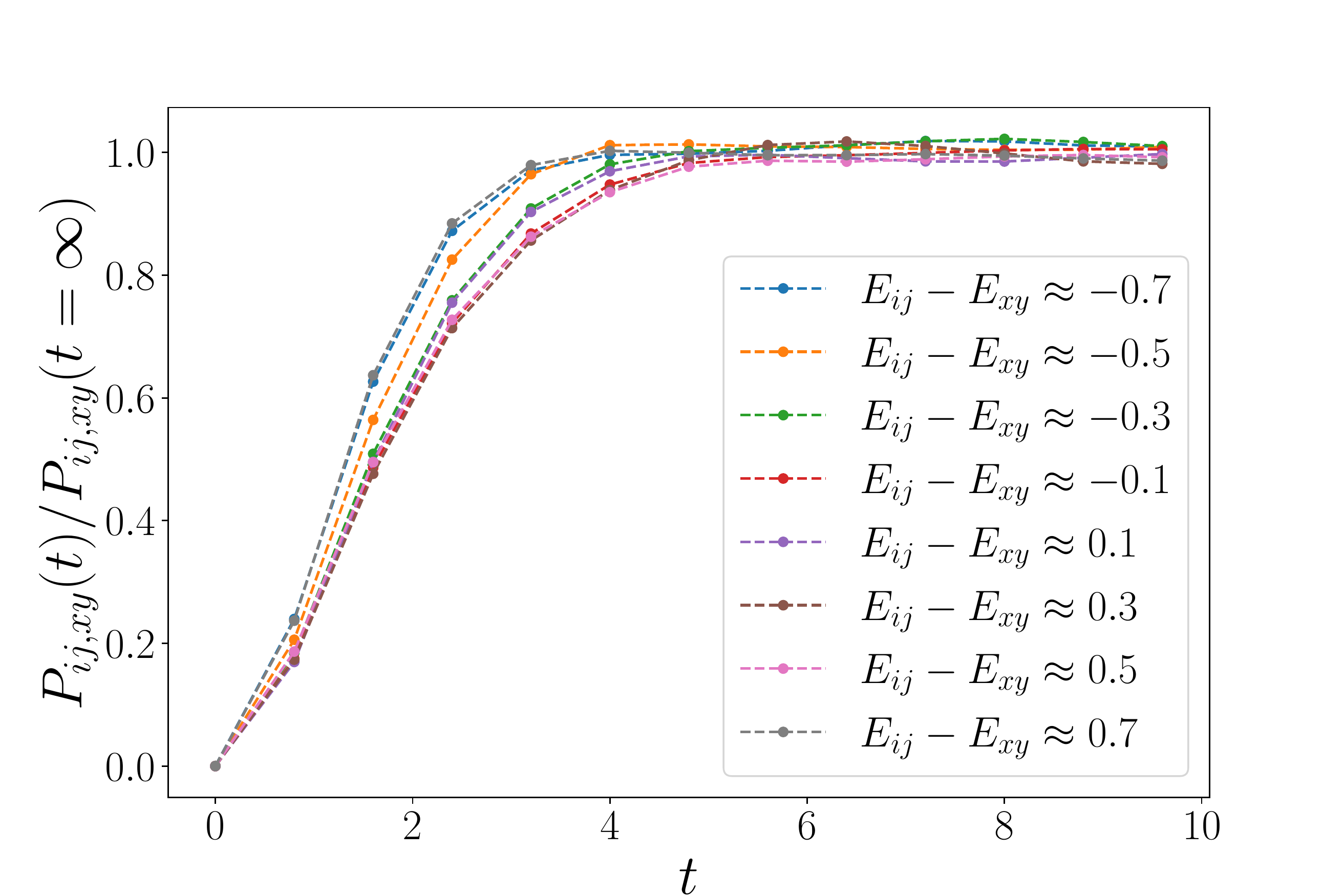}
        \caption{Time evolution}
        \label{subfig:P_ijxy_tdomain}
    \end{subfigure}
    \caption{We numerically simulate $P_{ij,xy}(t)$ for the mixed-field Ising model at $L = 20$, fixing $ij, xy$ to be near the middle of the spectrum. The equilibrium profile $P_{ij,xy}(t \rightarrow \infty)$ is found to be a Lorentzian of $E_{ij} - E_{xy}$ with width consistent with the analytic prediction. The normalized time evolution $P_{ij,xy}(t)/P_{ij,xy}(t = \infty)$ is a nontrivial smooth function of time which depends weakly on $E_{ij} - E_{xy}$. This nontrivial time-dependence is a manifestation of spatial locality.}
    \label{fig:Pijxy}
\end{figure}

Next we turn to the dynamics of $P_{ij,xy}(t)$ in the approach towards equilibrium. A simple approximation is to assume that $c^a_{ij}$ are taken from a statistical ensemble
\be  \label{rp_approx_1}
\overline{c^a_{ij} {c^b_{kl}}^{\ast}} = \overline{|c^a_{ij}|^2}  \de_{ab} \de_{ik} \de_{jl} , \quad
\overline{|c^a_{ij}|^2}  =
e^{-S\le(\frac{E_a + \bar E_{ij}}{2}\ri)} f\le(\frac{E_a + \bar E_{ij}}{2}, E_a -  \bar E_{ij}\ri) \ .
\ee
The ansatz \eqref{rp_approx_1} is an improvement over the EB ansatz~\eqref{erg} as it incorporates the non-trivial functional form~\eqref{eq:f_universal} 
of $\overline{|c^a_{ij}|^2}$, but it still assumes that $c^a_{ij}$ are independent Gaussian random variables.

Applying~\eqref{rp_approx_1} to~\eqref{euh}, we find that at all times $t$,
\be 
\overline{P_{ij, xy} (t)}  =  \sum_{a}  \overline{|c_{ij}^a|^2} \, \overline{|c_{xy}^{a}|^2  } = P_{ij,xy}(t\rightarrow \infty) \ .
\ee
The right hand side is time-independent, which tells us that under such an approximation, $P_{ij, xy} (t)$ already saturates to the equilibrium value for any time $t \sim \sO(L^0)$. However, as discussed in the introduction, we expect $P_{ij, xy} (t)$ to have nontrivial evolution dynamics for $t \sim \sO(L^0)$ due to the locality of $H_{AB}$. Numerical simulations of $P_{ij,xy} (t)$ shown in Fig.~\ref{subfig:P_ijxy_tdomain} are consistent with this expectation.
Thus~\eqref{rp_approx_1} is inadequate, and further structure must be included to capture effects of local dynamics. In the rest of this section, we discuss how to improve on~\eqref{rp_approx_1} by including correlations among $c^a_{ij}$. 


\subsection{Constraints from unitarity}\label{sec:31}

To capture the nontrivial evolution of $P_{ij, xy} (t)$, we first need to be able to capture that at $t=0$, it is zero
to order $O(e^{-S})$ due to orthogonality of $\ket{ij}$ and $\ket{xy}$.\footnote{Since the  saturated value is $O(e^{-S})$, we only need to work to this order.}

For this purpose, it is helpful to first study a simple toy model. 
Consider a system without any energy constraints, with $d$ the dimension of the Hilbert space. 
Consider two random orthogonal basis $\{\ket{a}\}$ and $\{\ket{m}\}$, which are related by 
\be 
\ket{m} = \sum_a u^a_m \ket{a} \ .
\ee
Let us first assume $u^a_m$ are independent Gaussian random variables, with
\be 
\overline{u^a_m u^{b*}_n} ={1 \ov d} \de_{ab} \de_{mn} ,
\ee
under which we have 
\be 
\overline{\vev{m_1|m_2}} = \sum_a \overline{u^a_{m_1} u^{a*}_{m_2}} = \de_{m_1 m_2} ,  
\ee
with the variance for $m_1 \neq m_2$ given by 
\be \label{overlap_sum}
\overline{|\braket{m_1|m_2}|^2} = \sum_{a_1 b_1} \overline{u^{a_1}_{m_1} {u^{a_1}_{m_2}}^{\ast} u^{b_1}_{m_2} {u^{b_1}_{m_1}}^{\ast}} 
= \sum_a \overline{|u^a_{m_1}|^2} \,  \overline{|u^a_{m_2}|^2} = \frac{1}{d} \ . 
\ee
  
While~\eqref{overlap_sum} is suppressed by $1/d$, there is something fundamental missing. Since $u$ is a unitary matrix, $u^a_m$ cannot be genuinely independent. A better approximation is to treat $u^a_m$ as a random unitary. From the standard results for a Haar random unitary 
we have 
\begin{align} 
\overline{u^{a_1}_{m_1} {u^{a_2}_{m_2}}^{\ast} u^{b_1}_{n_1} {u^{b_2}_{n_2}}^{\ast}} &= \frac{1}{d^2} (\delta_{a_1 a_2} \delta_{b_1 b_2} \, \delta_{m_1 m_2} \delta_{n_1n_2}+\delta_{a_1 b_2} \delta_{a_2 b_1} \delta_{m_1 n_2} \delta_{n_1 m_2 }) \nonumber \\
& - \frac{1}{d^3} (\delta_{a_1 a_2}\delta_{b_1 b_2} \delta_{m_1 n_2} \delta_{n_1m_2} + \delta_{a_1 b_2}\delta_{b_1 a_2} \delta_{m_1 m_2} \delta_{n_1 n_2}) + \cdots ,  \label{full_u}
\end{align} 
where $\cdots$ denotes higher order corrections in $1/d$. The first line of~\eqref{full_u} is the same as that for treating $u^a_m$ as independent Gaussian variables. The second line comes from correlations among different $u^a_m$'s. Naively, the second line can be neglected, as it is suppressed by an additional factor $1/d$. That is incorrect; while in the second line an individual term is suppressed compared with the first line, there are many more combinations of indices that have non-vanishing contribution. As a result, when calculating the variance $\overline{|\braket{m_1|m_2}|^2}$, we find that the contribution from the second line is of the same order as, and in fact exactly cancels, that from the first line, leading to the orthogonality at order $O(1/d)$ 
\be\label{oth1}
\overline{|\braket{m_1|m_2}|^2} =0 , \quad m_1 \neq m_2 \ .
\ee

Let us now return to the spin chain system and the coefficients $c^a_{ij}$. The story here is more intricate, as there is a  nontrivial interplay between the correlations from unitarity and constraints from energy conservation as well as locality. 
To take into account the correlations from unitarity, we may consider the following analog of~\eqref{full_u} (we use the $m$ indices as shorthand for $ij$ indices): 
\begin{align} \label{c2_ru} 
\overline{c^{a_1}_{m_1} {c^{a_2}_{m_2}}^{\ast} c^{b_1}_{n_1} {c^{b_2}_{n_2}}^{\ast}} &= \overline{ |c^{a_1}_{m_1}|^2} \, \overline{ |c^{b_1}_{n_1}|^2 } (\delta_{a_1 a_2} \delta_{b_1b_2} \delta_{m_1m_2} \delta_{n_1 n_2} + \delta_{a_1 b_2} \delta_{b_1a_2} \delta_{m_1n_2} \delta_{n_1 m_2} )\cr
& - k_{a_1, b_1, m_1, n_1} (\delta_{a_1 a_2} \delta_{b_1 b_2} \delta_{m_1n_2}\delta_{n_1m_2} + \delta_{a_1 b_2} \delta_{b_1 a_2} \delta_{m_1m_2}\delta_{n_1n_2} ) ,
\end{align}
The key difference from \eqref{full_u} is that  $\overline{|c^a_m|^2}$ are no longer constants with respect to $a$ and $m$, and are given by~\eqref{rp_approx_1}. Hence, $k_{a_1, b_1, m_1, n_1}$ also  cannot be constants with respect to $a_1, b_1, m_1, n_1$, 
but should be constrained by energy differences.  

In analogy with~\eqref{oth1},  $k_{a_1, b_1, m_1, n_1}$ are required to ensure $\overline{|\braket{m_1|m_2}|^2}=0$ and $\overline{|\braket{a|b}|^2}=0$, i.e. 
\begin{align}
& a \neq b \,  :~~~ \sum_{m_1, m_2} k_{a, b, m_1, m_2} = \sum_{m_1} \overline{|c^a_{m_1}|^2} \,  \overline{|c^b_{m_1}|^2} , \label{314}\\
& m_1 \neq m_2 \,  :~~~ \sum_{a, b} k_{a, b, m_1, m_2} = \sum_{a} \overline{|c^a_{m_1}|^2}\,   \overline{|c^a_{m_2}|^2}  \ . \label{315}
\end{align} 
Since the right hand sides of~\eqref{314}--\eqref{315} are smooth functions of various energies, it is natural to expect that 
$k_{a,b,m_1, m_2}$ should be a smooth function of $E_a, E_b, \bar E_{m_1}, \bar E_{m_2}$ and have support only for cases where the energy differences between any pair is of order $\sO(1)$. That is,  we can parameterize it as 
\bega 
 k_{a,b,m_1, m_2} =e^{- S (E_{ab})- S (\bar E_{m_1 m_2}) - S ({E_{ab} + \bar E_{m_1 m_2} \ov 2})} g (\om_{ab}, \bar \om_{m_1 m_2}, 
 E_{ab} - \bar E_{m_1 m_2} ),
 \\
 E_{ab} = {E_a + E_b \ov 2}, \;\; \bar E_{m_1 m_2} =  {\bar E_{m_1} + \bar E_{m_2} \ov 2} , \;\;
 \om_{ab} = E_a - E_b, \;\; \bar \om_{m_1 m_2} =  \bar E_{m_1} -  \bar E_{m_2}  \, .
\end{gather} 
Equations~\eqref{314}--\eqref{315} give 
\bega 
 \int d\om d \om' \,  e^{{\b \ov 2} \om'}  
g (\om_{ab}, \om, - \om') = \int d\om \, f (\om) f (\om - \om_{ab}) , \\
\int d \om d \om' \, e^{{\b \ov 2} \om'} g (\om, \bar \om_{m_1 m_2}, \om') 
= \int d\om \, f (\om) f (\om - \bar \om_{m_1 m_2})  \ .
\end{gather}

Applying~\eqref{c2_ru} to~\eqref{euh} we have 
\bea \label{t_local}
\overline{|\braket{m_1|e^{-iHt}|m_2}|^2} &=&  \sum_{a,b} e^{- i (E_a - E_b) t} \overline{c_{m_1}^a c_{m_2}^{a*} c_{m_2}^{b} c_{m_1}^{b*} } , \quad m_2 = ij, \, m_1 = xy 
\\
&=&
\sum_a \overline{|c^a_{m_1}|^2} \,  \overline{|c^a_{m_2}|^2} - \sum_{a,b} e^{- i(E_a-E_b)t}k_{a,b, m_1, m_2}
\\
&=&  e^{- S (\bar E_{m_1 m_2})} \int d\om d \om' \,  e^{{\b \ov 2} \om'}  \le(1 - e^{- i \om t} \ri) 
g (\om, \bar \om_{m_1 m_2}, \om') \ .
\label{eu1}
\eea
where $\beta= S'\le(E= \frac{\bar E_{ij,kl}+ E_{ab}}{2}\ri)$. 
The above expression exhibits nontrivial time evolution for $t \sim O(L^0)$. 

\subsection{Further correlations} \label{sec:fco}

Now consider the average of 
a more general product of two amplitudes 
\bea \label{c2c}
 C^{(2)}(m_1, m_1'; m_2, m_2';t)   &\equiv&  \overline{\braket{m_1| e^{-iHt}|m_1'} \braket{m_2| e^{iHt}|m_2'} } \\
&=&  \sum_{a,b} e^{- i (E_a - E_b) t} \overline{c^a_{m_1} c^{a*}_{m_1'} c^b_{m_2} c^{b*}_{m_2'}}
\\
&=& \int \,  d \omega \, e^{-i \omega t} \,  C^{(2)}(m_1, m_1'; m_2, m_2'; \omega)   \label{prod_2}
\eea
where we have also introduced the Fourier transform~\cite{liao2022_ETH_RQC} 
\bega \label{c2def}
C^{(2)}(m_1, m_1'; m_2, m_2'; \omega) = \sum_{a,b} \delta(\omega- (E_a-E_b)) \overline{c^a_{m_1} {c^a_{m_1'}}^{\ast}c^b_{m_2} {c^b_{m_2'}}^{\ast}}   \ .
\end{gather} 
From~\eqref{c2_ru}, $C^{(2)}$ is zero unless there exists some permutation $\sigma$ such that $m_i'= m_{\sigma(i)}, i=1,2$. That is, it is nonzero (i.e. at least of order $O(e^{-S})$) only for the two situations we considered already:

\ben 

\item For $m_1 = m_1' = m_2 = m_2'= m$, we have~\eqref{ptdef}, with the leading contribution given by~\eqref{ptf}, which is $\sO(1)$. 

\item For $m_1 = m_2' \neq m_2 = m_1'$, we have~\eqref{euh}, which is given by~\eqref{eu1} (of order $\sO(e^{-S})$). 

\een

We may wonder whether there are further correlations among $c^a_{ij}$ that are not captured by~\eqref{c2_ru}. Note that equation~\eqref{c2_ru} treats index $m=ij$ as a whole, and is thus ignorant of the product nature of the basis $\{ij\}$. 
Consider, for example, 
\be\label{ehn}
C^{(2)}(x_1y_1, ij; \, ij, x_2y_1;t) = \overline{\braket{x_1 y_1| e^{-iHt}|ij} \braket{ij| e^{iHt}|x_2 y_1} }, \quad x_1 \neq x_2 
\ee
which has $m_2 = m_1'$, but $m_1$ and $m_2'$  coincide partially (for the $B$ subsystem). Equation~\eqref{c2_ru} would imply that it is of higher order than $\sO(e^{-S})$, but numerical simulations show that $C^{(2)}(\ldots; t)$ for the index structure in \eqref{ehn} is comparable in magnitude to $P_{ij,xy}(t)$. Furthermore, Fig.~\ref{fig:C2_corr1} suggests that \eqref{ehn} depends smoothly on its energy arguments. We can parameterize the corresponding Fourier transform as 
\be 
C^{(2)}(x_1y_1, ij; \, ij, x_2y_1; \om) = e^{- S (\bar E_{ij})} 
\tilde g (\bar E_{x_1 y_1} - \bar E_{ij}, \bar E_{x_2 y_1} - \bar E_{ij}; \om) 
\ee 
where $\tilde g$ is a smooth function which has support only for energy differences of $O(L^0)$. 

These numerical results indicate correlations among $c^a_{ij}$ that go beyond~\eqref{c2_ru}, which in turn results in correlations between the transition amplitudes $\vev{x_1 y_1|e^{- i Ht}| ij}$ and $\vev{ij|e^{i Ht}| x_2 y_1}$.   We will see further examples of such correlations from studies of the Renyi entropy in the next section, where it is possible to tie them to locality. 


\begin{figure}[!h]
    \centering
    \includegraphics[width = 0.49 \textwidth]{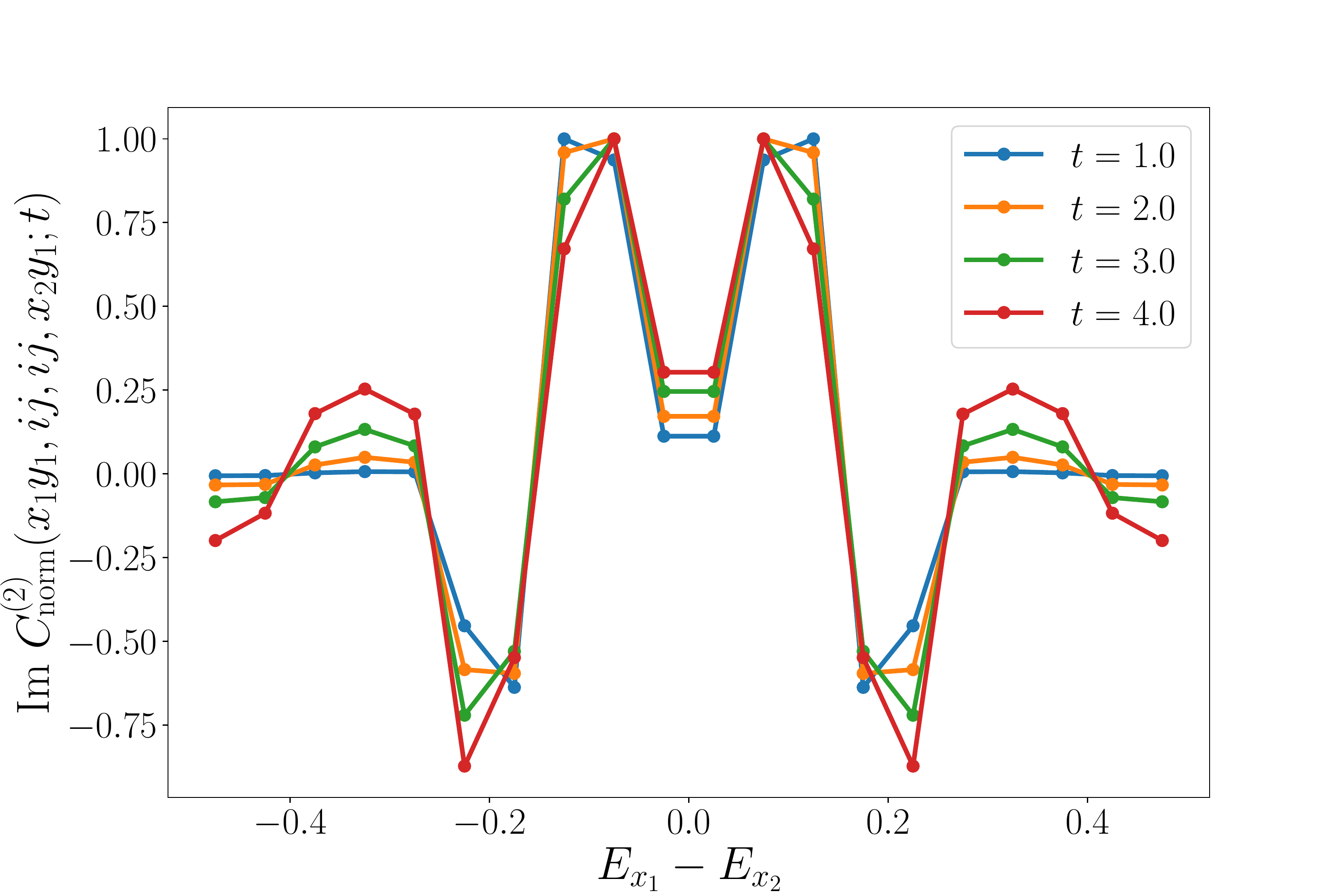}
    \includegraphics[width = 0.49 \textwidth]{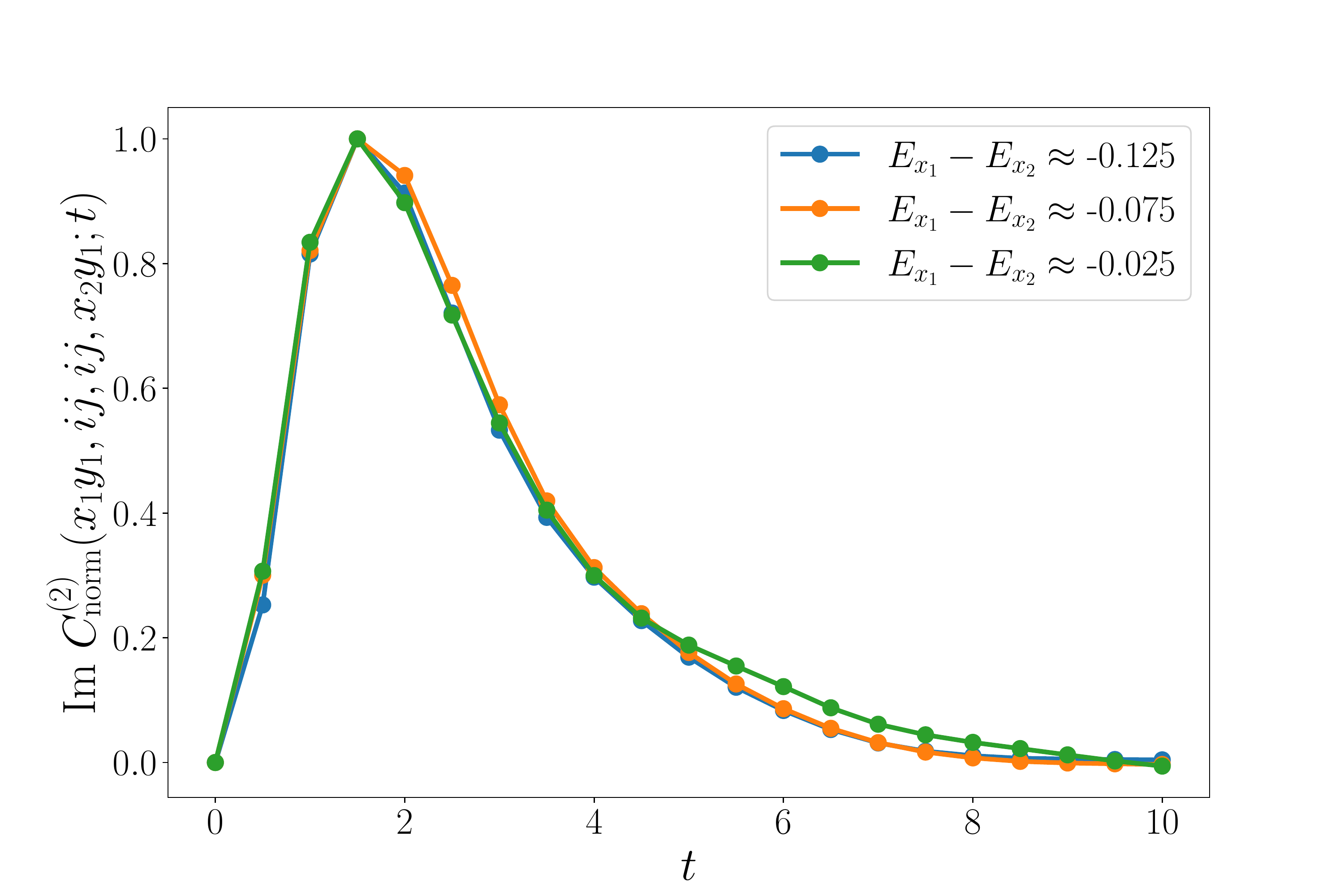}
    \caption{We numerically simulate $C^{(2)}(x_1y_1, ij, ij, x_2y_1; t)$ for the mixed-field Ising model at $L = 20$. We fix $i,j$ near the middle of the spectrum, average $E_{y_1}, (E_{x_1} + E_{x_2})/2$ over microcanonical windows, and study the dependence of $C^{(2)}(x_1y_1, ij, ij, x_2y_1; t)$ on $E_{x_1} - E_{x_2}$ and $t$. After normalizing the curves by their maximum values, we observe smooth dependence of this function on $E_{x_1}-E_{x_2}$ and on $t$. We only show the imaginary part for visual clarity, but the real part is also a smooth function of each of the arguments. The real and imaginary parts  are comparable in magnitude but have different functional forms.}
    \label{fig:C2_corr1}
\end{figure}

\subsection{Averages of higher products}

We can generalize~\eqref{c2_ru} to averages of higher products of $c^a_m$'s. We will again model these on the behavior of random unitaries. 

The generalization of~\eqref{full_u} to the average of a product of $2k$ random unitary matrix 
variables $u^a_m$ is given by
\be 
\overline{u^{a_1}_{m_1} {u^{a_1'}_{m_1'}}^{\ast} u^{a_2}_{m_2} {u^{a_2'}_{m_2'}}^{\ast} ... u^{a_k}_{m_k} {u^{a_k'}_{m_k'}}^{\ast}} = \frac{1}{d^k} \sum_{\sigma, \tau \in \sS_k} \text{wg}(\sigma, \tau, k)\braket{a_1 a_1' \, a_2 a_2'\,  ... a_k a_k' \,  |\sigma}\braket{\tau|m_1 m_1'\,  ... m_k m_k' }  \ .
\label{uk}
\ee
In~\eqref{uk}, $\sS_k$ denotes the permutation group of $k$ objects, and $\ket{\sig}$ for an element $\sig \in \sS_k$ is defined by 
\be 
\braket{a_1 a_1' ... a_k a_k'|\sigma} = \delta_{a_1 a'_{\sigma(1)}}...\delta_{a_k a'_{\sigma(k)}} \, .  
\ee
$\text{wg}(\sigma, \tau, k)$ is the inverse of the matrix $g_{\sigma, \tau} = d^{c(\sigma \tau^{-1})-k}$, where $c(\sigma)$ is the number of cycles in $\sigma$.

Now for more general products of $c^a_m$ in the spin chain model, we can generalize~\eqref{c2_ru} based on the structure of~\eqref{uk} as follows,  
\begin{align} 
 &\overline{c^{a_1}_{m_1} {c^{a_1'}_{m_1'}}^{\ast} ... c^{a_k}_{m_k} {c^{a_k'}_{m_k'}}^{\ast}}  = \overline{|c^{a_1}_{m_1}|^2}\, \overline{|c^{a_2}_{m_2}|^2}...\overline{|c^{a_k}_{m_k}|^2} \sum_{\sigma\in \sS_k}\braket{a_1 a_1' \, a_2 a_2'\,  ... a_k a_k' \,  |\sigma}\braket{\sigma|m_1 m_1'\,  ... m_k m_k' } \nonumber \\
&+ \sum_{\sigma, \tau \in \sS_k, \sig \neq \tau} h (\sigma, \tau; \,  a_1, ..., a_k \, m_1, ..., m_k) \braket{a_1 a_1' \, a_2 a_2'\,  ... a_k a_k' \,  |\sigma}\braket{\tau|m_1 m_1'\,  ... m_k m_k' } , \label{b_model}
\end{align} 
where we have explicitly separated the diagonal and off-diagonal pieces in terms of permutations. 
Similar to \eqref{314} and \eqref{315}, the coefficients $h(\sigma, \tau; \,  a_1, ..., a_k \, m_1, ..., m_k)$ must obey certain consistency conditions. These coefficients are suppressed in powers of $e^{-S}$ relative to $\overline{|c^{a_1}_{m_1}|^2}...\overline{|c^{a_k}_{m_k}|^2}$, which is $\sO(e^{-k S})$. Note that the structure of~\eqref{uk} is such that the right hand side is zero unless there exists some $\sigma, \tau \in \sS_k$ such that $m_i'= m_{\tau(i)}$ and $a_i' = a_{\sig(i)}$. 

We can also generalize the ``correlation functions'' of amplitudes~\eqref{c2c} to higher orders, for example, involving products of four amplitudes 
\begin{align} 
&  C^{(4)}(m_1, m_1'; m_2, m_2'; m_3, m_3'; m_4, m_4';  t) \cr
 & \equiv 
\overline{\braket{m_1| e^{-iHt}|m_1'} \braket{m_2| e^{iHt}|m_2'} \braket{m_3| e^{-iHt}|m_3'}\braket{m_4| e^{iHt}|m_4'}}  \nonumber \\
&= \int \,  d \omega \, e^{-i \omega t} \,  C^{(4)}(m_1, m_1'; m_2, m_2'; m_3, m_3'; m_4, m_4';  \omega) \label{amp}  
\end{align}
where 
\begin{align} \label{c4def}
C^{(4)}(& m_1, m_1';  \, m_2, m_2'; \, m_3, m_3'; \, m_4, m_4'; \, \omega) \nonumber \\
=& \sum_{a_1, b_1, a_2, b_2} \delta(\omega - (E_{a_1} + E_{a_2}- E_{b_1} - E_{b_2})) \overline{{c^{a_1}_{m_1}} {c^{a_1}_{m_1'}}^{\ast} \, {c^{b_1}_{m_2}} {c^{b_1}_{m_2'}}^{\ast} \,  
{c^{a_2}_{m_3}} {c^{a_2}_{m_3'}}^{\ast} \,  {c^{b_2}_{m_4}} {c^{b_2}_{m_4'}}^{\ast} }   \ .
\end{align} 

If we assume~\eqref{b_model}, $C^{(4)}$  is zero unless there exists some $\sigma \in \sS_k$ such that $m_i'= m_{\sig(i)}$. Note that it is possible that $C^{(4)}$ is nonzero, but the  averages of any subgroup of factors are zero.  
 One example is: 
 \bega
 C^{(4)}(m_1, m_2; m_2, m_3; m_3, m_4; m_4, m_1; t)  \quad (m_1 \neq m_2 \neq m_3 \neq m_4) \cr
 = \overline{\braket{m_1| e^{-iHt}|m_2} \braket{m_2| e^{iHt}|m_3} \braket{m_3| e^{-iHt}|m_4}\braket{m_4| e^{iHt}|m_1} } \ .
 \end{gather}
 Note that the quantity under the average is not a manifestly positive number. The fact that the average is non-zero shows that its phase  is not random. 

Given the discussion of Sec.~\ref{sec:fco}, we may wonder whether there are additional correlations not captured by~\eqref{b_model} at the level $C^{(4)}$. We will find that this is indeed the case.

\section{Evolution of Renyi entropy and further correlations from locality}\label{sec:ent}

In this section, we use the evolution of the second Renyi entropy $S_{2,A}(t)$ of the state $\ket{i}_A\ket{j}_B$ in the subsystem $A$ 
to infer certain collective properties of $c^a_{ij}$. We will see there are further  correlations from locality beyond the minimal ones required from unitarity, which play a dominant role in the evolution of this quantity. 


Consider the expression for the Renyi entropy ~\eqref{ren1} with $n=2$, which can be written in terms of $c^a_{ij}$ as
\begin{gather} \label{s2sum_simp}
e^{-S_{2,A}(t)} = \sum_{x_1, x_2, y_1, y_2} \braket{x_1 y_1| e^{-i Ht} |ij}  \braket{ij| e^{i Ht} |x_2 y_1}   \braket{x_2 y_2| e^{-i Ht} |ij}  \braket{ij| e^{i Ht} |x_1 y_2} \\
= \sum_{\substack{x_1, y_1, \\ x_2, y_2}} \sum_{\substack{a_1, b_1, \\ a_2, b_2}} e^{-i t (E_{a_1} +E_{a_2}-E_{b_1}-E_{b_2})} c^{a_1}_{x_1 y_1}(c^{a_1}_{ij})^{\ast} c^{b_1}_{ij} (c^{b_1}_{x_2 y_1})^{\ast} c^{a_2}_{x_2 y_2}(c^{a_2}_{ij})^{\ast}  c^{b_2}_{ij}  (c^{b_2}_{x_1 y_2})^{\ast} \, .  
\label{s2full_1}
\end{gather}
From the equilibrium approximation \cite{Vardhan_Liu_2021}, $S_2^{(A)}(t)$ should grow from its initial value of zero to a late-time saturation value of  
\be 
S_2^{(A)} \approx \, \text{min}\le(S_2^{(A)}\le(\rho(\, \eij\, )\ri), S_2^{(B)}\le(\rho(\, \eij\, )\ri)\ri) \,  . 
\ee
The way in which this value is approached depends on the collective behavior of $c^a_{ij}$'s.



To understand the structure of~\eqref{s2sum_simp}, it is useful to separate the sum into different cases: 
 \begin{align} 
e^{-S_{2,A}(t)} 
&= |P(t)|^2 +  T_1 + T_2 + T_3 + T_4 \, , \label{s2} \\
T_1 &=
\sum_{x_1, y_1 \neq  y_2} |\braket{x_1 y_1| e^{-i Ht} |ij}|^2 \,  |\braket{x_1 y_2| e^{-i Ht} |ij}|^2 \, ,  \\  T_2 &= \sum_{x_1\neq x_2, y_1} |\braket{x_1 y_1| e^{-i Ht} |ij}|^2 \,  |\braket{x_2 y_1| e^{-i Ht} |ij}|^2  \, ,\\ T_3 &=  \sum_{x_1\, y_1 \neq i \, j} |\braket{x_1 y_1| e^{-i Ht} |ij}|^4 \, , \\  T_4 &=  \sum_{x_1 \neq x_2, y_1 \neq y_2}  \braket{x_1 y_1| e^{-i Ht} |ij}  \braket{ij| e^{i Ht} |x_2 y_1}   \braket{x_2 y_2| e^{-i Ht} |ij}  \braket{ij| e^{i Ht} |x_1 y_2}  \label{75}
\end{align}
Here $P(t)$ is the return probability $|\braket{ij| e^{-iHt}|ij}|^2$, which as discussed in Sec.~\ref{sec:pt} is $\sO(1)$ in the thermodynamic limit.

Now suppose we approximate $e^{-S_{2,A} (t)}$ by doing an ensemble average of the right hand side. 
If we use the model \eqref{b_model}, then each of the remaining terms in \eqref{s2} is either zero or exponentially small relative to $|P(t)|^2$:

\ben 

\item The terms $T_1$ through $T_3$ are non-zero and manifestly positive, but are small in the thermodynamic limit.

  For example, for the terms in $T_1$, the average in each term approximately factorizes between the two transition probabilities, and we know from \eqref{t_local} that these probabilities are $\sO(e^{-S})$. Hence we have a sum over $e^{3S/2}$ terms of $\sO(e^{-2S})$, which is $\sO(e^{-S/2})$ and can be ignored relative to $|P(t)|^2$. Similarly, $T_2 \sim \sO(e^{-S/2})$ and  $T_3 \sim \sO(e^{-S})$.

\item Terms in $T_4$ are products of transition amplitudes and are in general complex. Its average can be written in terms of~\eqref{amp} as 
\be \label{t4}
 \overline T_4 =\sum_{x_1 \neq x_2, y_1 \neq y_2}  C^{(4)}(x_1y_1, ij; \,ij, x_2y_1 ;\, x_2 y_2, ij ; \, ij, x_1 y_2 ; t)
 \ee
but its index structure is such that there is no permutation $\sigma\in \sS_4$ for which $m_i = m'_{\sigma(i)}$. In the model of~\eqref{b_model} these terms average to zero. Hence~\eqref{t4} is expected to be highly suppressed. 

\een

\begin{figure}[!h]
\centering
\includegraphics[width=10cm]{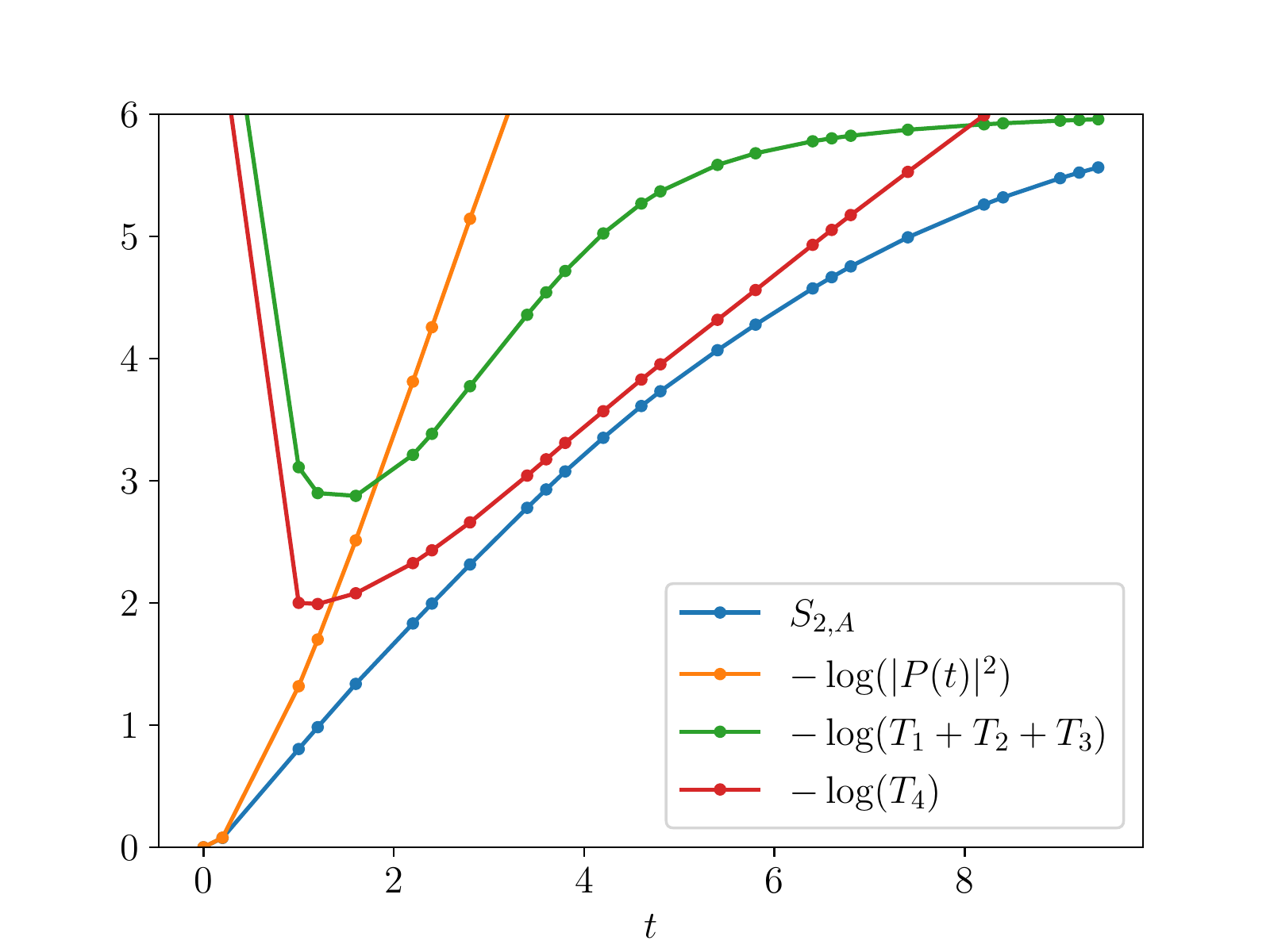}
\caption{Comparison of $S_2(t)$ with $-\log(|P(t)|^2)$, $-\log(T_1+T_2+T_3)$, and $-\log(T_4)$ for a single state $\ket{i}\ket{j}$ at $L=20$, with $i$ and $j$ both taken to be in the middle of the spectrum for the $L=10$ Hamiltonian ($i=512, j=512$). While this plot is for a single state, any other state with nearby energy 
(and hence the average over such states) shows similar behaviour of these different contributions.}
\label{fig:compare_20}
\end{figure} 

 From this analysis, at $t \sim O(L^0)$ we should have 
\be \label{s_simple}
e^{-S_{2, A}(t)} \approx |P(t)|^2 + \sO(e^{-S/2}) , 
\ee
which would imply that the evolution of $e^{-S_{2,A}(t)}$ is similar to that of $P(t)$, showing quadratic decay at early times and exponential decay subsequently. 
The exponential decay corresponds to a linear growth of $S_{2,A}$, which is qualitatively the expected behaviour for a local system~(see for instance~\cite{kimhuse_2013_entballistic, jonay_2018_coarse, zhou_nahum_2019_entRUC}).\footnote{See \cite{rakovszky_2019_diffusiverenyi, huang_2020_renyidiffusive} for  discussions of an expected transition to $\sqrt{t}$ growth at later times.}  The  saturation happens at times $t \sim O(L)$, when other terms can also contribute. 

Equation~\eqref{s_simple}, however, appears to give the wrong growth quantitatively. 
Numerically simulating the evolution of $S_{2,A}$ in the chaotic spin chain model \eqref{ham} with $g=-1.05$, $h=0.5$, we find that rate of growth of $S_{2,A}$ predicted by \eqref{s_simple} is much larger than the rate observed by directly evaluating $S_{2,A}$. $S_{2,A}$ and $-\log |P(t)|^2$ agree only at very early times. See Fig. \ref{fig:compare_20}.

Moreover, in Fig.~\ref{fig:compare_20}, we also compare $S_{2,A}(t)$ with $-\log(T_1+T_2+T_3)$, and $-\log(T_4)$. Remarkably, the dominant contribution at intermediate times seems to come from the term $T_4$, which we expected should average to zero from~\eqref{b_model}.  We verify this by analyzing the dependence of each contribution on the system size in Fig.~\ref{fig:scaling_log}. $T_1$ through $T_3$ decay exponentially with system size, consistent with the expectation from our simple model. On the other hand, $T_4$ is not only non-zero, but also does not decay with system size. Clearly, there are more correlations among $c^a_{ij}$ (which induced correlations among various transition amplitudes) than that captured by the simple model~\eqref{b_model}.

\begin{figure}[!h] 
\centering
\includegraphics[width=7cm]{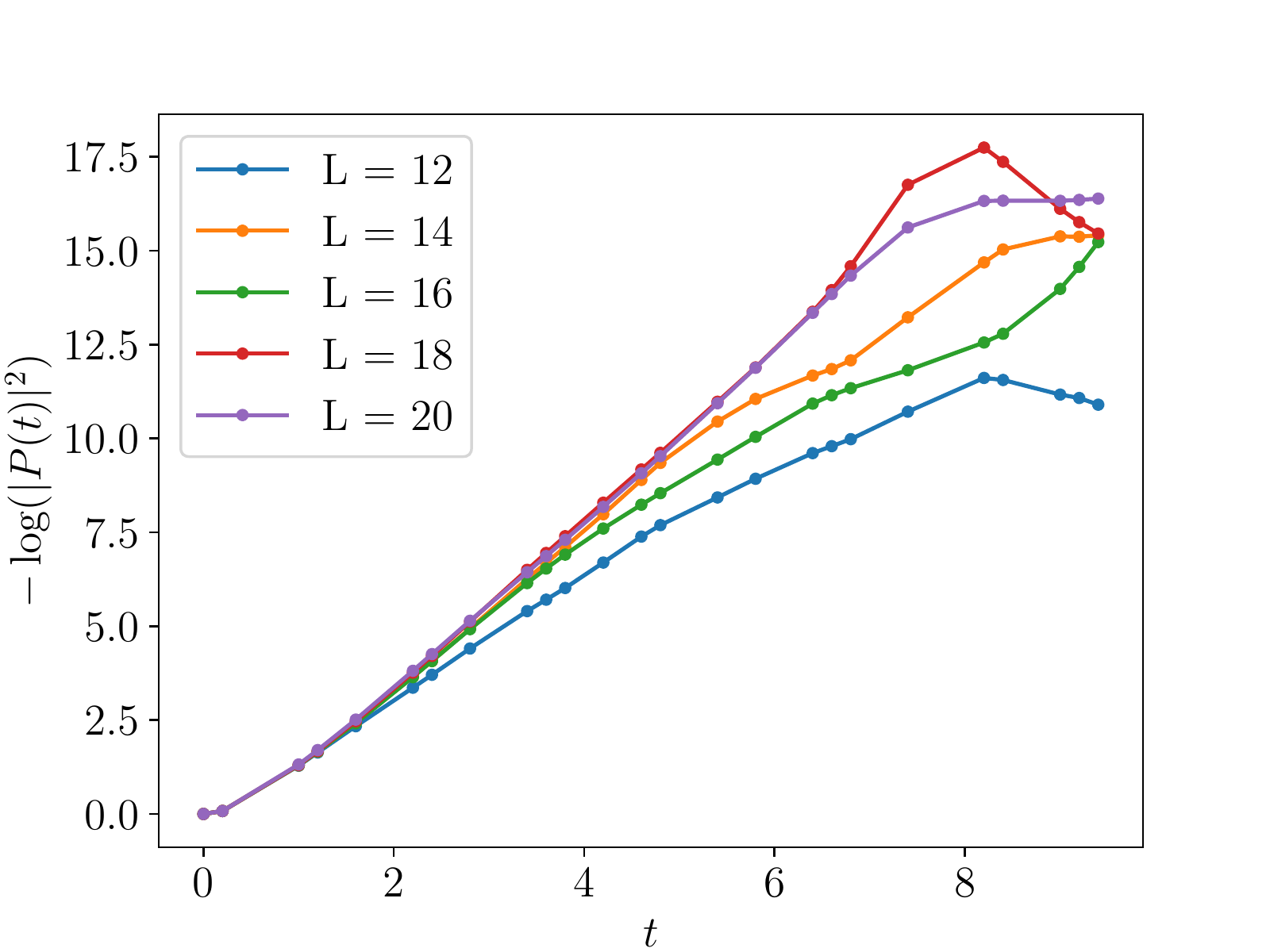} 
\includegraphics[width=7cm]{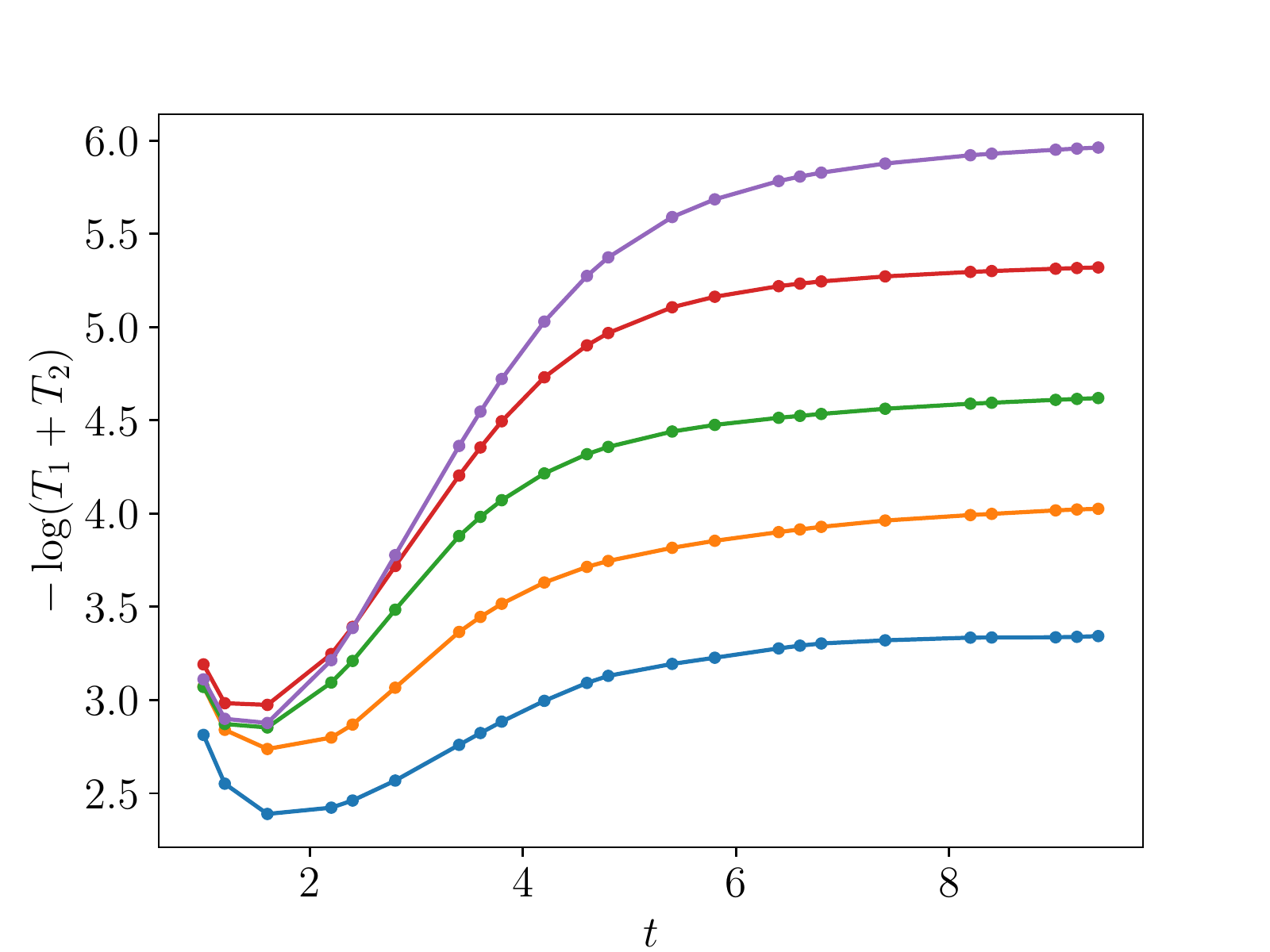} 
\includegraphics[width=7cm]{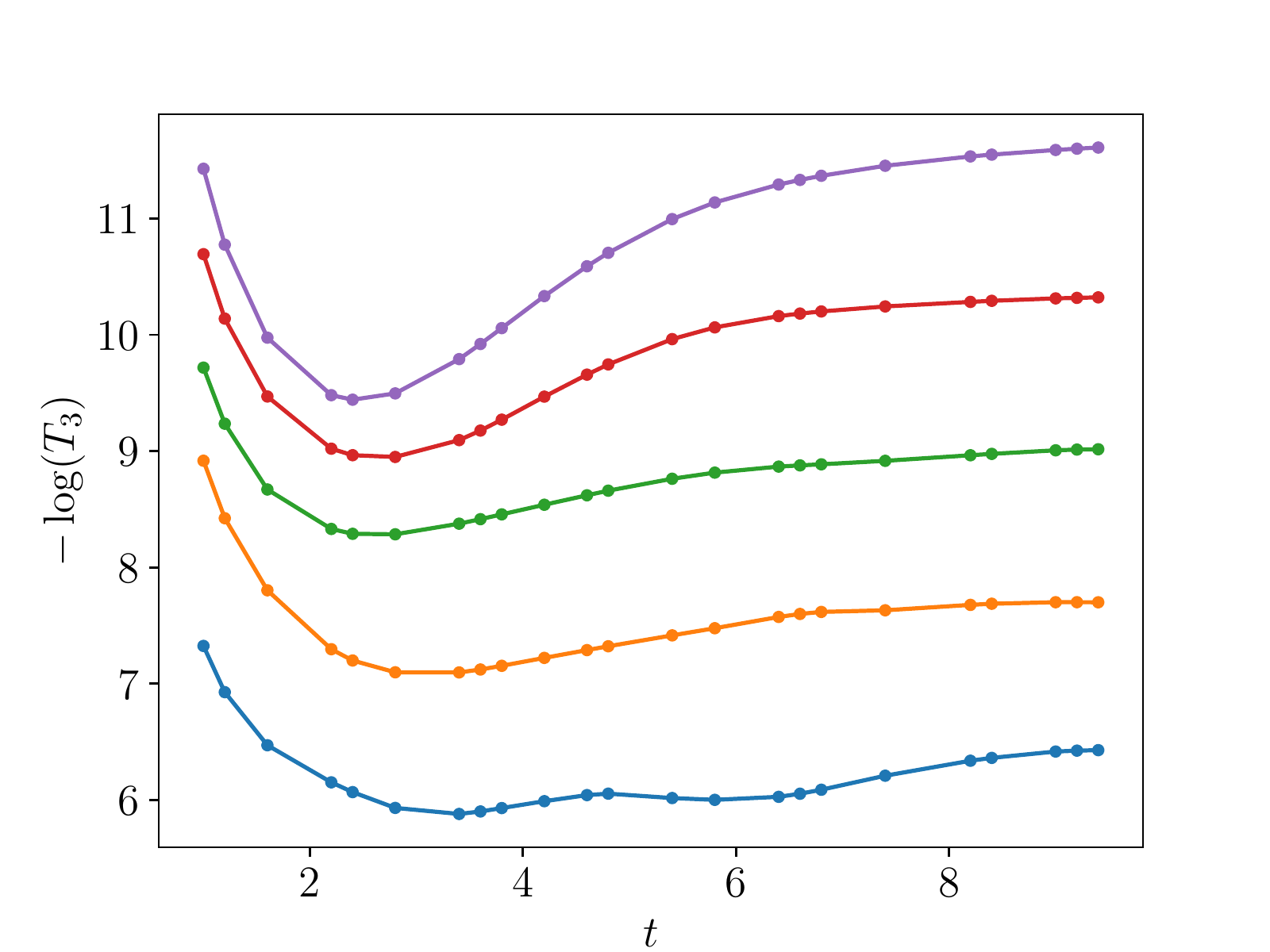} 
\includegraphics[width=7cm]{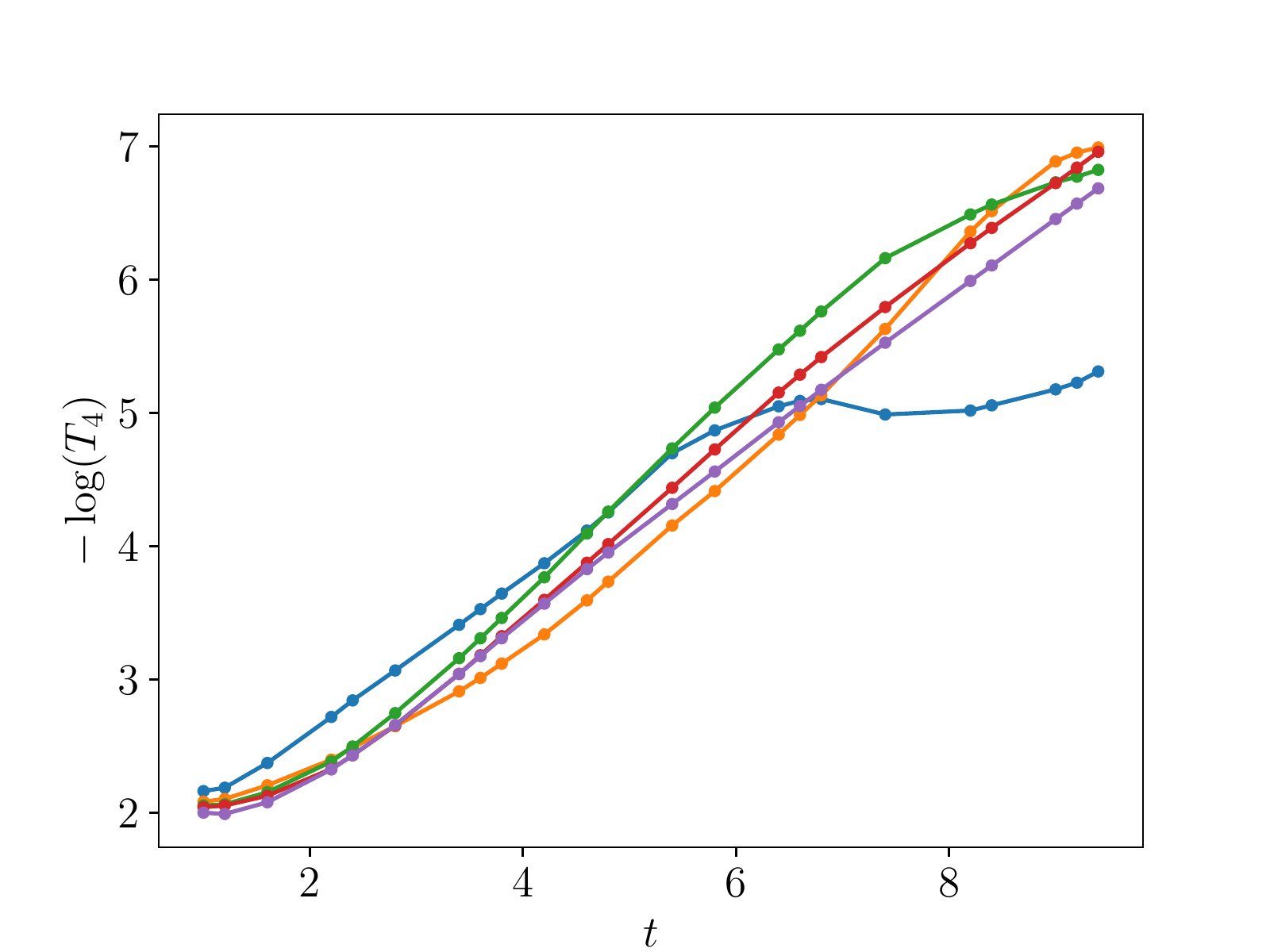} 
\caption{We consider states $\ket{i}\ket{j}$ with both $i$ and $j$ in the middle of the spectrum for different system sizes, and consider the dependence of each of the terms in \eqref{s2} on the system size from $L=12$ to $L=20$.}
\label{fig:scaling_log}
\end{figure}


That there are correlations among factors in~\eqref{75} may not be entirely surprising given the discussion of Sec.~\ref{sec:fco}, as the index structure of the factors in~\eqref{75} is similar to that in~\eqref{ehn}. 
The numerical results on $T_4$ suggest that we can parameterize 
\bega \label{c4_p}
C^{(4)}(x_1y_1, ij; \,ij, x_2y_1 ;\, x_2 y_2, ij ; \, ij, x_1 y_2 ; \omega) \quad (x_1 \neq x_2 \neq i,\, y_1 \neq y_2 \neq j) \cr
= e^{- 2 S(\bar E_{ij})} \hat g(\bar E_{x_1y_1} - \bar E_{ij}, \bar E_{x_2y_2} - \bar E_{ij}, \bar E_{x_1y_2}-\bar E_{ij}, \bar E_{x_2y_1}- \bar E_{ij}; \omega )
\end{gather} 
where $\hat g$ is an order $\sO(1)$ smooth function which is supported only for energy differences of order $\sO(1)$. 
Since $T_4$ is a sum of $e^{2S}$ such terms, it is $\sO(1)$. We provide numerical evidence for~\eqref{c4_p} in Fig.~\ref{fig:C4_corr}. 
Comparing the index structure in $T_4$ and that in~\eqref{ehn}, we may wonder whether the correlations in~\eqref{c4_p} may solely be understood from that in~\eqref{ehn}. In other words, whether~\eqref{c4_p} can be factorized into a product of $C^{(2)}$ of the form~\eqref{ehn}. This turns out not to be the case, as we find that 
\begin{equation}
    \begin{aligned}
       C^{(4)}(x_1y_1, ij; \,ij, x_2y_1 ;\, x_2 y_2, ij ; \, ij, x_1 y_2 ; t) 
        &\neq C^{(2)}(x_1y_1,ij;ij,x_2y_1; t) \; C^{(2)}(x_2y_2,ij;ij,x_1y_2; t) \,.
    \end{aligned}
\end{equation}
In particular, while the left hand side appears to be purely real, the right hand side is generally complex. 

Given that the average of $T_4$ provides the dominant contribution to $S_{2,A} (t)$, and that the  evolution of $S_{2,A} (t)$ is constrained from locality, we  expect that the correlations reflected by \eqref{c4_p} are a consequence of locality.

\begin{figure}[!h]
    \centering
        \begin{subfigure}{0.49\textwidth}
        \centering
        \includegraphics[width = \textwidth]{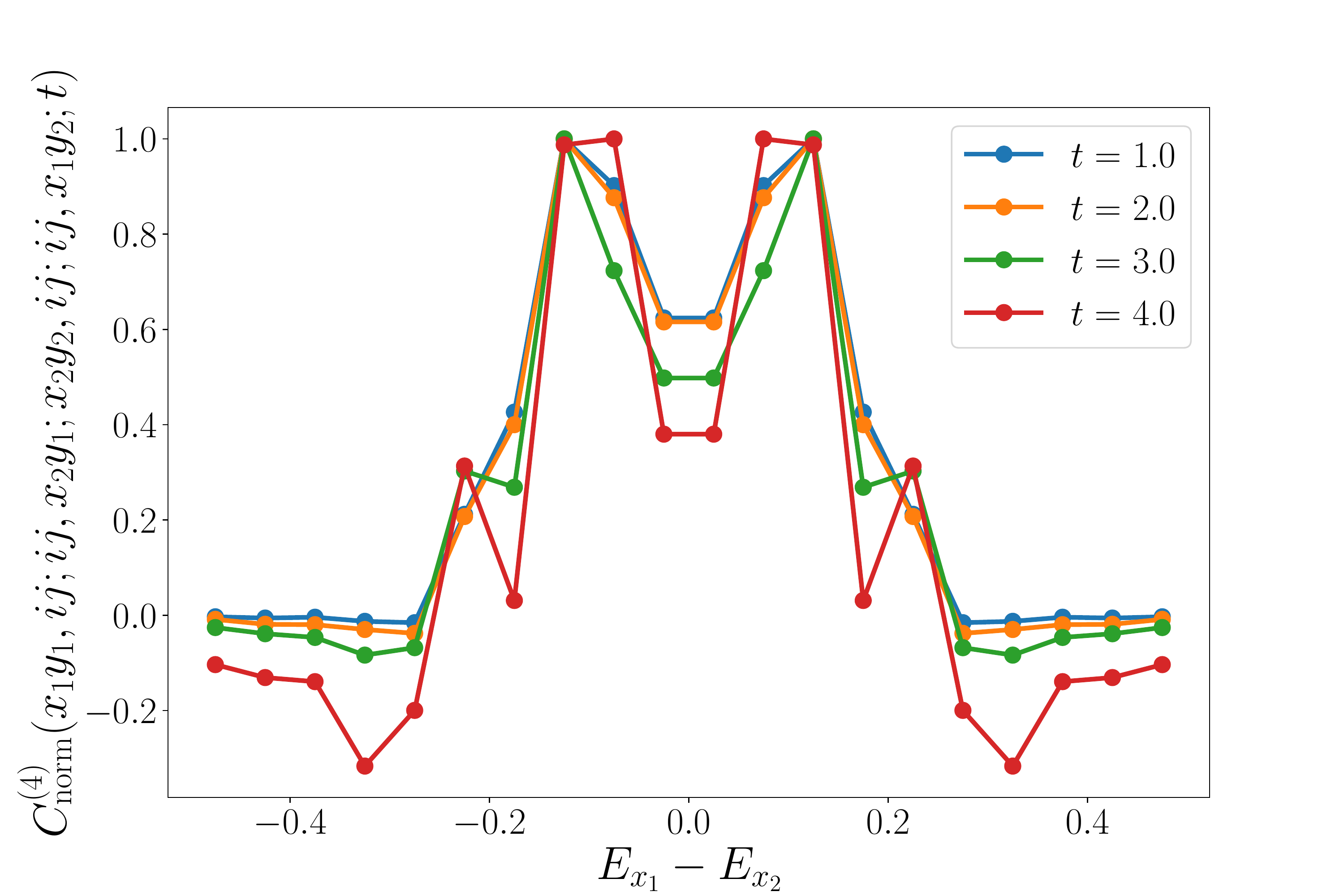}
        \caption{Late time saturation}
        \label{subfig:C4_Ex1x2_dependence}
    \end{subfigure}
    \begin{subfigure}{0.49\textwidth}
        \centering
        \includegraphics[width = \textwidth]{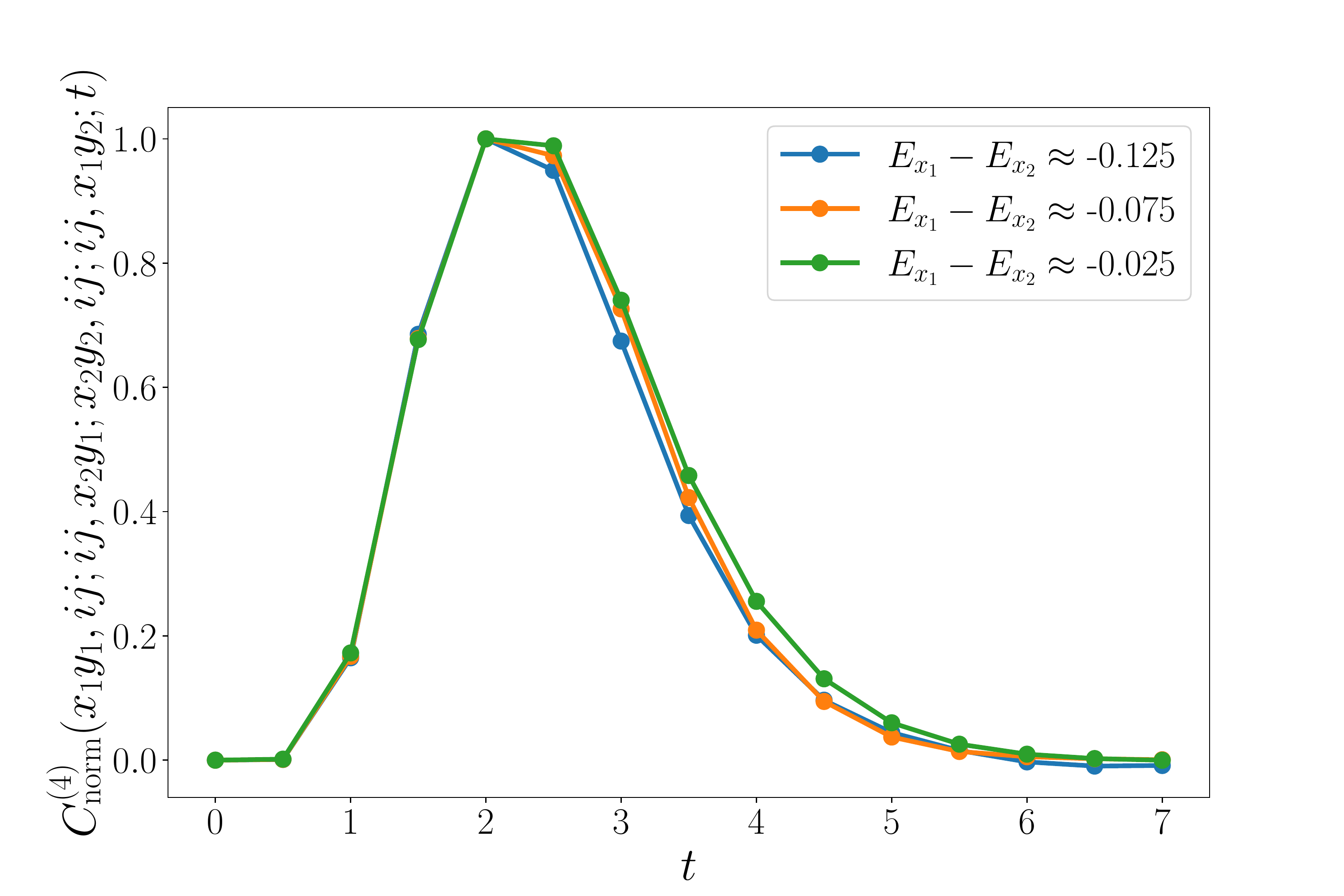}
        \caption{Time evolution}
        \label{subfig:C4_t_dependence}
    \end{subfigure}
    \caption{We numerically simulate $C^{(4)}(\ldots; t) = \int d \omega \,C^{(4)}(\ldots; \omega) e^{-i\omega t}$ for the index structure in \eqref{c4_p} in  the mixed-field Ising model at $L = 20$. It is computationally infeasible to study the dependence of $C^{(4)}$ on all energy indices. Therefore, we chose to fix $i,j$ near the middle of the spectrum, average $E_{y_1}, E_{y_2}, (E_{x_1} + E_{x_2})/2$ over microcanonical windows near the middle of the spectrum, and study the dependence of $C^{(4)}(\ldots; t)$ on $E_{x_1} - E_{x_2}$ and $t$. On the panel (a), we choose a few different values of $t$, and plot the dependence of $C^{(4)}(\ldots; t)$ on $E_{x_1} - E_{x_2}$, normalized by the maximum amplitude for each time. Up to fluctuations due to finite-size effects, the dependence is smooth and coherent across multiple choices of $t$. On panel (b), we plot $C^{(4)}(\ldots; t)$ for fixed $E_{x_1} - E_{x_2}$ and find a rapid decay with time, implying that its Fourier transform $C^{(4)}(\ldots; \omega)$ is a smooth function of $\omega$. Remarkably, although $C^{(4)}$ itself is a complex amplitude in general, the microcanonical average procedure outlined above renders the imaginary part negligible compared to the real part. 
    }
    \label{fig:C4_corr}
\end{figure}

It is also worth noting that while $S_{2,A}(t)$ in Fig. \ref{fig:compare_20} does not seem to have any clear linear regime, the growth of $-\log T_4$ does appear to be linear. This suggests that in the thermodynamic limit, where $T_4$ becomes closer to $S_{2,A}$ for intermediate times, the growth of $S_{2,A}(t)$ could become linear. We also note that for the system size and the choice of couplings we look at, we do not see a regime where $S_{2,A}(t)$ is proportional to $\sqrt{t}$, as was previously found in random product states~\cite{rakovszky_2019_diffusiverenyi}. 
It would be interesting to push the numerics to larger system sizes and check for the existence of a regime where $S_{2,A}(t)$ and $-\log T_4$ could both grow as $\sqrt{t}$.

\section{Evolution of correlation functions and interplay with ETH}
\label{sec:corr}

The series of dynamical quantities considered in the previous sections probe increasingly complex correlations amongst the $c^a_{ij}$ coefficients. In this section, we build on these results and explore the interplay between $c^a_{ij}$ and the coefficients appearing in the eigenstate thermalization hypothesis, as revealed by correlation functions of local operators. 
For the one-point function in the state $\ket{ij}$, correlations between $c^a_{ij}$ and the random numbers $R_{ab}$ appearing in the ETH for the full system eigenstates \eqref{eth_full} are necessary for explaining its initial value, while correlations amongst $c^a_{ij}$ imposed by unitarity (see Sec.~\ref{sec:phco}) play an important role in its time evolution. 
Furthermore, we find that 
the correlations among $c^a_{ij}$ which originate from locality (see Sec.~\ref{sec:ent}) do not affect two-point functions, but can play a role in four-point functions. In all correlation functions, we also derive constraints that relate the smooth function $h(\omega)$ in ETH \eqref{eth_full} to the eigenstate distribution $f(\omega)$. A more detailed physical understanding of these correlations is left for future work. 

\subsection{Evolution of one-point functions in $\ket{ij}$}

One commonly studied signature of thermalization 
is the evolution of expectation values of local observables in some out-of-equilibrium state to their thermal values. For the product of eigenstates, we consider the quantity 
\be 
\braket{ij| W(t) |ij} \label{oij}
\ee
for a local operator $W$. Let us say $W$ is in the $A$ half of the system. 

We will assume below that we can use the eigenstate thermalization hypothesis both in the full system and in each extensive subsystem. Let us state both versions of the ETH precisely. For two eigenstates $\ket{a}$ and $\ket{b}$ of the full system, we have 
\be \label{eth_full}
\braket{a|W|b} = W(E_a/V) \delta_{ab} + e^{-S(\bar E_{ab})/2} \,  R^{W}_{ab} \,  h_W(\omega_{ab}) 
\ee
where $\bar E_{ab} = (E_a+E_b)/2$,~ $\omega_{ab}= E_a - E_b$,~ $W(\epsilon)$ and $h_W(\omega)$ are $\sO(1)$ smooth functions of their arguments, and the  $R^{W}_{ab}$ are 
random variables with mean zero and variance 1, 
\be 
\label{r_var1}
\overline{R^{W}_{ab} {R^{W'}_{cd}}^{\ast}} = \delta_{ac} \delta_{bd} \delta_{WW'} \, . 
\ee
Similarly, for two eigenstates $\ket{i_1}$ and $\ket{i_2}$ of $H_A$, we have 
\be \label{eth_half}
\braket{i_1 |W |i_2} = W(E_{i_1}/V_A)\delta_{i_1i_2} + e^{-S_A(\bar E_{i_1i_2})/2} R^{A, W}_{i_1 i_2} h_{A, W}(\omega_{i_1 i_2}) 
\ee
where the notation is similar to that in \eqref{eth_full}, and in particular $W(\epsilon)$ is the same function of the energy density that appears in \eqref{eth_full}, and  $S_A$ is the thermodynamic entropy in $A$. 
Note the additional $A$ label in $R$ and $h$ to indicate the left half of the system. 

From \eqref{eth_half}, at $t=0$, \eqref{oij} is equal to the thermal expectation value in the microcanonical ensemble at energy density $E_i/V_A$: 
\begin{align} \label{eth_1}
\braket{ij| W |ij} &= \braket{i | W|i} = W(E_i/V_A) + e^{-S_A(E_i)/2} R^{A, W}_{ii} h_{A, W}(0) \, \nn
& = W(E_i/V_A) + \sO(e^{-S/4})
\end{align}
Let us now understand the interplay  of the properties of  $c^a_{ij}$ with the ETH in the evolution of $\braket{ij|W(t)|ij}$. First, instead of using the half system ETH as in \eqref{eth_1}, let us expand the expression at $t=0$ in terms of the full system eigenstates, and use the full system ETH as follows: 
\be \label{wab}
\braket{ij|W|ij} = \sum_{ab} c^a_{ij} {c^b_{ij}}^{\ast} W_{ab} = \sum_a |c^a_{ij}|^2 W(E_a/V) + \sum_{ab} c^a_{ij} {c^b_{ij}}^{\ast} e^{-S(\frac{E_a + E_b}{2})}R^W_{ab} h_W(\omega_{ab}) 
\ee

Comparing \eqref{eth_1} and \eqref{wab}, ignoring the exponentially suppressed contribution in \eqref{eth_1}, and  expressing $|c^a_{ij}|^2$ in terms of $f(\omega)$, we have 
\be 
 \sum_{ab} c^a_{ij} {c^b_{ij}}^{\ast} e^{-S(\frac{E_a + E_b}{2})}R^W_{ab} h_W(\omega_{ab})  = W\le(\frac{E_i}{V_A}\ri) - \int d\omega e^{\beta \omega/2} f(\omega) W\le(\frac{\bar E_{ij}+ \omega}{V}\ri)
\ee
Assuming $E_i - E_j = \sO(1)$, the RHS is $\sO(1/V)$. If we assume that $R^W_{ab}$, the matrix elements of the full system ETH, are uncorrelated with $c^a_{ij}$, then the LHS can be seen as a sum of $\sO(e^{2S})$ numbers of $\sO(e^{-2S})$ with random phases, so it is $\sO(e^{-S})$. We therefore see that we must have correlations between $c^a_{ij}$ and the matrix elements of the full system ETH such that the  LHS is instead $\sO(1/V)$. 

Let us now see whether the $c^a_{ij}$ also need to be correlated with the matrix elements $R^{A,W}_{i_1i_2}$ for the ETH of the $A$ subsystem. To see this, let us expand $W_{ab}$ further in \eqref{wab} in terms of subsystem eigenstates to write 
\be 
\braket{ij|W|ij} = \sum_{a,b, k, l, k'} c^a_{ij} {c^b_{ij}}^{\ast} c^a_{kl} c^b_{k'l} W_{kk'}
\ee
Then using ETH for the $A$ subsystem for $W_{kk'}$, we have  
\be \label{d6}
\braket{ij|W|ij} = \sum_{a,b,l,k} c^a_{ij} {c^b_{ij}}^{\ast} c^b_{kl} {c^a_{kl}}^{\ast}  W(E_k/V_A) + \sum_{a,b, l, k, k'} c^a_{ij} {c^b_{ij}}^{\ast} c^a_{kl} c^b_{k'l} e^{-S_A(\bar E_{kk'})/2} h_{A,W}(\omega_{kk'}) R^{A, W}_{kk'}
\ee
Let us expand the first term using the statistical model for $c^a_{ij}$ in  \eqref{c2_ru}, 
\begin{align} 
\sum_{a,b,l,k} c^a_{ij} {c^b_{ij}}^{\ast} c^b_{kl} {c^a_{kl}}^{\ast}  W(E_k/V_A) &= \sum_{a,b} |c^a_{ij}|^2 |c^b_{ij}|^2 W(E_i/V_A) - \sum_{a,b} W(E_i/V_A) k_{a, b, ij, ij}\nn
& + \sum_{a,  kl \neq ij}|c^a_{ij}|^2 |c^a_{kl}|^2 W(E_k/V_A) -\sum_{a, b, kl \neq ij} k_{a,b, ij,kl} W(E_k/V_A) \nn
&= W(E_i/V_A) + \sum_{kl \neq ij} W(E_k/V_A) (\sum_a |c^a_{ij}|^2 |c^a_{kl}|^2 - \sum_{a,b}k_{a,b,ij,kl}   ) \nn
& + \sO(e^{-S}) \nn
& = W(E_i/V_A) + \sO(e^{-S}) \label{523}
 \end{align}
If we assume that the coefficients $R^{A,W}_{kk'}$ for the ETH in the $A$ subsystem are uncorrelated with $c^a_{ij}$, the second term in \eqref{d6} gives a contribution of order $\sO(e^{-S/2})$. Since from \eqref{523}, the first term in \eqref{d6} is sufficient to explain \eqref{eth_1}, we find that there is no need for correlations between  $R^{A,W}_{kk'}$ and $c^a_{ij}$. 
 

Let us now understand the time-evolution of the expectation value in $\ket{ij}$ in terms of properties of the $c^a_{ij}$. We again use an expansion similar to \eqref{d6}, 
\begin{align} 
\braket{ij|W|ij} &= \sum_{a,b,l,k} e^{i(E_a-E_b)t}c^a_{ij} {c^b_{ij}}^{\ast} c^b_{kl} {c^a_{kl}}^{\ast}  W(E_k/V_A) \nn
& \quad \quad \quad \quad + \sum_{a,b, l, k, k'} e^{i(E_a-E_b)t} c^a_{ij} {c^b_{ij}}^{\ast} c^a_{kl} c^b_{k'l} e^{-S_A(\bar E_{kk'})/2} h_{A,W}(\omega_{kk'}) R^{A, W}_{kk'}\label{d61}
\end{align}
Based on the above discussion of the 
$t=0$ behaviour, let us ignore the second line of the above expression. Then using the model of \eqref{c2_ru} again, we have 
\begin{align}
\braket{ij| W(t)|ij} &=  W(E_i/V_A) \sum_{a,b} |c^a_{ij}|^2 |c^b_{ij}|^2 e^{i(E_a-E_b)t}+ \sum_{a, kl\neq ij} |c^a_{kl}|^2 |c^a_{ij}|^2 W(E_k/V_A)
\nn
& - \sum_{a,b, kl\neq ij} k_{a,b, ij,kl} e^{i(E_a-E_b)t} W(E_k/V_A) \nn
&= W(E_i/V_A) \, P(t) + \int d E_a d \bar E_{kl} e^{\frac{\beta}{2} (\bar E_{kl} - \bar E_{ij}) } f(E_a - \bar E_{kl}) f(E_a - \bar E_{ij}) W(E_k/V_A) \nn
&- \int dE_a d E_b d \bar E_{kl} e^{i \omega_{ab} t} e^{-\frac{3}{4} \beta \bar \omega + \ha\beta \omega'} g(\omega_{ab}, \bar \omega, \omega') W(E_k/V_A) 
\label{524}
\end{align} 
where $\beta=S'(E)|_{E= \bar E_{ij,kl}/2}$,~ $\omega_{ab}= E_a-E_b$,~ $\bar \omega = \bar E_{ij}-\bar E_{kl}$,~ $\omega'=\frac{E_{ab}-\bar E_{ij, kl}}{2}$,~ and $P(t)$ is the return probability for the state $\ket{ij}$. 

All three terms in \eqref{524} are $\sO(1)$. At $t=0$, the second and third terms combine to zero, and we have $\braket{ij| W(t=0)|ij}= W(E_i/V_A)$. The equilibrium value is given by the second term. The non-trivial time-evolution of the operator expectation value  is thus determined by both $f$ and $g$. Without including the correlations from unitarity, we would only get the first two terms in \eqref{524}.

The evolution of one-point functions in out-of-equilibrium initial states during thermalization was previously studied using spectral properties in \cite{reimann_fast}. The initial states considered there have a simple representation in some fixed reference basis, in which the representation of the operator $W$ is also simple. It is assumed that the energy eigenbasis is related to this reference basis by a random unitary. The expression obtained for  the time-evolution of the operator expectation value there is somewhat simpler than \eqref{524} due to the lack of energy constraints in the relation between the reference basis and the energy eigenbasis for that setup.

\subsection{Evolution of thermal correlation functions of local operators}

Let us consider   the relation between the properties of the coefficients $c^a_{ij}$ and  the behaviour of thermal correlation functions of local operators. 

We will consider the two-point function of an operator $W_x$ at site $x$ with itself, 
\be \label{2p}
\braket{W_x(t) W_x} = \frac{1}{Z_{\beta}} \text{Tr}[e^{-\beta H}W_x(t) W_x]  
\ee
and 
the out-of-time-ordered correlation function (OTOC) between $W_x$ and  $Z_y$, 
\be \label{4p}
\braket{W_x(t) Z_y W_x(t) Z_y} = \frac{1}{Z_{\beta}} \text{Tr}[e^{-\beta H} W_x(t) Z_y W_x(t) Z_y] \, . 
\ee
In a local chaotic system, both correlation functions defined above should be $\sO(1)$ at $\sO(1)$ times, and should eventually decay to a small value. To simplify the discussion, we will consider both correlators at infinite temperature $\beta=0$ below.~\footnote{
 In strongly-correlated systems without other energy scales, the two-point function decays on the time scale $\beta$. But for more general local chaotic systems, we expect that the two-point function should still be $\sO(1)$ at $\sO(1)$ times in the $\beta \rightarrow 0$ limit. The decay timescale is controlled by UV couplings rather than $\beta$.
 }  For this case $Z_{\beta}= d$, with $d$  the Hilbert space dimension. For simplicity we will assume that $\text{Tr}[W_x]=\text{Tr}[Z_y]=0$, and $\text{Tr}[W_x^2]=\text{Tr}[Z_y^2]=d$. We will find below that \eqref{2p} is not sensitive to correlations among $c^a_{ij}$, but \eqref{4p} may be affected by them at leading order.

In all equations below, there is no sum on repeated indices unless explicitly stated. 

Let us first consider the evolution of the two-point function \eqref{2p}. Expanding in the energy eigenbasis of the full system, we have 
\begin{align} \label{2p_e}
\braket{W(t) W} &= \frac{1}{d} \sum_{a,b} e^{i(E_a-E_b)t} W_{ab} W_{ba}\, . 
\end{align}
Let us use the ETH for the full system from \eqref{eth_full}, and assume for simplicity that the microcanonical expectation value $W(E)=0$ for any $E$ in the middle of the spectrum. Then 
using \eqref{eth_full} in \eqref{2p_e}, we find 
\be 
\braket{W(t)W} = \int d\omega\, e^{i \omega t} \, h_W(\omega)^2 
\ee
As expected, this quantity is $\sO(1)$ at $\sO(1)$ times. 

We can also relate the two-point functions to the $c^a_{ij}$. For an operator $W$ in the $A$ half of the system, 
\be 
W_{ab} W_{ba} = \sum_{\substack{i_1, i_2, i_3, i_4\\j_1, j_2 }} {c^a_{i_1j_1}}^{\ast} c^a_{i_4j_3} {c^b_{i_3j_3}}^{\ast} c^b_{i_2j_1} W_{i_1i_2} W_{i_3i_4} \,.
\ee
Using ETH for the matrix elements $W_{i_1 i_2}$ in the subsystem eigenstates, we have  
\be  \label{half_eth}
W_{i_1 i_2} W_{i_3 i_4} = e^{-S_{A}(\bar E_{i_1i_2})} h_{A, W}(\omega_{i_1 i_2})^2\delta_{i_1 i_4}\delta_{i_2i_3}\, . 
\ee
Assuming, like in the previous subsection, that
the coefficients $c^a_{ij}$ are uncorrelated with these matrix elements for the subsystem $A$, we have  
\begin{align}  \label{ab_exp}
W_{ab} W_{ba} =   \sum_{i_1, i_2, j_1, j_3} e^{-S_A(\bar E_{i_1i_2})} h_{A,W}(\omega_{i_1i_2})^2 {c^a_{i_1j_1}}^{\ast} c^a_{i_1j_3} {c^b_{i_2j_3}}^{\ast} c^b_{i_2j_1} \,.
\end{align}
From the model \eqref{c2_ru} for the coefficients, the only non-zero contributions in this sum are from cases where (i) $j_1=j_3$ or (ii) $i_1=i_2$. The contribution from the second case is suppressed by $\sO(e^{-S})$ compared to the first, so we have 
\begin{align}  \label{ab_exp2}
\overline{W_{ab} W_{ba}} & \approx  \sum_{i_1, i_2, j} e^{-S_A(\bar E_{i_1i_2})} h_{A,W}(\omega_{i_1i_2})^2\overline{|c^a_{i_1j}|^2} \, \overline{|c^{b}_{i_2j}|^2}  \,.
\end{align}
Now each $|c^a_{ij}|^2\approx e^{-S((E_a+\bar E_{ij})/2)}f(E_a -\bar E_{ij})$, so  that \eqref{ab_exp2} has the right scaling of $e^{-S}$. Note that even if there were additional correlations beyond \eqref{c2_ru} such that a general term in \eqref{ab_exp} were comparable to the terms with $i_1=i_2$, i.e. 
\be 
\overline{{c^a_{i_1j_1}}^{\ast} c^a_{i_1j_3} {c^b_{i_2j_3}}^{\ast} c^b_{i_2j_1}} \sim \sO(e^{-3S}) \, , 
\ee
the total contribution from such terms is still suppressed by $\sO(e^{-S/2})$ compared to \eqref{ab_exp2}. Hence, the two-point function is not sensitive to correlations among the $c^a_{ij}$. By comparing \eqref{ab_exp2} and the average of this quantity that we get using the full system ETH \eqref{eth_full}, 
\be 
\overline{W_{ab} W_{ba}} = e^{-S(\bar E_{ab})} h_W(\omega_{ab})^2 \, , 
\ee
we can relate $h_W$ to $h_{A,W}$ and $f$.

Let us now consider the evolution of the OTOC \eqref{4p}. Expanding in the energy eigenbasis of the full system,  
\begin{align} 
\label{fpt3}
\braket{W(t) Z W(t) Z} &= \frac{1}{d} \sum_{a,b, c, d} e^{i(E_a -E_b + E_c -E_d)t} \, W_{ab} Z_{bc} W_{cd} Z_{da}   
\end{align} 
If we assume that the $R^W_{ab}$, $R^Z_{ab}$ in \eqref{eth_full} are uncorrelated Gaussian random variables, this quantity appears to be of order $\sO(e^{-2S})$ for any $\sO(1)$ time. 
We can see that it is $\sO(1)$ at $\sO(1)$ times from the following generalized version of the ETH ansatz, 
proposed in \cite{foini2019_ETH_OTOC}~\footnote{\cite{foini2019_ETH_OTOC} considers the OTOC of an operator $W$ with itself, so the ansatz is expressed in terms of matrix elements only of a single operator. We extend this to the case of two distinct operators, along the lines of Appendix A of \cite{chan_eig}.}: 
\be \label{new4}
W_{ab} Z_{bc} W_{cd} Z_{da} = e^{-3S(\bar E)} \, h^{(4)}_{W,Z}(\omega_{ab}, \omega_{bc}, \omega_{cd})
\ee
where $\bar E=(E_a + E_b + E_c + E_d)/4$, and $h^{(4)}$ is some smooth $\sO(1)$ function.   


Let us now understand the contributions to \eqref{fpt3} from the $c^a_{ij}$.  Let the operator  operator $W$ be in the left half of the system, and $Z$ in the right half.  
\be 
W_{ab} Z_{bc} W_{cd} Z_{da}= \sum_{\substack{i_1, ..., i_6 \\j_1, ..., j_6}} c^a_{i_6j_6} {c^a_{i_1 j_1}}^{\ast} \, c^b_{i_2j_1} {c^b_{i_3j_2}}^{\ast} c^c_{i_3j_3} {c^c_{i_4j_4}}^{\ast} c^d_{i_5j_4} {c^d_{i_6j_5}}^{\ast} W_{i_1 i_2} Z_{j_2j_3} W_{i_4 i_5} Z_{j_5j_6} \,.
\ee
Assume again that the $c^a_{ij}$ are uncorrelated with $W_{i_1 i_2}$, $Z_{j_1 j_2}$.  Further, since $W$ and $Z$ are in different parts of the system, let us assume that their matrix elements in the subsystem eigenstates are uncorrelated.~\footnote{This assumption is reasonable despite the generalized ETH if we assume we do not have exact translation-invariance, or take the sizes of $A$ and $B$ to be different.} 
Then using \eqref{half_eth} and its analog for $Z$, 
\begin{align}  
\overline{W_{ab} Z_{bc} W_{cd} Z_{da}}= \sum_{\substack{i_1, i_2, i_3, i_6 \\j_1, j_2, j_3, j_4 }} & \overline{c^a_{i_6j_2} {c^a_{i_1 j_1}}^{\ast} \, c^b_{i_2j_1} {c^b_{i_3j_2}}^{\ast} c^c_{i_3j_3} {c^c_{i_2j_4}}^{\ast} c^d_{i_1j_4} {c^d_{i_6j_3}}^{\ast}} \nonumber \\
&\times e^{-S_{A}(\bar E_{i_1i_2})-S_B(\bar E_{j_2j_3})} h_{A, W}(\omega_{i_1 i_2})^2h_{B, Z}(\omega_{j_2 j_3})^2\,. \label{abcd}
\end{align}
Now if we assume that the $c^a_{ij}$ are uncorrelated random variables, then the only contribution to \eqref{abcd} for generic values of $a,b,c,d$ is 
\begin{align}  
\overline{W_{ab} Z_{bc} W_{cd} Z_{da}} ~~~~~~\myeq ~~~~~~ \sum_{\substack{i_1, i_2\\j_1, j_3 }} & \overline{|c^a_{i_1 j_1}|^2} \,  \overline{|c^b_{i_2j_1}|^2} \,  \overline{|c^c_{i_2j_3}|^2} \, \overline{|c^d_{i_1j_3}|^2} \nonumber \\
&\times e^{-S_{A}(\bar E_{i_1i_2})/2)-S_B(\bar E_{j_2 j_3})} h_{A, W}(\omega_{i_1 i_2})^2h_{B, Z}(\omega_{j_1 j_3})^2 \,.\label{abcd_2}
\end{align}
This expression scales as $\sO(e^{-3S})$,  so the uncorrelated approximation for $c^a_{ij}$ is in principle sufficient to account for \eqref{new4}. It is possible that in addition to this contribution, correlations among the $c^a_{ij}$ can also contribute at leading order to the four-point function. For example, note that the pattern of indices appearing in \eqref{abcd} is similar to the pattern appearing in the expression of the time-evolved second Renyi-entropy, \eqref{s2full_1}. We can write 
\begin{align} 
\overline{\braket{W(t) Z W(t) Z}} = \frac{1}{d}   \sum_{\substack{i_1, i_2, i_3, i_6 \\j_1, j_2, j_3, j_4 }} &  \,   C^{(4)}(i_6j_2, \, i_1 j_1; \, i_2j_1, \, i_3j_2; \, i_3j_3, \, i_2j_4; \, i_1j_4, \, i_6j_3; t) \nonumber \\
&\times e^{-S_{A}(\bar E_{i_1i_2})-S_B(\bar E_{j_2j_3})} h_{A, W}(\omega_{i_1 i_2})^2h_{B, Z}(\omega_{j_2 j_3})^2 \,.\label{abcd_3}
\end{align}
The argument of $C^{(4)}$ in \eqref{c4_p} is a special case of that in \eqref{abcd_3}. So if  the more general $C^{(4)}$ appearing in \eqref{abcd_3} also has the scaling of the RHS of \eqref{c4_p}, then we can get a contribution similar to $T_4$ in the four-point function, which competes with \eqref{abcd_2}. We leave a detailed numerical study of these contributions to future work.

One interesting question is how \eqref{abcd_3} can give rise to the dependence of the OTOC on the distance between the operators $A$ and $B$. The contribution $C^{(4)}$ does not contain this information, so it may be encoded in the functions $h_{A,W}$ and $h_{B,Z}$, which may have some dependence on the distance between the operator and the boundary of the subsystem. 

\section{Conclusions and discussion}
\label{sec:disc}

In the last several decades, a lot of progress has been made in understanding universal properties of the energy eigenvalues and eigenstates of chaotic quantum many-body systems. The powerful framework of random matrix theory has driven much of this progress.  However, random matrix theory ignores the local structure of interactions in realistic Hamiltonians, and fails to capture universal dynamical properties resulting from locality, such as ballistic spreading of operators and linear growth of entanglement entropy. An important open question is about how these dynamical phenomena emerge from spectral properties, which in principle underlie all features of the dynamics. 

In this work, we addressed this question in a family of chaotic spin chain models in one spatial dimension. We studied the time-evolution of a few different  quantities during the thermalization  of a product of eigenstates of two extensive subsystems, including the return probability, transition probability, Renyi entropy, and correlation functions. We attributed the dynamics of these quantities to various  collective properties of the coefficients relating the energy eigenbasis of the full Hamiltonian to products of eigenstates of the subsystem Hamiltonians. The magnitudes of the coefficients were found to take a simple universal form, given by a Lorentzian at small energy differences that crosses over to an exponential decay for large energy differences. The Lorentzian form of the magnitude explained the exponential decay of the return probability with time.  Correlations among these coefficients, which are often assumed to be negligible, were found to play a key role in the evolution of the transition probability and the entanglement entropy during thermalization.



One important direction for future work is to see whether the spectral properties we identified in this paper also hold in other examples of local chaotic quantum many-body systems. For example, in the SYK chain model \cite{syk_chain}, it may be possible to analytically study various properties of the coefficients $c^a_{ij}$. Such studies could also provide a better understanding of  the somewhat mysterious correlations from locality discussed in Sec.~\ref{sec:ent}, which govern the evolution of the second Renyi entropy. 

The coefficients $c^a_{ij}$ can also be studied in the context of quantum field theories, with the caveat that their definition would depend on the UV cutoff. For example, one could consider the overlap between the eigenstates of two semi-infinite  boundary conformal field theories (BCFTs) and the eigenstates of a single CFT where they are joined together. The evolution of the von Neumann entropy on joining together two pure states in BCFTs  has previously been studied in (1+1)D in \cite{joining_1, joining_2}.   To address the questions posed in this paper, one would need to extend these calculations to the case where the two pure states are excited eigenstates, and to quantities like the return probability and transition probability. The holographic dual of BCFTs is well-understood \cite{karch_randall, tak_1, tak_2}, and it would be interesting to see what features of the bulk theory  underlie the behaviour of $P(t)$, $P_{ij,xy}(t)$, and $c^a_{ij}$. 

The dependence of the various quantities discussed above on the UV cutoff is due to the fact that the full  Hilbert space $\sH_{AB}$ in a continuum QFT does not factorize into $\sH_A\otimes \sH_B$. One interesting conceptual question is whether it is possible to capture some aspects of the collective behaviour of $c^a_{ij}$ with quantities that are well-defined in the continuum  using algebraic QFT \cite{takesaki2002, witten2018}.

Another important  direction is the generalization of our results to higher dimensions. 
For the magnitude of the coefficients, our arguments for a Lorentzian form of the eigenstate distribution function in Sec.~\ref{subsec:derivation_f_form} seem to be equally applicable to one or higher spatial dimensions.  
On the other hand, the argument of~\cite{murthy_2019_structure} suggests a Gaussian form. We discuss some limitations of the argument for a Gaussian form in Appendix \ref{app:MSlimitation}, but since both arguments involve approximations it would be useful to understand the behaviour explicitly in a concrete model in higher dimensions.
Similarly, understanding the nature of the phase correlations and their contribution to linear growth of entanglement entropy in higher dimensions is an important question.  

Another exciting direction is to develop protocols to test the predictions of this paper in experimental setups in the near term. As mentioned in the introduction, $P(t)$ is an example of a Loschmidt echo, which has been indirectly probed for other sets of initial states in NMR experiments  \cite{echo}. Recently, algorithms have also been developed for measuring the return probability for simple initial states using interferometry and other techniques that can be practically implemented in systems like trapped ion simulators and Rydberg atoms in optical lattices  
(see Appendix A of \cite{cirac} and references therein). 
While it is not easy to prepare products of  excited energy eigenstates as initial states in experimental setups, some algorithms have recently been developed to apply a ``filtering" operation to a product state in order to prepare a pure state with an arbitrarily small energy variance \cite{ent_var, cirac}. To probe the properties of the eigenstate distribution function, it may be sufficient to prepare states with some small $\sO(1)$ variance $\delta^2$.  For instance, consider the quantity 
\be  Q_{\delta, ij} = \bra{i_{\delta}}_A\bra{j_{\delta}}_B e^{-i H t} \ket{i_{\delta}}_A\ket{j_{\delta}}_B , \quad \quad  \ket{i_{\delta}} = \sum_{E_{i'} \in [E_i -\delta, E_i + \delta]} d_{i'} \ket{i'} \,   
\ee
where $d_{i}$ are random coefficients, which can be obtained for instance from taking random initial product states in the filtering protocol of 
\cite{ent_var, cirac}. Then by averaging over different realizations of the $d_i$ we find that 
\be 
\overline{Q_{\delta, ij}} \approx \braket{ij| e^{-iHt}|ij} \sim e^{-\Gamma t} \, . 
\ee
Making these ideas more precise and implementing them in realistic systems  should provide an interesting challenge for future work. 


Finally, one can try to understand the dynamics of more general initial states in terms of collective properties of the coefficients $c_{\psi}^a$ defined in \eqref{cp}. As discussed in the introduction, for product states the return probability has the Gaussian form in \eqref{prod_return}. On expressing $S_{2,A}(t)$ for a product state in terms of $c^a_{\psi}$, if we assume that the $c^a_{\psi}$ are uncorrelated random variables, we find that $S_{2,A}(t) \approx -\log |\braket{\psi| e^{-i H t}|\psi}|^2$, precisely as in Sec.~\ref{sec:ent}. This would give the unphysical prediction that the second Renyi entropy for such states grows quadratically at a rate proportional to the system size, unlike the linear and slower than linear growth at a finite rate observed in \cite{rakovszky_2019_diffusiverenyi}. Hence, there must be significant correlations  among the $c^a_{\psi}$ giving rise to the bounded growth of entanglement, which should be characterized more carefully in future work. 




\acknowledgments

We would like to thank J. Ignacio Cirac, Anatoly Dymarsky, Matthew Hastings, Veronika Hubeny, Izabella Lovaz, Daniel K. Mark, Chaitanya Murthy, Daniel Ranard, Mukund Rangamani, Douglas Stanford, and Tianci Zhou for helpful discussions. We would also like to thank Chaitanya Murthy for comments on the draft. The authors acknowledge the MIT SuperCloud and Lincoln Laboratory Supercomputing Center for providing HPC resources that have contributed to the research results reported within this paper.
Z.D.S. is supported by the Jerome I. Friedman Fellowship Fund, as well as in part by the Department of Energy under grant DE-SC0008739. S.V. is supported by Google. 
HL is supported by the Office of High Energy Physics of U.S. Department of Energy under grant Contract Number DE-SC0012567 and DE-SC0020360 (MIT contract \# 578218).

\begin{appendix} 

\section{Subtleties of the Murthy-Srednicki argument}\label{app:MSlimitation}

In this appendix, we review the  derivation of the form of the eigenstate distribution function by Murthy and Sredniki in \cite{murthy_2019_structure}, identify the implicit assumptions involved, and clarify its range of applicability. 

Murthy and Sredniki consider a similar division of a local  Hamiltonian $H$ into $H_A$, $H_{B}$ and $H_{AB}$ as we have in the main text, but focus on the case of (2+1) and higher dimensions. In such cases, the area $\sA$  of the boundary between $A$ and $B$ (measured in units of lattice spacing) is also a large quantity. Various errors in their approximations are suppressed in powers of $1/\sqrt{\sA}$. 
$H_{AB}$ has the form 
\be 
H_{AB} = \sum_{x\in \partial A} h_{x} \, 
\ee
where the $h_x$ are local operators supported near the boundary $\partial A$ between $A$ and $B$.

To derive the EDF for a given state $\ket{a}$, they shift $H_{AB}$ by a constant such that $\braket{a|H_{AB}|a}=0$. 
They assume the following ansatz for the coefficients $c^a_{ij}$, 
\begin{equation}
    c^a_{ij} = F(E_{ij} - E_a) e^{-S(E_{ij})/2} M^a_{ij} \,,
\end{equation}
where $M^a_{ij}$ is a matrix of erratically varying $\sO(1)$ numbers, $S(E)$ is the microcanonical entropy, and $F$ is taken to be a smooth function. Based on the eigenstate thermalization hypothesis, they further assume that 
 the eigenstate $\ket{a}$ has an $\sO(1)$ correlation length $\xi$. 
To study the functional form of $F$, consider its moments:
\begin{equation}\label{eq:MS_F_expand}
    \begin{aligned}
    \int d \omega F(\omega) \omega^n &= \bra{a} (H_A + H_B - E_a)^n\ket{a} = \bra{a} (H - E_a - H_{AB})^n \ket{a} \\
    &= \bra{a} H_{AB}^n \ket{a} - \bra{a}H_{AB} (H-E_a) H_{AB}^{n-2}\ket{a} + \ldots \,.
    \end{aligned}
\end{equation}
 The concrete statement proved in \cite{murthy_2019_structure} is the following:
\begin{claim}
    For any integer $n \ll \sqrt{\mathcal{A}}$, we have
    \begin{equation} \label{msf}
        \int d\omega F(\omega) \omega^n \approx \begin{cases}
            (n-1)!! \Delta^{n} & n \text{ even} \\ 0 & n \text{ odd}
        \end{cases} \,\, + \mathcal{O}\left(\mathcal{A}^{(n-1)/2}\right) \, .  
    \end{equation}
    Here $\Delta$ is the standard deviation of $H_{AB}$ in $\ket{a}$, 
\be 
    \Delta^2 \equiv \bra{a} H_{AB}^2 \ket{a} =  \mathcal{O}(\mathcal{A}) \, , 
\ee
and hence the errors in the even moments are suppressed by powers of $1/\sqrt{\mathcal{A}}$ relative to the leading contribution.  
\end{claim}
\begin{proof}
    The $n=1$ and $n=2$ cases are true by definition. For higher moments, we have to contend with the complicated sum in \eqref{eq:MS_F_expand}. The first term in the sum can be expanded as
    \begin{equation}
        \bra{a}H_{AB}^n\ket{a} = \sum_{x_1,\ldots, x_n \in B} \bra{a} h_{x_1} \ldots h_{x_n} \ket{a}  \,.
    \end{equation}
    Since the correlation length is $\xi$, the higher-point function $\bra{a} h_{x_1}\ldots h_{x_n}\ket{a}$ vanishes whenever $x_i$ is separated from all the other $x_j$'s by a distance much larger than $\xi$, since we have for instance 
    \begin{align} 
    \sum_{|x_1-x_j|\gg \epsilon}\braket{a| h_{x_1} h_{x_2} ... h_{x_n} |a} &\approx \sum_{|x_1-x_j|\gg \epsilon}\braket{a| h_{x_1}|a} \braket{h_{x_2} ... h_{x_n} |a} \nonumber \\
    &\approx \braket{a|H_{AB}|a}\sum_{x_2, ..., x_n} \braket{a|h_{x_2} ... h_{x_n} |a} \nonumber  \\
    &=0 \, . 
    \end{align}
     When $n \ll \mathcal{A}$, the entropically favorable non-zero configurations are those in which the $\{x_i\}'s$ organize into well-separated pairs.  The total contribution from configurations where more than two of the indices $x_i$ are spatially proximate is suppressed by powers of $1/\mathcal{A}$. When $n$ is odd, 
     in order to get a non-zero contribution, we must have one set of three nearby $h_{x_i}$'s in each configuration, and can form $(n-3)/2$ pairs out of the rest.  Hence the correlation function is approximately $\sO(\sA^{\frac{n-1}{2}})$. On the other hand, when $n$ is even, all the indices can be paired and we obtain a sum over all Wick contractions, such that the result is $\sO(\sA^{\frac{n}{2}})$. This combinatorial structure immediately gives approximately Gaussian moments
    \begin{equation} \label{h12_moments}
        \bra{a}H_{AB}^n\ket{a} = \begin{cases}
            (n-1)!! \Delta^n & n \text{ even}\\ 0 & n \text{ odd}
        \end{cases} + \mathcal{O}\left(\mathcal{A}^{(n-1)/2}\right)
    \end{equation}
    so long as $n \ll \mathcal{A}$. For $n \approx A$, we cannot have well-separated pairs of $h_{x_i}$, so \eqref{h12_moments} no longer holds. 
    
    Next, we show that the remaining terms in \eqref{eq:MS_F_expand} are suppressed relative to the first term when $n \ll \sqrt{A}$. Generally, each remaining term contains factors of $(H-E_a)$ sandwiched between powers of $H_{AB}$. Since $H$ is local, $H_{AB}^k \ket{a}$ is a superposition of eigenstates with weight concentrated below $E_a + \mathcal{O}(k ||h_x||)$. Therefore, when $k < n \ll \sqrt{\mathcal{A}}$, each factor of $H-E_a$ contributes at most $\mathcal{O}(k ||h_x||)$. From \eqref{h12_moments}, each factor of $H_{AB}$ approximately contributes a factor of $\sqrt{A}$. Hence, the remaining terms in \eqref{eq:MS_F_expand} are $k/\sqrt{\mathcal{A}}$ suppressed relative to the leading term. This concludes the argument. 
\end{proof} 

On the basis of Claim A.1, \cite{murthy_2019_structure} concludes that the distribution $F(\omega)$ can be approximated by a Gaussian: 
\be  \label{a10}
F(\omega) \approx F_{\rm gauss}(\omega) = \frac{e^{- \omega^2/2\Delta^2}}{\sqrt{2\pi} \Delta} \, . 
\ee
It is not clear whether this conclusion can be drawn, for two reasons: firstly, note that \eqref{msf} is not exact. Indeed, the errors in \eqref{msf} are proportional to powers of $\sA$ and hence large, although the errors in even powers are small relative to the leading approximation. Moreover, even in the case where we do have exact matching of a certain number of moments between two distributions, this does not guarantee that the distributions are equal. Suppose we have a distribution $F_0$ such that the $n$-th moments of $F_0$ and $F_{\rm Gauss}$ are equal for all $n\leq p$. Then for $p\gg 1$ and $\omega \gtrsim 1$, their difference is bounded as
\cite{lindsay_moments}
\begin{equation} \label{a10_bound}
    |F_0(\omega) - F_{\rm Gauss}(\omega)| \leq \mathcal{O}(|\omega|^{-p}) \,.
\end{equation}
Motivated by the argument of \cite{murthy_2019_structure}, suppose we take $1 \ll p \lesssim \sqrt{\sA}$. We now consider the usefulness of this bound in different ranges of $\omega$: 
\begin{itemize}
\item For $\omega <1$, \eqref{a10_bound} does not apply, and the distribution is unconstrained. 
\item For $\omega >1$ and  $\mathcal{O}(1)$, 
$F_{\rm Gauss}(\omega) \sim 1/\sqrt{\mathcal{A}}$, while $|\omega|^{-p} \sim e^{- \# p}$. Since we can always take $p \gg \ln \mathcal{A}$, the bound gives a strong constraint, $|F_0(\omega) - F_{\rm Gauss}(\omega)| \ll F_{\rm Gauss}(\omega)$.  
\item For $1 \ll \omega \lesssim \mathcal{O}(\sqrt{p} \Delta)$, $F_{\rm Gauss}(\omega) \gtrsim \mathcal{O}(e^{-\# p})$, while  $ |\omega|^{-p} \sim \mathcal{O}(e^{-\# p \log \mathcal{A}})$, so again we get a strong constraint. 
\item For $\omega \gtrsim \mathcal{O}(\sqrt{p} \Delta \ln \Delta)$, $F_{\rm Gauss}(\omega) \ll |\omega|^{-p}$ and the constraint coming from moment-matching becomes  weak. 
\end{itemize} 
Therefore, we conclude that exact moment-matching up to $p = \mathcal{O}(\sqrt{\sA})$ would constrain the behavior of the distribution function $F_0(\omega)$ in the intermediate range $1 \lesssim \omega \ll \sA^{3/4}$.

Finally, we emphasize that when $\mathcal{A}$ is $\mathcal{O}(1)$, as is always the case in a (1+1)-dimensional system, the Gaussian ansatz does not hold in any regime, as already pointed out in \cite{murthy_2019_structure}. This was also numerically verified in Sec.~\ref{subsec:numerical_f_form}. 

\section{Justifying the approximation \texorpdfstring{$S(E) \approx S_{\rm fac}(E)$}{}}
\label{app:therm}

In this appendix, we justify the approximation that the thermodynamic entropy of $H$ is approximately equal to that of $H_A + H_B$, which we use in several arguments in the main text. We first recall a key result from~\cite{keating2015spectra}:
\begin{claim}\label{thm:spec_gauss_convergence}
    (Keating, Linden Wells 2014): Consider a sequence of local qudit Hamiltonians on a $d$-dimensional lattice with volume $V \sim L^d$
    \begin{equation}
        H(V) = \sum_{\langle i,j \rangle} \sum_{a=0}^3 \sum_{b = 0}^3 J_{a,b}\,\sigma^{(a)}_i \sigma^{(b)}_j \,,
    \end{equation}
    where $\sigma^{(a)}_j$ label onsite Pauli matrices, $\langle ij \rangle$ is a sum over nearest neighbor sites and $J_{a,b}$ are a set of finite real constants. Let $\Delta_V^2$ denote the variance of the Hamiltonian $H_V$. Then there is a positive constant $\sigma$ such that
    \begin{equation}
        \lim_{V \rightarrow \infty} \Delta_V^2/V = \sigma^2 \,,
    \end{equation}
    and the density of states $\tilde \rho(\epsilon, V)$ of $H(V)/(\sqrt{V} \sigma)$ converges weakly to the unit-variance Gaussian distribution. 
\end{claim}
Now let us apply this result to the Hamiltonians $H(V)$ and $H_A(V)+H_B(V)$ respectively
. For every lattice volume $V$, $H(V)$ and $H_A(V)+H_B(V)$ are both local Hamiltonians defined on the same Hilbert space. Moreover, since they only differ by a finite number of terms in the $V \rightarrow \infty$ limit, the difference between their variances is $\sO(1)$, although their variances are individually $\sO(V)$. Therefore, the normalized variance $\sigma^2$ appearing in Claim~\ref{thm:spec_gauss_convergence} is identical for $H(V)$ and $H_A(V)+H_B(V)$. This guarantees that the normalized density of states $\tilde \rho(\epsilon, V), \tilde \rho_{\rm fac}(\epsilon,V)$ associated with $H(V)/(\sigma \sqrt{V})$ and $\left[H_A(V)+H_B(V)\right]/(\sigma \sqrt{V})$ converge to the same Gaussian. 

In the main text, we are interested in the 
density of states $\rho(E) = e^{S(E)}$ of $H(V)$ and $\rho_{\rm fac}(E) = e^{S_{\rm fac}(E)}$ of $H_A(V)+H_B(V)$
, rather than 
$\tilde \rho$ and $\tilde \rho_{\rm fac}$. After rescaling by the variances, 
$S$ and $S_{\rm fac}$ tend to slightly different Gaussians
\begin{equation}
    S(E) = -\frac{E^2}{2 \sigma^2 V} \quad S_{\rm fac}(E) = -\frac{E^2}{2 (\sigma^2 + \delta/V) V}
\end{equation}
where $\delta$ quantifies the $\sO(1)$ difference between the variances of $H(V)$ and $H_A(V)+H_B(V)$. For $E \ll \sO(V)$, the difference between $S(E)$ and $S_{\rm fac}(E)$ can be shown to vanish in the thermodynamic limit
\begin{equation}
    S(E) - S_{\rm fac}(E) \approx - \frac{\delta}{2 \sigma^4} \frac{E^2}{V^2} \rightarrow 0 \,. 
\end{equation}
However, when $E = \epsilon V$ for some $\epsilon = \mathcal{O}(1)$, the difference $S(E) - S_{\rm fac}(E)$ becomes $\sO(1)$. As a result, the ratio of $\rho(E)$ and $\rho_{\rm fac}(E)$ at a finite energy density $E = \epsilon V$ approaches a finite ratio in the thermodynamic
\begin{equation}
    e^{S_{\rm fac}(E)}/e^{S(E)} \approx e^{\frac{\epsilon^2\delta}{\Delta^2}} \,.
\end{equation}
This $\mathcal{O}(1)$ multiplicative correction does not change any of the qualitative conclusions we drew in the main text about the functional form of $f(\omega)$ and its dynamical consequences.

\section{Alternative approximations for   the EDF}\label{app:EDF_derivation_details}

Recall that in Sec. \ref{subsec:derivation_f_form}, we found the exact characteristic equations 
\begin{equation}\label{char_1}
    E_a - \bar E_{ij} = \frac{1}{\pi} \sum_m \frac{\Gamma_{ij}(\epsilon_m) \rho(\epsilon_m)^{-1}}{E_a - \epsilon_m}\,, \quad |c^a_{ij}|^{-2} \rho(E_a)^{-1} \approx \frac{1}{\pi \rho(E_a)}\sum_{m} \frac{\Gamma_{ij}(\epsilon_m) \rho(\epsilon_m)^{-1}}{(E_a - \epsilon_m)^2}  \,.
\end{equation}
and evaluated the sums using the approximation that the energy levels are exactly evenly spaced. In this appendix, we do not make this assumption, but give an alternative argument for the Lorentzian regime, 
\be \label{lor}
   f(E_a-\bar E_{ij}) = \frac{1}{\pi} \frac{\Gamma_{ij}(\bar E_{ij})}{(E_a - \bar E_{ij})^2 + \Gamma_{ij}(\bar E_{ij})^2} \,,  
\ee
using assumptions about the functional form of $\Gamma_{ij}(\epsilon)$.  

Before we state and use these assumptions, we first perform some exact manipulations that are independent of the form of $\Gamma_{ij}(\epsilon)$.
We start by squaring both sides of the first equation in \eqref{eq:3.4heuristic_3}:
\begin{equation}
    \begin{aligned}
    (E_a - \bar E_{ij})^2 &= \frac{1}{\pi^2} \sum_{m,n} \frac{\Gamma_{ij}(\epsilon_m) \rho(\epsilon_m)^{-1} \Gamma_{ij}(\epsilon_n) \rho(\epsilon_n)^{-1}}{(E_a - \epsilon_m)(E_a - \epsilon_n)} \\
    &= \frac{1}{\pi^2} \sum_{m} \frac{\Gamma_{ij}(\epsilon_m)^2 \rho(\epsilon_m)^{-2}}{(E_a - \epsilon_m)^2} + \frac{1}{\pi^2} \sum_{m \neq n} \frac{\Gamma_{ij}(\epsilon_m) \rho(\epsilon_m)^{-1}\Gamma_{ij}(\epsilon_n) \rho(\epsilon_n)^{-1}}{(E_a - \epsilon_m)(E_a - \epsilon_n)}  \,.
    \end{aligned}
\end{equation}
For the first term, since the sum is sharply peaked at $\epsilon_m \approx E_a$, we can pull out a factor of $\Gamma_{ij}(E_a)/\rho(E_a)$ and relate it to the EDF:
\begin{equation}
    \text{Term 1} = \frac{\Gamma_{ij}(E_a)}{\pi} |c^a_{ij}|^{-2} \rho(E_a)^{-1} \approx \frac{\Gamma_{ij}(E_a)}{\pi} f(E_a - \bar E_{ij})^{-1} \,. 
\end{equation}
The second term can be simplified using partial fractions
\begin{equation}
    \begin{aligned}
    \text{Term 2} &= \frac{1}{\pi^2} \sum_{m \neq n} 
    \frac{\Gamma_{ij}(\epsilon_m) \rho(\epsilon_m)^{-1} \Gamma_{ij}(\epsilon_n) \rho(\epsilon_n)^{-1}}{\epsilon_m - \epsilon_n} \left[\frac{1}{E_a - \epsilon_m} - \frac{1}{E_a - \epsilon_n}\right] \\
    &= \frac{2}{\pi^2} \sum_{m\neq n} \frac{\Gamma_{ij}(\epsilon_m) \rho(\epsilon_m)^{-1} \Gamma_{ij}(\epsilon_n) \rho(\epsilon_n)^{-1}}{(\epsilon_m - \epsilon_n)(E_a - \epsilon_m)} = \frac{2}{\pi} \sum_m g_{ij}(\epsilon_m)  \frac{\Gamma_{ij}(\epsilon_m) \rho(\epsilon_m)^{-1}}{E_a - \epsilon_m} 
    \end{aligned}
\end{equation}
where $g_{ij}(\epsilon_m)$ is the Hilbert transform of $\Gamma_{ij}(\epsilon)$
\begin{equation}
    g_{ij}(\epsilon_m) = \frac{1}{\pi} \sum_{n,\,  n \neq m} \frac{\Gamma_{ij}(\epsilon_n) \rho(\epsilon_n)^{-1}}{\epsilon_m - \epsilon_n} = \frac{1}{\pi} P \int \frac{\Gamma_{ij}(\epsilon)}{\epsilon_m - \epsilon} d \epsilon \,.
\end{equation}
Since we expect $g_{ij}(\epsilon_m)$ to be a smooth function, we can expand $g_{ij}(\epsilon_m)$ around $\epsilon_m = E_a$. Assuming that $\Gamma_{ij}(\epsilon_m)$ is even in $\epsilon_m - \bar E_{ij}$, we obtain a further simplification
\begin{equation}\label{eq:term2_derivative_expansion}
    \begin{aligned}
        \text{Term 2} &= \frac{2}{\pi} \sum_m \left[g_{ij}(\epsilon_m) - g_{ij}(E_a)\right] \frac{\Gamma_{ij}(\epsilon_m) \rho(\epsilon_m)^{-1}}{E_a - \epsilon_m} + \frac{2 g_{ij}(E_a)}{\pi} \sum_m \frac{\Gamma_{ij}(\epsilon_m) \rho(\epsilon_m)^{-1}}{E_a - \epsilon_m}\\
        &= -\frac{2}{\pi} \sum_m \sum_{p=1}^{\infty} \frac{g_{ij}^{(p)}(E_a)}{p!} \frac{(\epsilon_m - E_a)^p}{\epsilon_m - E_a} \Gamma_{ij}(\epsilon_m) \rho(\epsilon_m)^{-1} + 2 g_{ij}(E_a) (E_a - \bar E_{ij}) \\
        &= - \frac{2}{\pi} \sum_{p=0}^{\infty} \frac{g_{ij}^{(1+p)}(E_a)}{(1+p)!} \sum_m \Gamma_{ij}(\epsilon_m) \rho(\epsilon_m)^{-1} (\epsilon_m - E_a)^{p} + 2 g_{ij}(E_a) (E_a - \bar E_{ij}) \,, 
    \end{aligned}
\end{equation}
where $g_{ij}^{(p)}(\epsilon)$ is the $p$-th derivative of $g_{ij}(\epsilon)$. In this new representation of Term 2, the summand is completely non-singular. Therefore, all sums can be approximated by integrals. After further expanding in powers of $\delta = E_a - \bar E_{ij}$, we find
\begin{equation}\label{delta_exp}
    \begin{aligned}
    \text{Term 2} &= - \frac{2}{\pi} \int d \epsilon \Gamma_{ij}(\epsilon) \sum_{p=0}^{\infty} \frac{g_{ij}^{(1+2p)}(\bar E_{ij} + \delta)}{(1+2p)!} (\epsilon_m - \bar E_{ij} - \delta)^{2p} + 2 g_{ij}(\bar E_{ij} + \delta) \delta \\
    &= - \frac{2}{\pi} \int d \epsilon \Gamma_{ij}(\epsilon) \sum_{p=0}^{\infty} \sum_{k=0}^{\infty} \sum_{q=0}^{p} \frac{g_{ij}^{(1+2p+2k)}(\bar E_{ij})}{(1+2p)!} \frac{\delta^{2k}}{(2k)!} (\epsilon_m - \bar E_{ij})^{2q} (- \delta)^{2p-2q} {2p \choose 2q} \\
    &\hspace{240pt}+ 2 \sum_{s=1}^{\infty} \frac{g_{ij}^{(s)}(\bar E_{ij})}{s!} \delta^{1+s}\\
    & = \sum_{n=0}^{\infty} C_{2n} \delta^{2n} \,,
    \end{aligned}
\end{equation}
where $C_{2n}$ are a set of real coefficients that can in principle be calculated. Putting these together, we find
\begin{equation}
\label{c7}
    f(E_a - \bar E_{ij}) \approx \frac{1}{\pi} \frac{\Gamma_{ij}(E_a)}{(E_a - \bar E_{ij})^2 - \sum_{n=0}^{\infty} C_{2n} (E_a - \bar E_{ij})^{2n}} \,. 
\end{equation}
As we will show in the rest of the section, $C_0 < 0$. Therefore, the Lorentzian prediction for $f(\omega)$ is robust for generic functional forms of $\Gamma_{ij}(\epsilon)$, as long as $\Gamma_{ij}$ varies slowly for some range $\sigma$ larger than the characteristic width of $f$. Corrections coming from $C_{2n}$ with $n > 0$ distort and Lorentzian and lead to additional features observed in Fig.~\ref{fig:f_form}. 

Now let us explicitly evaluate $C_0$ using two different assumptions for the functional form of $\Gamma_{ij}(\epsilon)$.
Let us first consider the simplest box approximation,
\be 
\Gamma_{ij}(\epsilon) = \begin{cases}
\frac{\pi \sigma_{E,ij}^2}{2 \sigma} & |\epsilon - \bar E_{ij}| \leq \sigma \\
0 & |\epsilon - \bar E_{ij}| >  \sigma 
\end{cases} \,,
\ee
which satisfies the exact normalization constraint 
\begin{equation} \label{gnorm}
    \int \Gamma_{ij}(\epsilon) d \epsilon = \pi \sum_m |V_{m,ij}|^2 = \pi \left(\bra{ij} H^2 \ket{ij} - \bra{ij} H \ket{ij}^2\right) = \pi \sigma_{E, ij}^2\,.
\end{equation}
The Hilbert transform then evaluates to
\begin{equation}
    g_{ij}(\epsilon_m) = \frac{1}{2 \sigma \pi} P \int_{\bar E_{ij} - \sigma}^{\bar E_{ij} + \sigma} \frac{\pi \sigma_{E,ij}^2}{\epsilon_m - \epsilon} = \frac{\sigma_{E,ij}^2}{2\sigma} \log \left(\frac{\sigma + \epsilon_m - \bar E_{ij}}{\sigma - \epsilon_m + \bar E_{ij}}\right) \,.
\end{equation}
By explicit computation, we can show that
\begin{equation}
    g_{ij}^{(1+2p)}(\epsilon_m) = \frac{\Gamma_{ij}(\bar E_{ij})}{\pi} (2p)! \left[(\sigma + \bar E_{ij} - \epsilon_m)^{-(1+2p)} + (\sigma - \bar E_{ij} + \epsilon_m)^{-(1+2p)} \right] \,,
\end{equation}
\begin{equation}
    \int d \epsilon \Gamma_{ij}(\epsilon) (\epsilon - \bar E_{ij})^{2p} = \Gamma_{ij}(\bar E_{ij}) \frac{2}{1+2p} \sigma^{1+2p} \,. 
\end{equation}
Plugging these results back into Term 2, we find that for $E_a = \bar E_{ij}$, 
\begin{equation}
    \begin{aligned}
        C_0 &= - \frac{2}{\pi} \sum_{p=0}^{\infty} \frac{2 \Gamma_{ij}(\bar E_{ij}) (2p)!}{\pi (1+2p)!} \sigma^{-(1+2p)} \Gamma_{ij}(\bar E_{ij}) \frac{2}{1+2p} \sigma^{1+2p}  \\
        &= - \Gamma_{ij}(\bar E_{ij})^2 \frac{8}{\pi^2} \sum_{p=0}^{\infty} \frac{1}{(1+2p)^2} = - \Gamma_{ij}(\bar E_{ij})^2 \,. 
    \end{aligned}
\end{equation}
Assuming that the quadratic correction to Term 2 is suppressed (as we confirm below), this gives the approximation \eqref{lor} for $f(\omega)$.


Next, we assume a more realistic situation where $\Gamma_{ij}(\epsilon)$ is an approximate Lorentzian with width $\sigma$ and decays rapidly for $|\epsilon - \bar E_{ij}| \gg \sigma$. Using the normalization constraint again, we have:
\be \label{c8}
 \Gamma_{ij}(\epsilon) \approx \frac{\sigma_{E,ij}^2 \sigma}{\sigma^2 + (\epsilon - \bar E_{ij})^2} \,.
\ee
Up to small errors due to the deviation of $\Gamma_{ij}(\epsilon)$ from a Lorentzian at large $\epsilon - \bar E_{ij}$, the Hilbert transform takes the form 
\be 
g_{ij}(\epsilon) = \frac{\sigma_{E,ij}^2(\epsilon - \bar E_{ij})}{\sigma^2 + (\epsilon - \bar E_{ij})^2} \,.\label{gij}
\ee 
Putting \eqref{gij} into Term 2, and evaluating at $E_a = \bar E_{ij}$, we find
\begin{equation} 
    \begin{aligned}
    C_0 &= \frac{2}{\pi} \sum_m \frac{ \sigma_{E,ij}^2(\epsilon_m - \bar E_{ij})}{\sigma^2 + (\epsilon_m - \bar E_{ij})^2} \frac{\Gamma_{ij}(\epsilon_m) \rho(\epsilon_m)^{-1}}{\bar E_{ij} - \epsilon_m} \\
    &= - \frac{2 \sigma_{E,ij}^2}{\pi} \sum_m \frac{\Gamma_{ij}(\epsilon_m) \rho(\epsilon_m)^{-1}}{\sigma^2 + (\epsilon_m - \bar E_{ij})^2} = - \Gamma_{ij}(\bar E_{ij})^2 \,,
    \end{aligned}
    \label{c11}
\end{equation}
where we approximated the last sum with an integral because the summand is non-singular. The width again matches \eqref{eq:3.4_f_form_final}. 

For both the box approximation and the Lorentzian approximation, we should also explicitly check that the corrections to $\text{Term 2}$ in \eqref{delta_exp} for $E_a \neq \bar E_{ij}$ are suppressed for small $|E_a - \bar E_{ij}|$. 
From this expansion, it is clear that the leading corrections are quadratic in $\delta$. To extract the quadratic coefficient $C_2$, we need only evaluate the terms with $k = 0, p-q=1$ and $k=1, p=q$ in the first sum and the terms with $s = 1$ in the second sum. Using either the box approximation or the Lorentzian approximation for $\Gamma_{ij}(\epsilon)$ (we omit the details which are tedious but straightforward), we find that
\begin{equation}
    \text{Term 2} \approx - \left[\Gamma_{ij}(\bar E_{ij})^2 - C_2 (E_a - \bar E_{ij})^2 + \mathcal{O}\left((E_a - \bar E_{ij})^4\right)\right] \,. 
\end{equation}
where $C_2 \lesssim \mathcal{O}(\frac{\Gamma_{ij}(\bar E_{ij})}{\sigma})$. The upper bound guarantees that the Lorentzian holds over a much larger range of $|E_a - \bar E_{ij}|$ than the width $\Gamma_{ij}(\bar E_{ij})$. 

Finally, we note in passing in passing that we could have obtained the above result by starting with a slightly weaker assumption about the analytic structure of $\Gamma_{ij}(\epsilon)$. We can start by requiring that $\Gamma_{ij}(\epsilon)$ is an even analytic function of $\epsilon- \bar E_{ij}$ whose only singularity in the upper half plane is a pole at $\epsilon_* - \bar E_{ij} = i \sigma$.  Consider the Hilbert transform
\begin{equation}
    g_{ij}(\epsilon_m) = \frac{1}{\pi} P \int \frac{\Gamma_{ij}(\epsilon)}{\epsilon_m - \epsilon} d \epsilon \,. 
\end{equation}
By completing the principal value integral to an integral from $-\infty + i 0^+$ to $\infty + i 0^+$ and then closing the integration contour in the upper half plane, we pick up two residues at $\epsilon = \epsilon_m$ and $\epsilon = \bar E_{ij} + i \sigma$. If the residue at $\epsilon = \bar E_{ij} + i \sigma$ is $R_1 + i R_2$, then 
\begin{equation}
    g_{ij}(\epsilon_m) = -i \Gamma_{ij}(\epsilon_m) + \frac{2i (R_1 + i R_2)}{\epsilon_m - \bar E_{ij} - i \sigma} = -i \Gamma_{ij}(\epsilon_m) + 2 (R_1 + i R_2) \frac{-\sigma + i (\epsilon_m - \bar E_{ij})}{\sigma^2 + (\epsilon_m - \bar E_{ij})^2} \,. 
\end{equation}
By definition, $g_{ij}(\epsilon_m)$ is a real function for real $\epsilon_m$. Thus, the above equation leads to the constraints
\begin{equation}
    \Gamma_{ij}(\epsilon_m) = \frac{2 R_1 (\epsilon_m - \bar E_{ij}) - 2R_2 \sigma}{\sigma^2+(\epsilon_m-\bar E_{ij})^2} \,, \quad  g_{ij}(\epsilon_m) = -2 \frac{R_1 \sigma + R_2 (\epsilon_m - \bar E_{ij})}{\sigma^2 + (\epsilon_m - \bar E_{ij})^2} \,.
\end{equation}
Since $\Gamma_{ij}(\epsilon_m)$ is real, manifestly positive, and even for $\epsilon_m \in \mathbb{R}$, we must have $R_1 = 0$ and $R_2 < 0$.

Putting all of these ingredients together, and using the normalization condition \eqref{gnorm}, we conclude that
\begin{equation}
    \Gamma_{ij}(\epsilon) = \frac{\sigma  \sigma_{E,ij}^2}{\sigma^2 + (\epsilon - \bar E_{ij})^2} \,,\quad g_{ij}(\epsilon) = \frac{\sigma_{E,ij}^2(\epsilon - \bar E_{ij})}{\sigma^2 + (\epsilon - \bar E_{ij})^2} \,
\end{equation}
which was precisely our assumption in \eqref{c8} and \eqref{gij}.

\end{appendix}

\bibliographystyle{jhep.bst}
\bibliography{entropy.bib}

\end{document}